\documentclass[letterpaper, 12pt]{article}

\usepackage[utf8]{inputenc}
 
\usepackage{geometry}
\usepackage{fancyhdr}
\usepackage{amsmath,amsfonts,amssymb,amsthm}  
\usepackage[inline]{enumitem}
\usepackage{dsfont}
\usepackage{listings}
\usepackage{color}
\usepackage{longtable}
\usepackage{array}
\usepackage{float}
\usepackage{booktabs}
\usepackage{graphicx}
\usepackage{subcaption}
\usepackage[english]{babel}
\usepackage{commath}
\usepackage{color}
\usepackage{multirow}
\usepackage{bm}
\usepackage{tikz}

\usepackage{setspace}
\usepackage{anyfontsize}
\usepackage{dsfont}
\usepackage{mathtools}	
\usepackage{hyperref}
\usepackage{listings}
\usepackage{natbib}
\usepackage{hhline}
\usepackage{authblk}
\usepackage{verbatim}
\usepackage{dcolumn}
\usepackage{calc}

\usepackage[flushleft]{threeparttable}

\newtheorem{proposition}{Proposition} 
\newtheorem{corollary}{Corollary}
\newtheorem{theorem}{Theorem}

\newtheorem{lemma}{Lemma}
\newtheorem{assumption}{Assumption}
\theoremstyle{definition}

\newtheorem{remark}{Remark}

\newenvironment{assbis}[1]
{\renewcommand{\theassumption}{\ref*{#1}*}%
	\addtocounter{assumption}{-1}%
	\begin{assumption}}
	{\end{assumption}}

\usepackage{chngcntr}
\usepackage{apptools}
\AtAppendix{\counterwithin{lemma}{section}}
\AtAppendix{\counterwithin{theorem}{section}}
\AtAppendix{\counterwithin{proposition}{section}}
\AtAppendix{\counterwithin{equation}{section}}
\AtAppendix{\counterwithin{remark}{section}}
\AtAppendix{\counterwithin{table}{section}}
\AtAppendix{\counterwithin{assumption}{subsection}}
\AtAppendix{\counterwithin{corollary}{section}}

\newcommand{\E}{\mathbb{E}}	
\newcommand{\Prob}{\mathbb{P}}
\newcommand{\Var}{\mathbb{V}}

\newcommand{\Cov}{\textrm{Cov}}

\newcommand{\wh}{\widehat}
\newcommand{\mc}{\mathcal}
\newcommand{\wt}{\widetilde}

\DeclareMathOperator*{\argmin}{arg\,min}

\newcommand{\one}[1]{\mathds{1}\{#1\}}

\newcommand\khi{ k_{h,i} }
\newcommand\kX{ k(X_h) }
\newcommand\kXi{ k(X_{h,i}) }
\newcommand\kh{ k_h(X_i-x_0) }
\newcommand\khX{ k_h(X-x_0) }
\newcommand\Xihj{ X_{h,i}^j}

\newcommand{\sss}{\scriptscriptstyle}
\newcommand{\ccdot}{\,\cdot\,}

\newcommand{\Sum}{ \sum_{i=1}^n }

\hypersetup{
	colorlinks   = true, 
	urlcolor     = blue,
	linkcolor    = black, 
	citecolor   =  black 
}

\newcommand{\ThanksText}{I am grateful to my Ph.D. advisor Christoph Rothe for his invaluable support. I thank François Gerard for kindly running my estimation routine on a restricted-use dataset. I thank Timo Dimitriadis, Claudia Noack, Yoshiyasu Rai, and participants of IAAE Conference 2021, EEA-ESEM Virtual 2021, Econometrics Seminar in Mannheim, HKMetrics-Workshop, and Bonn-Mannheim PhD Workshop for their helpful comments. I gratefully acknowledge funding by the German Research Foundation (DFG) through CRC TR 224 (Project A04)  and by the European Research Council (ERC) through grant SH1-77202. }

\author{Tomasz Olma\thanks{\ThanksText
\textit{Address:} University of Mannheim, Department of Economics, L7, 3--5; 68161 Mannheim, Germany. 
\textit{Email:} \href{mailto:tomekolma@gmail.com}{tomasz.j.olma@gmail.com}.
\textit{Website:} \href{ https://tomaszolma.github.io}{tomaszolma.github.io}.}
} 	

\affil{University of Mannheim}

\title{Nonparametric Estimation of Truncated Conditional Expectation Functions}

\date{\today}

\begin{document}
\maketitle
\begin{abstract}
	Truncated conditional expectation functions are objects of interest in a wide range of economic applications, including income inequality measurement, financial risk management, and impact evaluation. They typically involve truncating the outcome variable above or below certain quantiles of its conditional distribution.
	In this paper, based on local linear methods, a novel, two-stage, nonparametric estimator of such functions is proposed.
	In this estimation problem, the conditional quantile function is a nuisance parameter that has to be estimated in the first stage. The proposed estimator is insensitive to the first-stage estimation error owing to the use of a Neyman-orthogonal moment in the second stage.
	This construction ensures that inference methods developed for the standard nonparametric regression can be readily adapted to conduct inference on truncated conditional expectations.
	As an extension, estimation with an estimated truncation quantile level is considered.
	The proposed estimator is applied in two empirical settings: sharp regression discontinuity designs with a~manipulated running variable and randomized experiments with sample selection.
\end{abstract}


\newpage
\section{Introduction}\label{sec:Introduction}
A truncated sample mean is the mean calculated after discarding some of the highest and/or lowest values in a sample.
Such quantities, which estimate the corresponding truncated expectations, are used in a wide range of economic applications.
In studies of inequality, income dispersion can be summarized by reporting the mean income in different quintiles of its distribution, i.e., the mean income of the 20\% of households with the lowest income, followed by the mean income of households between the 20th and 40th percentile of the income distribution, etc.\ \citep[e.g.,][]{Semega2020}. 
In finance, the expected shortfall denotes the expected value of a certain proportion, e.g.\ 5\%, of top losses.
It is a widely-used risk measure, which informs about the performance of a portfolio of assets in the worst-case scenarios \citep[e.g.,][]{Chen2008}. 
Truncated means are also used in settings with contaminated data, where the sharp bounds on the true expected outcome are obtained by considering the extreme scenarios in which the contaminated data points have the highest or the lowest outcomes, and by trimming the respective tails of the outcome distribution \citep{horowitz1995identification}. 
This partial identification approach  has been adapted to several impact evaluation settings to address sample selection problems; see, e.g., \citet{ZhangRubin2003, Lee2009, chen2015bounds}.

In all the above examples, the analysis can be enriched by incorporating covariates.
First, the anatomy of income inequality can be better understood when analyzed conditionally on characteristics such as age or work experience. Second, an estimator of the expected shortfall can be more informative if it takes into account covariates, such as past returns. Third, in impact evaluation, the heterogeneity of treatment effects can be explored based on individuals' characteristics.
Furthermore, \citet{gerard2020bounds} apply the trimming approach of \citet{horowitz1995identification} to regression discontinuity designs with a manipulated running variable, which necessarily involve conditioning on a covariate. 

In this paper, I propose a novel, nonparametric estimator of truncated conditional expectation functions. As in the above-mentioned applications, I consider setups where the outcome variable needs to be truncated above or below certain quantiles of its conditional distribution. 
For ease of exposition, I focus on one-sided truncation.
I consider a nonparametric setting with a continuous outcome variable, denoted by $Y$, and a vector of continuous covariates, denoted by $X$.\footnote{If the covariates take on only a small number of distinct values, then the truncated conditional expectation function can be estimated using sample truncated means binned by covariate values.} For a quantile level $\eta \in (0,1)$ and $x$ in the support of $X$, let $Q(\eta,x)$ be the conditional $\eta$-quantile of $Y$ given $X=x$. 
The object of interest is the following function:
\begin{align}\label{eq:meta}
m(\eta,x) & =\E[Y|Y \leq Q(\eta,X), X=x].
\end{align}
I refer to $\eta$ in the above definition as the truncation quantile level.
It might be chosen by the analyst, in which case it is a fixed, known number, but in some applications the truncation quantile level has to be estimated from the data. The considered setting is nonparametric, meaning that only smoothness restrictions on the functions $m(\eta,x)$ and $Q(\eta,x)$ are imposed.

In this estimation problem, the function $Q(\eta,\ccdot)$ is a nuisance parameter. If it was known, then based on a sample $\{(X_i,Y_i)\}_{i=1}^n$ from the distribution of $(X,Y)$, one could estimate $m(\eta,x)$ using standard nonparametric regression techniques, e.g., kernel estimators, applied to the sample restricted to observations with $Y_i \leq Q(\eta,X_i)$. Alternatively, motivated by the equivalent representation of the estimand as:
\begin{equation}\label{eq:def2}
m(\eta,x) =\frac{1}{\eta} \E[Y \one{Y \leq Q(\eta,X)}|X=x],
\end{equation}
one could run a nonparametric regression with $\frac{1}{\eta}Y_i \one{Y_i \leq Q(\eta,X_i)}$ as the outcome variable. Feasible versions of these two estimators, however, require estimating the function $Q(\eta, \ccdot)$ in the first stage. This additional estimation step may affect the properties of the resulting estimators in a potentially complicated manner.

In order to alleviate the impact of the first-stage estimation error on the final estimator, I~propose a modification of the latter approach using a conditional moment equation that is Neyman-orthogonal to the conditional quantile function \citep{neyman1979c}. Specifically, the proposed estimation approach is based on the following representation of the estimand:
\begin{equation}\label{eq:def3}
m(\eta,x)= \frac{1}{\eta} \E[ Y\one{Y \leq Q(\eta,X)} - Q(\eta,X)(\one{Y \leq Q(\eta,X)}-\eta)|X=x].
\end{equation}
Compared to \eqref{eq:def2}, the conditional moment in \eqref{eq:def3} contains an additional term that is mean-zero conditional on $X$.\footnote{The conditional moment in \eqref{eq:def3} is the quantity of interest when the outcome variable has mass points, but, as shown in this paper, there are reasons to consider this formula even with a continuous outcome variable.} 
Its inclusion renders the whole expression insensitive to small perturbations of $Q(\eta,\ccdot)$ in the sense that its derivative with respect to the the conditional quantile evaluated at the truth equals zero,
\begin{equation}\label{eq:orthogonality}
\frac{\partial}{\partial q} \E[Y\one{Y \leq q} - q(\one{Y \leq q}-\eta)|X=x] |_{q=Q(\eta,x)}=0.
\end{equation}
Such orthogonal, or locally-robust, conditional moments feature prominently in the literature in setups where a nuisance parameter has to be estimated in the first stage \citep[e.g.,][]{belloni2017program, chernozhukov2018double}. The orthogonality property ensures that estimation of the nuisance parameter has no first-order effect on the asymptotic distribution of the final estimator.

Based on the conditional moment in equation \eqref{eq:def3}, my proposed estimator is constructed in two steps using local linear methods \citep{fan1996local}. In the first stage, I estimate the local linear approximation of the function $Q(\eta,\ccdot)$. In the second stage, I run a local linear regression with a generated outcome variable corresponding to the expression under the conditional expectation in equation \eqref{eq:def3}.\footnote{Based on the local linear methods, one can also construct estimators of truncated conditional expectations functions motivated by the expressions in \eqref{eq:meta} and \eqref{eq:def2}. I discuss them in Online Appendix \ref{A:Alternative}.} The estimator is easy to implement, and the bandwidths for the two local linear regressions can be selected as in standard nonparametric regressions.

This paper contains two main theoretical results. First, I show that the proposed estimator is asymptotically equivalent to its infeasible analog using the true conditional quantile function. Given this result, the asymptotic properties follow from the standard theory of local linear estimation. The proposed estimator has good bias properties, and it is straightforward to adapt existing inference methods to do inference on truncated conditional expectation functions.
Second, I study the asymptotic properties of my estimator when the truncation quantile level is estimated from the data.
Under a high-level assumption on $\widehat{\eta}$, I derive an expansion of the proposed estimator evaluated at $\widehat{\eta}$ about the estimator evaluated at the true value $\eta$. This expansion can be used on a case-by-case basis to derive the asymptotic distribution for specific estimators $\wh\eta$.

I apply the proposed estimator in two empirical settings. First, I estimate bounds on the local average treatment effect in regression discontinuity designs with a manipulated running variable \citep{gerard2020bounds}. Second, I estimate bounds on the conditional wage effect of a job training program \citep{Lee2009}. These bounds involve truncated conditional expectation functions with truncation quantile levels that need to be estimated from the data.

\paragraph*{Related Literature.} Nonparametric estimation of truncated conditional expectation functions has been extensively studied in the context of the conditional expected shortfall estimation. \citet{scaillet2005nonparametric}, \citet{Cai2008}, and \citet{Kato2012} propose estimators based on first-stage estimates of the conditional cumulative distribution function (c.d.f.) of the outcome variable. This estimation strategy, however, is not well-suited for estimation at points on the boundary of the support of the conditioning variables. 
The Nadaraya-Watson estimator of the conditional c.d.f., employed by \citet{scaillet2005nonparametric}, exhibits the so-called boundary effects in that its bias is of larger order at the boundary than in the interior.\footnote{Estimation of a conditional c.d.f. can be cast as a regression problem with outcome variable $\one{Y_i \leq y}$.} \citet{Cai2008} and \citet{Kato2012} use the weighted Nadaraya-Watson estimator, which is asymptotically equivalent to the local linear estimator at interior points, but, unlike the local linear estimator, it is guaranteed to yield a proper c.d.f. The weighted Nadaraya-Watson estimator, however, is not defined for boundary points. In contrast, my proposed approach is well-suited for estimation at boundary points.

\citet{linton2013estimation} propose an estimator based on the orthogonal conditional moment equation in \eqref{eq:def3}. Their analysis, however, applies specifically to setups where the conditional variance of the outcome variable is infinite, which results in the first-stage local polynomial quantile estimator converging faster than the final estimator. The proof of \citet{linton2013estimation} does not apply to models with finite variance of the outcome variable considered in this paper, where the first stage and the final estimator have the same rates of convergence. Their estimator is also more computationally intensive as it requires estimating a separate local polynomial quantile regression for each observation used in the second stage.

Various ways of estimating truncated conditional expectation functions have also been proposed in parametric settings. \citet{koenker1978regression},  \citet{Ruppert1980}, and \citet{jurevckova1984regression} consider generalizations of truncated means to linear models. In the first stage, they estimate quantile regressions, and in the second stage they run a regression on a sample truncated according to the first-stage estimates. In an independent work, \citet{Barendse2020} also uses a generated outcome variable based on the orthogonal moment equation. He additionally considers efficient weighting, analogous to, possibly nonlinear, weighted least squares. \citet{dimitriadis2019joint} develop a joint quantile and expected shortfall estimation framework and find estimators that can be more efficient than the simple two-stage procedure described above. The efficiency gains of \citet{dimitriadis2019joint} and \citet{Barendse2020}, however, are specific to parametric models.

In the above-cited papers, it is assumed that the truncation quantile level is chosen by the analyst. A setting with estimated conditional truncation quantile levels and possibly continuous covariates is studied by \citet{semenova2020better}. 
She exploits a moment that is similar to \eqref{eq:def3}, but it includes additional terms that render the expression orthogonal also to the truncation quantile level.\footnote{This property is achieved using a specific conditional moment defining the truncation quantile level.}
Her focus, however, is on integrated truncated conditional expectations, and she does not provide conditional estimates.
Truncated means with estimated trimming proportions have also been studied in the unconditional case, e.g., by \citet{shorack1974random} and \citet{Lee2009}.

\paragraph*{Outline of the Paper.}
The remainder of this paper is structured as follows. In Section~\ref{sec:TLLR}, I~formally introduce the proposed estimator. Its asymptotic properties are studied in Section~\ref{sec:Asymptotics}. In Section~\ref{sec:Inference}, I discuss inference. I present a Monte Carlo study in Section~\ref{sec:MonteCarlo}. In Section~\ref{sec:Applications}, I~consider two empirical applications: (i) sharp regression discontinuity designs with a manipulated running variable and (ii) randomized experiments with sample selection. Section \ref{sec1:Conclusions} concludes.

\section{Estimator}\label{sec:TLLR}
In this section, I formally introduce the proposed estimator. To simplify the exposition, $X$ is assumed to be univariate. A natural extension for the multivariate case is presented in Online Appendix~\ref{A:multivariate}. I consider estimation at a selected covariate value $x_0$. The truncation quantile level, in turn, might be known or might have to be estimated from the data.

In the first stage,  I estimate the conditional $\eta$-quantile function $Q(\eta,\ccdot)$. For the second-stage estimator, it suffices if $Q(\eta,\ccdot)$ is estimated well for covariate values close to $x_0$.  The level and slope of the function $Q(\eta,\ccdot)$ at $x_0$ are estimated in a local linear quantile regression as
\begin{equation*}\label{eq:TLLRq_est}
(\widehat{q}_{0}(\eta, x_0; a), \widehat{q}_1(\eta, x_0;a))^\top = \argmin_{(\beta_0,\beta_1)} \Sum k_{a}(X_i-x_0) \rho_\eta(Y_i-\beta_0-\beta_1(X_i-x_0)),
\end{equation*}
where $\rho_\eta(v)=v(\eta - \one{v \leq 0})$ is the `check' function, $k(\cdot)$ is a kernel function, $a$ is a~bandwidth, and $k_a(v)=k(v/a)/a$. Based on these estimates, for $x$ in the estimation window relevant for the second stage,  $Q(\eta,x)$ is estimated with its implied local linear approximation:
\begin{equation*}\label{eq:TLLRq_fun}
\widehat{Q}^{ll}(\eta, x;x_0,a)= \widehat{q}_{0}(\eta, x_0;a) +  \widehat{q}_{1}(\eta, x_0;a) (x-x_0).
\end{equation*}

In the second stage, I run a local linear regression with a generated outcome variable corresponding to the expression in \eqref{eq:def3} to estimate the truncated conditional expectation $m(\eta,x_0)$:
\begin{equation*}
\widehat{m}(\eta, x_0; a, h) = e_1^\top \argmin_{ (\beta_0,\beta_1) } \Sum k_h(X_i-x_0)\big( \psi_i(\eta, \widehat{Q}^{ll}(\eta,X_i;x_0,a)) -\beta_0 - \beta_1(X_i-x_0)\big)^2,\label{eq:TLLRm_est}
\end{equation*}
where $e_1=(0,1)^\top$, $h$ is another bandwidth, and  
$$
\psi_i(\eta,q )= \frac{1}{\eta} \left( Y_i\one{ Y_i \leq q } -  q(\one{Y_i \leq q}-\eta)\right).
$$
If $\eta$ is not known, but an estimate $\wh \eta$ is available, I estimate $m(\eta,x_0)$ as $\widehat{m}(\wh\eta, x_0; a, h)$.

\section{Asymptotic Properties}\label{sec:Asymptotics}
In this section, I introduce the assumptions and study the asymptotic properties of the proposed estimator. I use the following notation. I put $\partial_x^k m(\eta,x_0)= \frac{\partial^k}{\partial x^k}m(\eta,x)|_{x=x_0}$ and $\partial_x^k Q(\eta,x_0)= \frac{\partial^k}{\partial x^k}Q(\eta,x)|_{x=x_0}$. For positive sequences $b_n$ and $c_n$, I write $b_n \prec c_n$ if $b_n / c_n \to 0$, and $b_n \asymp c_n$ if $C_1 b_n \leq c_n \leq C_2 b_n$ for some positive constants $C_1$ and $C_2.$

\subsection{Assumptions}\label{subsec:Assumptions}
I consider estimation based on independent and identically distributed (i.i.d.) data. This modeling assumption is appropriate for the microeconometric applications considered in this paper.\footnote{The asymptotic analysis could be extended to allow for dependent data satisfying the $\alpha$-mixing condition under assumptions similar to those imposed by \citet{masry1997local} for the standard nonparametric regression.}

\begin{assumption}\label{ass:ass1}\,
	\begin{enumerate}[label=(\alph*),nosep]
		\item $\{(X_i, Y_i)\}_{i=1}^n$ are continuous i.i.d. random variables;
		\item The support of~$X$, denoted by $\mc X$, is an interval, and $x_0\in\mc X$;
		\item $\eta \in (0,1)$.
	\end{enumerate}
\end{assumption}

I follow the classic literature on local polynomial modeling and assume that the covariate is continuous. 
The density of $X$ is denoted by $f_X(x)$. The conditional distribution function of $Y$ given $X$ is denoted by $F_{Y|X}(y|x)$, and the corresponding conditional density by $f_{Y|X}(y|x)$. 
Subsequent assumptions involve smoothness requirements for the functions $Q(\eta,\ccdot)$ and $m(\eta,\ccdot)$. I adopt the following convention. For a point on the left (right) boundary of $\mathcal{X}$, I~define the derivative with respect to the covariate value as the right (left) derivative at that point.

\begin{assumption}\label{ass:quantile}\,
	\begin{enumerate}[label=(\alph*),nosep]
		\item $Q(\eta,x)$ is differentiable with respect to $x$ on $\mc X$ and $\partial_xQ(\eta,x)$ is Lipschitz continuous in $x$;
		\item $f_X(x)$ is continuous and positive on $\mc X$;
		\item $f_{Y|X}(y|x)$ is continuous and positive on $\{(x,y):\,x  \in \mathcal{X},\,y \in [Q(\eta,x)\pm\epsilon]\}$ for some $\epsilon>0$.
	\end{enumerate}
\end{assumption}

Assumption~\ref{ass:quantile} comprises standard conditions for the asymptotic analysis of the local linear quantile estimator.
The smoothness assumption on $Q(\eta,x)$ is used to control the order of the bias introduced by approximating the possibly nonlinear function $Q(\eta,x)$ with its first-order Taylor expansion in~$x$. The restrictions on the density $f_X(x)$ ensure that with high probability there are observations around the estimation point. The restrictions on the conditional density $f_{Y|X}(y|x)$ ensure that the conditional $\eta$-quantile function can be precisely estimated. Assumption~\ref{ass:quantile} implies also smoothness of the coefficients in the local linear approximation of the function $Q(u,\ccdot)$ for quantile levels close to $\eta$, which is exploited in the analysis with an estimated truncation quantile level.

\begin{assumption}\label{ass:mean}\,
	\begin{enumerate}[label=(\alph*),nosep]
		\item $m(\eta,x)$ is twice continuously differentiable with respect to $x$ on $\mc X$;
		\item $\Var[Y|Y \leq Q(\eta,x),X=x]$ is bounded, bounded away from zero, and continuous in $x$ on $\mc X$;
		\item $\E[|Y|^{2+\xi}\one{Y\leq Q(\eta,X)}|X=x]$ is bounded uniformly over $x$ in $\mc X$ for some  $\xi > 0$.
	\end{enumerate}
\end{assumption}

Assumption~\ref{ass:mean} is a natural adaptation of the standard conditions for the local linear estimator in the nonparametric mean regression for estimating truncated conditional expectations.  Even if the function $Q(\eta, \ccdot)$ was known, a continuous second-order derivative of $m(\eta,x)$ with respect to~$x$ would be required to characterize the leading bias introduced by approximating the function $m(\eta,\ccdot)$ with its first-order Taylor expansion. Parts (b) and (c) are needed to obtain asymptotic normality of the proposed estimator.

\begin{assumption}\label{ass:kernel}\,
	\begin{enumerate}[label=(\alph*),nosep]
		\item The kernel $k$ is a continuous, symmetric density function with compact support, say $[-1,1]$;
		\item As $n\to \infty$,  $h \to 0$, $a \to 0$, $nh \to \infty$, and $na \to \infty$.
	\end{enumerate}
\end{assumption}
The restrictions on the kernel are standard. The requirements on the bandwidths are necessary for ensuring consistency in both stages. In a preliminary study of the proposed estimator, the convergence rates of the two bandwidths are not linked, but I impose further restrictions in Theorems~\ref{th:asy_distribution}~and~\ref{th:esteta}. All results cover an important special case where $a\asymp h$.

\subsection{Asymptotic Distribution}\label{subsec:distribution}
In this section, I study the asymptotic properties of the proposed estimator when the truncation quantile level is known. The key technical result is stated in Lemma~\ref{lemma:asy_equivalence}. It shows that the estimator $\widehat{m}$ is asymptotically equivalent to its infeasible analog using the true conditional quantile function.

\begin{lemma}\label{lemma:asy_equivalence}
	Suppose that Assumptions~\ref{ass:ass1}, \ref{ass:quantile}, and \ref{ass:kernel} hold. Then
	\[
	R(\eta,x_0;a,h) \equiv \widehat{m}(\eta,x_0;a,h) - \widetilde{m}(\eta,x_0;h) = O_p(w_n(nh)^{-1/2} + w_n^2 ),
	\]
	where $w_n= a^2 + h^2 + (a+h)(a^3n)^{-1/2}$ and
	\[
	\widetilde{m}(\eta, x_0; h) = e_1^\top \argmin_{ (\beta_0,\beta_1) } \Sum k_h(X_i-x_0)(\psi_i(\eta, Q(\eta,X_i))-\beta_0 - \beta_1(X_i-x_0))^2 .
	\]
	In particular, if $a \asymp h$, then $R(\eta,x_0;a,h)=O_p(h^4+(nh)^{-1})$.
\end{lemma}

The remainder $R(\eta,x_0;a,h)$ is driven by the estimation error from the first stage on the interval $\mathcal{X}(x_0,h)\equiv [x_0-h,x_0+h] \cap \mathcal{X}$, which is relevant for the second-stage estimator. There are two sources of this estimation error. First, the function $Q(\eta, \ccdot)$ is replaced with its local linear approximation, which results in an error of order $O(h^2)$. Second, the intercept and slope of this approximation are estimated at rates $O_p(a^2+(an)^{-1/2})$ and $O_p(a+(a^3n)^{-1/2})$, respectively.\footnote{In fact, these are the only properties of the first-stage estimator required in the proof of Lemma~\ref{lemma:asy_equivalence}.} As a result, the estimated conditional quantile function satisfies
\begin{equation}
\sup_{x \in \mathcal{X}(x_0,h)} |\widehat{Q}^{ll}(\eta,x;x_0,a)- Q(\eta,x)|  = O_p(w_n).
\end{equation}
If $h(nh)^{-1/3}\prec a$, then $w_n\to 0$, and  $R(\eta,x_0;a,h)$ is of order smaller than $O_p(w_n)$. This low sensitivity to the first-stage estimation error is obtained by construction, owing to the use of an orthogonal moment.

Lemma~\ref{lemma:asy_equivalence} holds regardless of whether the variance of the outcome variable is finite or infinite. If Assumption~\ref{ass:mean} holds in addition to the assumptions of Lemma~\ref{lemma:asy_equivalence}, then asymptotic normality of $\wh m$  follows from the standard theory of local linear estimation. If the variance of the outcome variable is infinite, then the asymptotic distribution of $\wh m$ can be obtained under alternative assumptions following the steps of \citet{linton2013estimation}. I focus on the former case. 

The asymptotic distribution presented in Theorem~\ref{th:asy_distribution} involves typical kernel constants, which differ depending on whether $x_0$ lies in the interior or on the boundary of the support of $X$, but this dependence is left implicit. Let $\mu_2=\int v^2 k(v)dv$, $\kappa_0= \int k(v)^2dv $, $\bar\mu= (\bar{\mu}_2^2 - \bar{\mu}_1 \bar{\mu}_{3})/(\bar{\mu}_2\bar{\mu}_0-\bar{\mu}_1^2)$, and $\bar\kappa= \int_0^\infty(k(v)(\bar{\mu}_1 v - \bar{\mu}_2))^2dv/ (\bar{\mu}_2\bar{\mu}_0-\bar{\mu}_1^2)^2 $, where $\bar{\mu}_{j}= \int_{0}^{\infty}v^j k(v)dv$. 
If $x_0$ lies in the interior of $\mathcal{X}$, I~put $\mu=\mu_2$ and $\kappa= \kappa_0 $. 
If $x_0$ lies on the boundary of $\mathcal{X}$, I~put $\mu=\bar\mu$ and $\kappa= \bar\kappa$.

\begin{theorem}\label{th:asy_distribution}
	Suppose that Assumptions~\ref{ass:ass1}--\ref{ass:kernel} hold, and  $h(nh)^{-1/6} \prec a \prec \sqrt{h} $, e.g., $a=h$. Then
	\begin{align*}
	\sqrt{nh}\big(\widehat{m}(\eta, x_0; a, h)  - m(\eta,x_0) - \mathcal{B}(\eta, x_0)h^2 \big) \xrightarrow{d} \mathcal{N}(0, V(\eta,x_0)),
	\end{align*}
	where
	\begin{align*}
	&\mathcal{B}(\eta,x_0)  =  \frac{1}{2} \mu\, \partial_x^2 m(\eta, x_0) + o_p(1),\\
	&V(\eta, x_0) = \frac{\kappa}{\eta f_X(x_0)} \left( \Var[Y|Y\leq Q(\eta,x_0), X=x_0] + (1-\eta) \left(Q(\eta,x_0) - m(\eta,x_0)\right)^2 \right).
	\end{align*}
\end{theorem}

As in the standard nonparametric regression, the leading bias is proportional to the second derivative of the function that is being estimated. The variance is fully analogous to the variance of the unconditional truncated mean.
The additional conditions imposed on the bandwidths ensure that the remainder $R(\eta,x_0;a,h)$ is of order $o_p(h^2+(nh)^{-1/2})$. These conditions admit certain degrees of both under- and oversmoothing in the first stage relative to the second stage. For example, if $h \asymp n^{-1/5}$, then I require that $n^{-1/3}\prec a \prec n^{-1/10}$. Subject to these restrictions, the choice of the first-stage bandwidth does not affect the first-order asymptotic distribution of $\wh m$. In practice, the two bandwidths can be set equal. This choice yields the minimal rate of the remainder $R(\eta,x_0;a,h)$.

\subsection{Estimated Truncation Quantile Level}\label{sec:Extensions}
In some applications, the truncation quantile level of interest has to be estimated from the data.
In this section, I study the properties of the proposed estimator evaluated at an estimated truncation quantile level $\wh \eta$ under a high-level condition on $\wh \eta$. 

Theorem~\ref{th:esteta} provides an expansion of the estimator with an estimated truncation quantile level about the estimator using the true quantile level. To keep the exposition transparent, I restrict the analysis to bandwidths such that $a \asymp h$.
The estimator $\wh \eta$ is only required to converge at a rate not slower than the estimator $\wh{m}(\eta,x_0;a,h)$ does.

\begin{theorem}\label{th:esteta} Suppose that $\widehat{\eta} - \eta_n = O_p\big((nh)^{-1/2}\big)$ and $\eta_n-\eta=O(h^2)$  for some deterministic sequence $\eta_n$. If Assumptions~\ref{ass:ass1}--\ref{ass:kernel} hold and $a \asymp h$, then
	\begin{align*}
	\widehat{m}(\widehat{\eta},x_0;a,h) =  \widetilde{m}(\eta,x_0;h) + \partial_{\eta} m(\eta,x_0) (\widehat{\eta}-\eta)  + O_p(h^4+ (nh)^{-1}),
	\end{align*}
	where $\partial_{\eta} m(\eta,x_0)  =\frac{1}{\eta} (Q(\eta,x_0)-m(\eta,x_0)).$
\end{theorem}

The coefficient on the term $\wh \eta - \eta$ equals the first derivative of the estimand with respect to the pre-estimated parameter, which is typical for such two-step estimation problems. It is also in line with the results of \citet{shorack1974random} and \citet{Lee2009}, who study the unconditional truncated mean with random trimming proportions.
The expansion in Theorem~\ref{th:esteta} can be used on a case-by-case basis to derive the asymptotic distribution of $\widehat{m}( \wh{\eta}, x_0;a,h)$ for specific estimators $\wh{\eta}$. Two such examples are discussed in Section \ref{sec:Applications}.
I note that in Theorem~\ref{th:esteta}, it is essential that $\eta<1$, imposed in Assumption~\ref{ass:ass1}. Otherwise, if $Y$ has unbounded support, the derivative $\partial_{\eta} m(\eta,x_0)$ is infinite, and the above expansion is not valid.

\section{Bandwidth Choice and Inference}\label{sec:Inference}
The asymptotic results in Section~\ref{sec:Asymptotics} suggest that the bandwidth can be selected and inference can be conducted following methods developed for the standard nonparametric regression, ignoring the fact that the conditional quantile function is estimated in the first stage.

The bandwidth can be chosen, e.g., so as to minimize the asymptotic mean squared error, defined as $AMSE_n(h) = B(\eta,x_0)^2 h^4 + V(\eta,x_0)/(nh)$, where $B(\eta,x_0) = \frac{1}{2} \mu\, \partial_x^2 m(\eta, x_0)$.
The optimal bandwidth is then given by $h_{\text{opt}}=\left(V(\eta,x_0)/(4B(\eta,x_0)^2)\right)^{1/5}n^{-1/5}$. It can be estimated following procedures analogous to those proposed by \citet{imbens2012optimal} and \citet{calonico2014robust}. To implement these bandwidth selectors, under additional assumptions, one can estimate $\partial^2_x m(\eta,x_0)$ using the local quadratic version of the estimator $\wh m$, discussed in Online Appendix~\ref{subsec:porder}, and estimate the asymptotic variance based on the second-stage residuals.

Given a bandwidth $h$, the asymptotic distribution in Theorem~\ref{th:asy_distribution} forms the basis for conducting statistical inference.\footnote{As discussed in Section~\ref{sec:Asymptotics}, the estimator is relatively insensitive to the choice of the first-stage bandwidth. In practice, one can set $a=h$.} Constructing a confidence interval (CI) requires accounting for the bias, which can be done by adapting any of the three following approaches.
The first, classic approach is called undersmoothing (US). It relies on choosing a~`small' bandwidth that ensures that the bias is asymptotically negligible. If $h \prec n^{-1/5}$, or equivalently $nh^5 \to 0$, then the bias is of smaller order than the standard error. As a result, an asymptotically valid  $1-\alpha$ CI can be formed as 
\begin{equation}\label{eq1:CI_US}
CI_\alpha^{US}=	[\wh{m}(\eta,x_0;h,h)\pm  z_{1-\alpha/2} \cdot \wh{se}(h)],
\end{equation}
where $z_u$ is the $u$-quantile of the standard normal distribution and $\wh{se}(h)$ is some consistent standard error.
The two further approaches allow for bandwidths of order $n^{-1/5}$, such as the AMSE-optimal bandwidth.

The second approach is analogous to the robust bias corrections proposed by \citet{calonico2014robust}. It involves subtracting an estimate of the leading bias term and accounting for the additional variation in the bias-corrected estimator when forming a CI. The CI takes the form as in~\eqref{eq1:CI_US}, except that a bias-corrected estimator and an adjusted standard error are used.

The third approach is motivated by the `bias-aware' approach of \citet{armstrong2020simple}, who propose `honest' CIs that account for the largest possible bias under restrictions on the smoothness of the function that is being estimated. Suppose that $|\partial_x^2m(\eta,x_0)|$ is bounded by some known constant $M$. Then the leading bias term is bounded in absolute value by $\frac{1}{2}|\mu|Mh^2$, and an asymptotically valid $1-\alpha$ confidence interval can be formed as
\begin{equation}
CI_\alpha= [\widehat{m}(\eta,x_0;h,h) \pm  \textrm{cv}_{1-\alpha}(\widehat{r}(h)) \cdot \widehat{se}(h) ],
\end{equation}
where $\widehat{r}(h)=\frac{1}{2}|\mu|Mh^2/\widehat{se}(h)$ and $\textrm{cv}_{1-\alpha}(t)$ is the $1-\alpha$ quantile of the folded normal distribution $|\mathcal{N}(t,1)|$.\footnote{I do not discuss coverage properties uniform in the data generating processes, which would require ensuring that the remainder in Lemma~\ref{lemma:asy_equivalence} is uniformly small. } One can also account for the maximal bias of the infeasible estimator $\wt m$ conditional on the realizations of the covariate. The bandwidth can be also chosen so as to minimize the worst-case mean squared error or the length of the CI. Implementation of bandwidth selectors and of the CIs requires imposing a bound on $\partial_x^2m(\eta,x)$. See \citet{armstrong2020simple} and \citet{noack2021bias} for discussions of the choice of the smoothness constant in the standard nonparametric regression.


\section{Monte Carlo Study}\label{sec:MonteCarlo}
In this section, I present simulation evidence for two claims. First, I show that the feasible estimator $\widehat{m}$ is close to the infeasible estimator $\widetilde{m}$ in terms of the mean squared difference. Second, I show that inference based on $\widehat{m}$ performs almost identically as inference based on the infeasible estimator $\widetilde{m}$. For concreteness, in this simulation study, I use the third approach discussed in Section~\ref{sec:Inference}, which exploits a bound on $\partial_x^2m(\eta,x)$. In its implementation, I account for the exact worst-case bias of the infeasible estimator $\widetilde{m}$ conditional on the realizations of the covariate.

I generate data from a location-scale model of the form
\begin{equation}
Y=m(X) + \sigma(X)\varepsilon,
\end{equation}
where $X$ is uniformly distributed on $[-1,1]$ and $\varepsilon \sim \mathcal{N}(0,1)$.  I consider three specifications for the conditional expectation function, which were used by \citet{armstrong2020simple} in their Monte Carlo study comparing different inference methods. Let
\begin{align*}
m_1(x) & = x^2-2s(|x|-0.25),\\
m_2(x) & = x^2-2s(|x|-0.2)+2s(|x|-0.5) - 2s(|x|-0.65),\\
m_3(x) & = (x+1)^2-2s(x+0.2)+2s(x-0.2) - 2 s(x-0.4) + 2s(x-0.7) - 0.92,
\end{align*}
where $s(x)=\max\{x,0\}^2 $.  These functions are illustrated in Figure~\ref{fig:model}. Their second derivatives are bounded in absolute value by $M=2$. I~consider homoskedastic and hetersokedastic residuals, induced by functions $\sigma_1(x)=0.5$ and $\sigma_2(x)=0.5+x$, respectively.

\begin{figure}[thb]
	\centering
\begin{tikzpicture}[x=1pt,y=1pt]
\definecolor{fillColor}{RGB}{255,255,255}
\path[use as bounding box,fill=fillColor,fill opacity=0.00] (0,0) rectangle (252.94,187.90);
\begin{scope}
\path[clip] (  0.00,  0.00) rectangle (252.94,187.90);
\definecolor{drawColor}{RGB}{255,255,255}
\definecolor{fillColor}{RGB}{255,255,255}

\path[draw=drawColor,line width= 0.6pt,line join=round,line cap=round,fill=fillColor] (  0.00,  0.00) rectangle (252.94,187.90);
\end{scope}
\begin{scope}
\path[clip] ( 55.03, 30.69) rectangle (247.44,182.40);
\definecolor{fillColor}{RGB}{255,255,255}

\path[fill=fillColor] ( 55.03, 30.69) rectangle (247.44,182.40);
\definecolor{drawColor}{RGB}{0,0,0}

\path[draw=drawColor,line width= 1.1pt,dash pattern=on 4pt off 4pt ,line join=round] ( 63.78, 96.53) --
	( 67.42,100.06) --
	( 71.07,103.21) --
	( 74.71,105.99) --
	( 78.35,108.40) --
	( 82.00,110.44) --
	( 85.64,112.11) --
	( 89.29,113.40) --
	( 92.93,114.33) --
	( 96.57,114.89) --
	(100.22,115.07) --
	(103.86,114.89) --
	(107.51,114.33) --
	(111.15,113.40) --
	(114.80,112.11) --
	(118.44,110.44) --
	(122.08,108.77) --
	(125.73,107.47) --
	(129.37,106.54) --
	(133.02,105.99) --
	(136.66,105.80) --
	(140.31,105.99) --
	(143.95,106.54) --
	(147.59,107.47) --
	(151.24,108.77) --
	(154.88,110.44) --
	(158.53,112.11) --
	(162.17,113.40) --
	(165.81,114.33) --
	(169.46,114.89) --
	(173.10,115.07) --
	(176.75,114.89) --
	(180.39,114.33) --
	(184.04,113.40) --
	(187.68,112.11) --
	(191.32,110.44) --
	(194.97,108.40) --
	(198.61,105.99) --
	(202.26,103.21) --
	(205.90,100.06) --
	(209.55, 96.53);

\path[draw=drawColor,line width= 1.1pt,line join=round] ( 63.78,103.95) --
	( 67.42,105.99) --
	( 71.07,107.66) --
	( 74.71,108.95) --
	( 78.35,109.88) --
	( 82.00,110.44) --
	( 85.64,110.62) --
	( 89.29,110.44) --
	( 92.93,110.25) --
	( 96.57,110.44) --
	(100.22,110.99) --
	(103.86,111.55) --
	(107.51,111.73) --
	(111.15,111.55) --
	(114.80,110.99) --
	(118.44,110.07) --
	(122.08,108.77) --
	(125.73,107.47) --
	(129.37,106.54) --
	(133.02,105.99) --
	(136.66,105.80) --
	(140.31,105.99) --
	(143.95,106.54) --
	(147.59,107.47) --
	(151.24,108.77) --
	(154.88,110.07) --
	(158.53,110.99) --
	(162.17,111.55) --
	(165.81,111.73) --
	(169.46,111.55) --
	(173.10,110.99) --
	(176.75,110.44) --
	(180.39,110.25) --
	(184.04,110.44) --
	(187.68,110.62) --
	(191.32,110.44) --
	(194.97,109.88) --
	(198.61,108.95) --
	(202.26,107.66) --
	(205.90,105.99) --
	(209.55,103.95);

\path[draw=drawColor,line width= 1.1pt,line join=round] ( 63.78, 37.58) --
	( 67.42, 37.77) --
	( 71.07, 38.32) --
	( 74.71, 39.25) --
	( 78.35, 40.55) --
	( 82.00, 42.22) --
	( 85.64, 44.26) --
	( 89.29, 46.67) --
	( 92.93, 49.45) --
	( 96.57, 52.60) --
	(100.22, 56.12) --
	(103.86, 60.01) --
	(107.51, 64.28) --
	(111.15, 68.91) --
	(114.80, 73.92) --
	(118.44, 79.29) --
	(122.08, 85.04) --
	(125.73, 90.79) --
	(129.37, 96.16) --
	(133.02,101.17) --
	(136.66,105.80) --
	(140.31,110.07) --
	(143.95,113.96) --
	(147.59,117.48) --
	(151.24,120.63) --
	(154.88,123.78) --
	(158.53,127.31) --
	(162.17,131.20) --
	(165.81,135.46) --
	(169.46,139.73) --
	(173.10,143.62) --
	(176.75,147.14) --
	(180.39,150.29) --
	(184.04,153.07) --
	(187.68,155.48) --
	(191.32,157.89) --
	(194.97,160.68) --
	(198.61,163.83) --
	(202.26,167.35) --
	(205.90,171.24) --
	(209.55,175.51);

\node[text=drawColor,anchor=base,inner sep=0pt, outer sep=0pt, scale=  0.92] at (227.77, 91.25) {Design 1};

\node[text=drawColor,anchor=base,inner sep=0pt, outer sep=0pt, scale=  0.92] at (227.77,101.63) {Design 2};

\node[text=drawColor,anchor=base,inner sep=0pt, outer sep=0pt, scale=  0.92] at (227.77,171.70) {Design 3};
\end{scope}
\begin{scope}
\path[clip] (  0.00,  0.00) rectangle (252.94,187.90);
\definecolor{drawColor}{RGB}{0,0,0}

\path[draw=drawColor,line width= 0.6pt,line join=round] ( 55.03, 30.69) --
	( 55.03,182.40);
\end{scope}
\begin{scope}
\path[clip] (  0.00,  0.00) rectangle (252.94,187.90);
\definecolor{drawColor}{gray}{0.30}

\node[text=drawColor,anchor=base east,inner sep=0pt, outer sep=0pt, scale=  0.73] at ( 50.08, 28.62) {-1.0};

\node[text=drawColor,anchor=base east,inner sep=0pt, outer sep=0pt, scale=  0.73] at ( 50.08, 65.70) {-0.5};

\node[text=drawColor,anchor=base east,inner sep=0pt, outer sep=0pt, scale=  0.73] at ( 50.08,102.77) {0.0};

\node[text=drawColor,anchor=base east,inner sep=0pt, outer sep=0pt, scale=  0.73] at ( 50.08,139.85) {0.5};

\node[text=drawColor,anchor=base east,inner sep=0pt, outer sep=0pt, scale=  0.73] at ( 50.08,176.92) {1.0};
\end{scope}
\begin{scope}
\path[clip] (  0.00,  0.00) rectangle (252.94,187.90);
\definecolor{drawColor}{gray}{0.20}

\path[draw=drawColor,line width= 0.6pt,line join=round] ( 52.28, 31.65) --
	( 55.03, 31.65);

\path[draw=drawColor,line width= 0.6pt,line join=round] ( 52.28, 68.73) --
	( 55.03, 68.73);

\path[draw=drawColor,line width= 0.6pt,line join=round] ( 52.28,105.80) --
	( 55.03,105.80);

\path[draw=drawColor,line width= 0.6pt,line join=round] ( 52.28,142.88) --
	( 55.03,142.88);

\path[draw=drawColor,line width= 0.6pt,line join=round] ( 52.28,179.95) --
	( 55.03,179.95);
\end{scope}
\begin{scope}
\path[clip] (  0.00,  0.00) rectangle (252.94,187.90);
\definecolor{drawColor}{RGB}{0,0,0}

\path[draw=drawColor,line width= 0.6pt,line join=round] ( 55.03, 30.69) --
	(247.44, 30.69);
\end{scope}
\begin{scope}
\path[clip] (  0.00,  0.00) rectangle (252.94,187.90);
\definecolor{drawColor}{gray}{0.20}

\path[draw=drawColor,line width= 0.6pt,line join=round] ( 63.78, 27.94) --
	( 63.78, 30.69);

\path[draw=drawColor,line width= 0.6pt,line join=round] (100.22, 27.94) --
	(100.22, 30.69);

\path[draw=drawColor,line width= 0.6pt,line join=round] (136.66, 27.94) --
	(136.66, 30.69);

\path[draw=drawColor,line width= 0.6pt,line join=round] (173.10, 27.94) --
	(173.10, 30.69);

\path[draw=drawColor,line width= 0.6pt,line join=round] (209.55, 27.94) --
	(209.55, 30.69);

\path[draw=drawColor,line width= 0.6pt,line join=round] (245.99, 27.94) --
	(245.99, 30.69);
\end{scope}
\begin{scope}
\path[clip] (  0.00,  0.00) rectangle (252.94,187.90);
\definecolor{drawColor}{gray}{0.30}

\node[text=drawColor,anchor=base,inner sep=0pt, outer sep=0pt, scale=  0.73] at ( 63.78, 19.68) {-1.0};

\node[text=drawColor,anchor=base,inner sep=0pt, outer sep=0pt, scale=  0.73] at (100.22, 19.68) {-0.5};

\node[text=drawColor,anchor=base,inner sep=0pt, outer sep=0pt, scale=  0.73] at (136.66, 19.68) {0.0};

\node[text=drawColor,anchor=base,inner sep=0pt, outer sep=0pt, scale=  0.73] at (173.10, 19.68) {0.5};

\node[text=drawColor,anchor=base,inner sep=0pt, outer sep=0pt, scale=  0.73] at (209.55, 19.68) {1.0};

\node[text=drawColor,anchor=base,inner sep=0pt, outer sep=0pt, scale=  0.73] at (245.99, 19.68) {1.5};
\end{scope}
\begin{scope}
\path[clip] (  0.00,  0.00) rectangle (252.94,187.90);
\definecolor{drawColor}{RGB}{0,0,0}

\node[text=drawColor,anchor=base,inner sep=0pt, outer sep=0pt, scale=  0.92] at (151.24,  7.64) {$x$};
\end{scope}
\begin{scope}
\path[clip] (  0.00,  0.00) rectangle (252.94,187.90);
\definecolor{drawColor}{RGB}{0,0,0}

\node[text=drawColor,anchor=base,inner sep=0pt, outer sep=0pt, scale=  0.92] at ( 19.48,102.76) {$m_j(x)$};
\end{scope}
\end{tikzpicture}
	\vspace{-12pt}
	\caption{Conditional expectation functions $m_j(x)$.}
	\label{fig:model}
\end{figure}
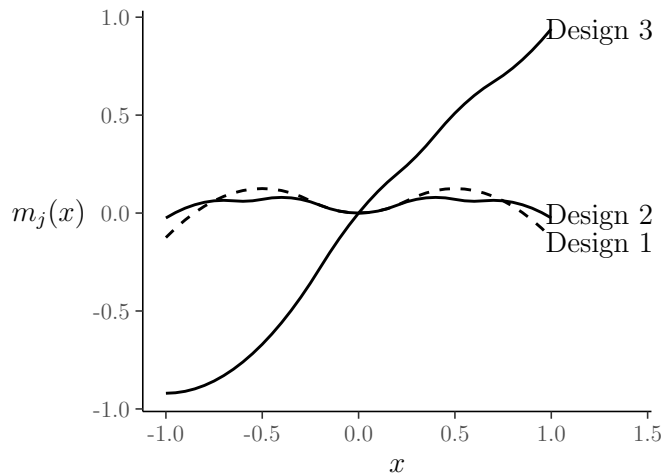

\begin{figure}[htb]
	\centering
	\begin{subfigure}{0.48\textwidth}
		\scalebox{0.95}{
\begin{tikzpicture}[x=1pt,y=1pt]
\definecolor{fillColor}{RGB}{255,255,255}
\path[use as bounding box,fill=fillColor,fill opacity=0.00] (0,0) rectangle (231.26,173.45);
\begin{scope}
\path[clip] (  0.00,  0.00) rectangle (231.26,173.45);
\definecolor{drawColor}{RGB}{255,255,255}
\definecolor{fillColor}{RGB}{255,255,255}

\path[draw=drawColor,line width= 0.6pt,line join=round,line cap=round,fill=fillColor] (  0.00,  0.00) rectangle (231.26,173.45);
\end{scope}
\begin{scope}
\path[clip] ( 27.08, 30.69) rectangle (225.76,167.95);
\definecolor{fillColor}{RGB}{255,255,255}

\path[fill=fillColor] ( 27.08, 30.69) rectangle (225.76,167.95);
\definecolor{drawColor}{RGB}{0,0,0}

\path[draw=drawColor,line width= 1.1pt,line join=round] ( 36.11, 68.12) --
	( 39.87, 72.07) --
	( 43.64, 75.61) --
	( 47.40, 78.73) --
	( 51.16, 81.43) --
	( 54.92, 83.72) --
	( 58.69, 85.59) --
	( 62.45, 87.05) --
	( 66.21, 88.09) --
	( 69.98, 88.71) --
	( 73.74, 88.92) --
	( 77.50, 88.71) --
	( 81.27, 88.09) --
	( 85.03, 87.05) --
	( 88.79, 85.59) --
	( 92.55, 83.72) --
	( 96.32, 81.85) --
	(100.08, 80.39) --
	(103.84, 79.35) --
	(107.61, 78.73) --
	(111.37, 78.52) --
	(115.13, 78.73) --
	(118.90, 79.35) --
	(122.66, 80.39) --
	(126.42, 81.85) --
	(130.18, 83.72) --
	(133.95, 85.59) --
	(137.71, 87.05) --
	(141.47, 88.09) --
	(145.24, 88.71) --
	(149.00, 88.92) --
	(152.76, 88.71) --
	(156.53, 88.09) --
	(160.29, 87.05) --
	(164.05, 85.59) --
	(167.81, 83.72) --
	(171.58, 81.43) --
	(175.34, 78.73) --
	(179.10, 75.61) --
	(182.87, 72.07) --
	(186.63, 68.12);

\path[draw=drawColor,line width= 1.1pt,dash pattern=on 2pt off 2pt on 6pt off 2pt ,line join=round] ( 36.11, 56.48) --
	( 39.87, 60.43) --
	( 43.64, 63.96) --
	( 47.40, 67.08) --
	( 51.16, 69.79) --
	( 54.92, 72.07) --
	( 58.69, 73.95) --
	( 62.45, 75.40) --
	( 66.21, 76.44) --
	( 69.98, 77.07) --
	( 73.74, 77.27) --
	( 77.50, 77.07) --
	( 81.27, 76.44) --
	( 85.03, 75.40) --
	( 88.79, 73.95) --
	( 92.55, 72.07) --
	( 96.32, 70.20) --
	(100.08, 68.75) --
	(103.84, 67.71) --
	(107.61, 67.08) --
	(111.37, 66.87) --
	(115.13, 67.08) --
	(118.90, 67.71) --
	(122.66, 68.75) --
	(126.42, 70.20) --
	(130.18, 72.07) --
	(133.95, 73.95) --
	(137.71, 75.40) --
	(141.47, 76.44) --
	(145.24, 77.07) --
	(149.00, 77.27) --
	(152.76, 77.07) --
	(156.53, 76.44) --
	(160.29, 75.40) --
	(164.05, 73.95) --
	(167.81, 72.07) --
	(171.58, 69.79) --
	(175.34, 67.08) --
	(179.10, 63.96) --
	(182.87, 60.43) --
	(186.63, 56.48);

\path[draw=drawColor,line width= 1.1pt,dash pattern=on 4pt off 4pt ,line join=round] ( 36.11,103.13) --
	( 39.87,107.08) --
	( 43.64,110.61) --
	( 47.40,113.73) --
	( 51.16,116.44) --
	( 54.92,118.73) --
	( 58.69,120.60) --
	( 62.45,122.05) --
	( 66.21,123.09) --
	( 69.98,123.72) --
	( 73.74,123.93) --
	( 77.50,123.72) --
	( 81.27,123.09) --
	( 85.03,122.05) --
	( 88.79,120.60) --
	( 92.55,118.73) --
	( 96.32,116.85) --
	(100.08,115.40) --
	(103.84,114.36) --
	(107.61,113.73) --
	(111.37,113.53) --
	(115.13,113.73) --
	(118.90,114.36) --
	(122.66,115.40) --
	(126.42,116.85) --
	(130.18,118.73) --
	(133.95,120.60) --
	(137.71,122.05) --
	(141.47,123.09) --
	(145.24,123.72) --
	(149.00,123.93) --
	(152.76,123.72) --
	(156.53,123.09) --
	(160.29,122.05) --
	(164.05,120.60) --
	(167.81,118.73) --
	(171.58,116.44) --
	(175.34,113.73) --
	(179.10,110.61) --
	(182.87,107.08) --
	(186.63,103.13);

\node[text=drawColor,anchor=base,inner sep=0pt, outer sep=0pt, scale=  0.92] at (205.44, 64.32) {$m(x)$};

\node[text=drawColor,anchor=base,inner sep=0pt, outer sep=0pt, scale=  0.92] at (205.44, 52.67) {$m(\eta,x)$};

\node[text=drawColor,anchor=base,inner sep=0pt, outer sep=0pt, scale=  0.92] at (205.44, 99.33) {$Q(\eta,x)$};
\end{scope}
\begin{scope}
\path[clip] (  0.00,  0.00) rectangle (231.26,173.45);
\definecolor{drawColor}{RGB}{0,0,0}

\path[draw=drawColor,line width= 0.6pt,line join=round] ( 27.08, 30.69) --
	( 27.08,167.95);
\end{scope}
\begin{scope}
\path[clip] (  0.00,  0.00) rectangle (231.26,173.45);
\definecolor{drawColor}{gray}{0.30}

\node[text=drawColor,anchor=base east,inner sep=0pt, outer sep=0pt, scale=  0.73] at ( 22.13, 33.89) {-0.5};

\node[text=drawColor,anchor=base east,inner sep=0pt, outer sep=0pt, scale=  0.73] at ( 22.13, 75.49) {0.0};

\node[text=drawColor,anchor=base east,inner sep=0pt, outer sep=0pt, scale=  0.73] at ( 22.13,117.08) {0.5};

\node[text=drawColor,anchor=base east,inner sep=0pt, outer sep=0pt, scale=  0.73] at ( 22.13,158.68) {1.0};
\end{scope}
\begin{scope}
\path[clip] (  0.00,  0.00) rectangle (231.26,173.45);
\definecolor{drawColor}{gray}{0.20}

\path[draw=drawColor,line width= 0.6pt,line join=round] ( 24.33, 36.92) --
	( 27.08, 36.92);

\path[draw=drawColor,line width= 0.6pt,line join=round] ( 24.33, 78.52) --
	( 27.08, 78.52);

\path[draw=drawColor,line width= 0.6pt,line join=round] ( 24.33,120.11) --
	( 27.08,120.11);

\path[draw=drawColor,line width= 0.6pt,line join=round] ( 24.33,161.71) --
	( 27.08,161.71);
\end{scope}
\begin{scope}
\path[clip] (  0.00,  0.00) rectangle (231.26,173.45);
\definecolor{drawColor}{RGB}{0,0,0}

\path[draw=drawColor,line width= 0.6pt,line join=round] ( 27.08, 30.69) --
	(225.76, 30.69);
\end{scope}
\begin{scope}
\path[clip] (  0.00,  0.00) rectangle (231.26,173.45);
\definecolor{drawColor}{gray}{0.20}

\path[draw=drawColor,line width= 0.6pt,line join=round] ( 36.11, 27.94) --
	( 36.11, 30.69);

\path[draw=drawColor,line width= 0.6pt,line join=round] ( 73.74, 27.94) --
	( 73.74, 30.69);

\path[draw=drawColor,line width= 0.6pt,line join=round] (111.37, 27.94) --
	(111.37, 30.69);

\path[draw=drawColor,line width= 0.6pt,line join=round] (149.00, 27.94) --
	(149.00, 30.69);

\path[draw=drawColor,line width= 0.6pt,line join=round] (186.63, 27.94) --
	(186.63, 30.69);

\path[draw=drawColor,line width= 0.6pt,line join=round] (224.26, 27.94) --
	(224.26, 30.69);
\end{scope}
\begin{scope}
\path[clip] (  0.00,  0.00) rectangle (231.26,173.45);
\definecolor{drawColor}{gray}{0.30}

\node[text=drawColor,anchor=base,inner sep=0pt, outer sep=0pt, scale=  0.73] at ( 36.11, 19.68) {-1.0};

\node[text=drawColor,anchor=base,inner sep=0pt, outer sep=0pt, scale=  0.73] at ( 73.74, 19.68) {-0.5};

\node[text=drawColor,anchor=base,inner sep=0pt, outer sep=0pt, scale=  0.73] at (111.37, 19.68) {0.0};

\node[text=drawColor,anchor=base,inner sep=0pt, outer sep=0pt, scale=  0.73] at (149.00, 19.68) {0.5};

\node[text=drawColor,anchor=base,inner sep=0pt, outer sep=0pt, scale=  0.73] at (186.63, 19.68) {1.0};

\node[text=drawColor,anchor=base,inner sep=0pt, outer sep=0pt, scale=  0.73] at (224.26, 19.68) {1.5};
\end{scope}
\begin{scope}
\path[clip] (  0.00,  0.00) rectangle (231.26,173.45);
\definecolor{drawColor}{RGB}{0,0,0}

\node[text=drawColor,anchor=base,inner sep=0pt, outer sep=0pt, scale=  0.92] at (126.42,  7.64) {$x$};
\end{scope}
\end{tikzpicture}}
		\vspace{-12pt}
		\subcaption{Homoskedastic case with $\sigma(x)=0.5$.}
	\end{subfigure}
	\begin{subfigure}{0.48\textwidth}
		\scalebox{0.95}{
\begin{tikzpicture}[x=1pt,y=1pt]
\definecolor{fillColor}{RGB}{255,255,255}
\path[use as bounding box,fill=fillColor,fill opacity=0.00] (0,0) rectangle (231.26,173.45);
\begin{scope}
\path[clip] (  0.00,  0.00) rectangle (231.26,173.45);
\definecolor{drawColor}{RGB}{255,255,255}
\definecolor{fillColor}{RGB}{255,255,255}

\path[draw=drawColor,line width= 0.6pt,line join=round,line cap=round,fill=fillColor] (  0.00,  0.00) rectangle (231.26,173.45);
\end{scope}
\begin{scope}
\path[clip] ( 27.08, 30.69) rectangle (225.76,167.95);
\definecolor{fillColor}{RGB}{255,255,255}

\path[fill=fillColor] ( 27.08, 30.69) rectangle (225.76,167.95);
\definecolor{drawColor}{RGB}{0,0,0}

\path[draw=drawColor,line width= 1.1pt,line join=round] ( 36.11, 68.12) --
	( 39.87, 72.07) --
	( 43.64, 75.61) --
	( 47.40, 78.73) --
	( 51.16, 81.43) --
	( 54.92, 83.72) --
	( 58.69, 85.59) --
	( 62.45, 87.05) --
	( 66.21, 88.09) --
	( 69.98, 88.71) --
	( 73.74, 88.92) --
	( 77.50, 88.71) --
	( 81.27, 88.09) --
	( 85.03, 87.05) --
	( 88.79, 85.59) --
	( 92.55, 83.72) --
	( 96.32, 81.85) --
	(100.08, 80.39) --
	(103.84, 79.35) --
	(107.61, 78.73) --
	(111.37, 78.52) --
	(115.13, 78.73) --
	(118.90, 79.35) --
	(122.66, 80.39) --
	(126.42, 81.85) --
	(130.18, 83.72) --
	(133.95, 85.59) --
	(137.71, 87.05) --
	(141.47, 88.09) --
	(145.24, 88.71) --
	(149.00, 88.92) --
	(152.76, 88.71) --
	(156.53, 88.09) --
	(160.29, 87.05) --
	(164.05, 85.59) --
	(167.81, 83.72) --
	(171.58, 81.43) --
	(175.34, 78.73) --
	(179.10, 75.61) --
	(182.87, 72.07) --
	(186.63, 68.12);

\path[draw=drawColor,line width= 1.1pt,dash pattern=on 2pt off 2pt on 6pt off 2pt ,line join=round] ( 36.11, 68.12) --
	( 39.87, 71.49) --
	( 43.64, 74.44) --
	( 47.40, 76.98) --
	( 51.16, 79.10) --
	( 54.92, 80.81) --
	( 58.69, 82.10) --
	( 62.45, 82.97) --
	( 66.21, 83.43) --
	( 69.98, 83.47) --
	( 73.74, 83.10) --
	( 77.50, 82.31) --
	( 81.27, 81.10) --
	( 85.03, 79.48) --
	( 88.79, 77.44) --
	( 92.55, 74.99) --
	( 96.32, 72.53) --
	(100.08, 70.49) --
	(103.84, 68.87) --
	(107.61, 67.66) --
	(111.37, 66.87) --
	(115.13, 66.50) --
	(118.90, 66.54) --
	(122.66, 67.00) --
	(126.42, 67.87) --
	(130.18, 69.16) --
	(133.95, 70.45) --
	(137.71, 71.33) --
	(141.47, 71.78) --
	(145.24, 71.83) --
	(149.00, 71.45) --
	(152.76, 70.66) --
	(156.53, 69.45) --
	(160.29, 67.83) --
	(164.05, 65.79) --
	(167.81, 63.34) --
	(171.58, 60.47) --
	(175.34, 57.18) --
	(179.10, 53.48) --
	(182.87, 49.36) --
	(186.63, 44.83);

\path[draw=drawColor,line width= 1.1pt,dash pattern=on 4pt off 4pt ,line join=round] ( 36.11, 68.12) --
	( 39.87, 73.82) --
	( 43.64, 79.11) --
	( 47.40, 83.98) --
	( 51.16, 88.43) --
	( 54.92, 92.47) --
	( 58.69, 96.09) --
	( 62.45, 99.30) --
	( 66.21,102.09) --
	( 69.98,104.46) --
	( 73.74,106.42) --
	( 77.50,107.96) --
	( 81.27,109.09) --
	( 85.03,109.80) --
	( 88.79,110.10) --
	( 92.55,109.97) --
	( 96.32,109.85) --
	(100.08,110.15) --
	(103.84,110.86) --
	(107.61,111.98) --
	(111.37,113.53) --
	(115.13,115.48) --
	(118.90,117.86) --
	(122.66,120.65) --
	(126.42,123.86) --
	(130.18,127.48) --
	(133.95,131.10) --
	(137.71,134.31) --
	(141.47,137.10) --
	(145.24,139.47) --
	(149.00,141.43) --
	(152.76,142.97) --
	(156.53,144.10) --
	(160.29,144.81) --
	(164.05,145.10) --
	(167.81,144.98) --
	(171.58,144.44) --
	(175.34,143.49) --
	(179.10,142.12) --
	(182.87,140.34) --
	(186.63,138.13);

\node[text=drawColor,anchor=base,inner sep=0pt, outer sep=0pt, scale=  0.92] at (205.44, 64.32) {$m(x)$};

\node[text=drawColor,anchor=base,inner sep=0pt, outer sep=0pt, scale=  0.92] at (205.44, 41.03) {$m(\eta,x)$};

\node[text=drawColor,anchor=base,inner sep=0pt, outer sep=0pt, scale=  0.92] at (205.44,134.33) {$Q(\eta,x)$};
\end{scope}
\begin{scope}
\path[clip] (  0.00,  0.00) rectangle (231.26,173.45);
\definecolor{drawColor}{RGB}{0,0,0}

\path[draw=drawColor,line width= 0.6pt,line join=round] ( 27.08, 30.69) --
	( 27.08,167.95);
\end{scope}
\begin{scope}
\path[clip] (  0.00,  0.00) rectangle (231.26,173.45);
\definecolor{drawColor}{gray}{0.30}

\node[text=drawColor,anchor=base east,inner sep=0pt, outer sep=0pt, scale=  0.73] at ( 22.13, 33.89) {-0.5};

\node[text=drawColor,anchor=base east,inner sep=0pt, outer sep=0pt, scale=  0.73] at ( 22.13, 75.49) {0.0};

\node[text=drawColor,anchor=base east,inner sep=0pt, outer sep=0pt, scale=  0.73] at ( 22.13,117.08) {0.5};

\node[text=drawColor,anchor=base east,inner sep=0pt, outer sep=0pt, scale=  0.73] at ( 22.13,158.68) {1.0};
\end{scope}
\begin{scope}
\path[clip] (  0.00,  0.00) rectangle (231.26,173.45);
\definecolor{drawColor}{gray}{0.20}

\path[draw=drawColor,line width= 0.6pt,line join=round] ( 24.33, 36.92) --
	( 27.08, 36.92);

\path[draw=drawColor,line width= 0.6pt,line join=round] ( 24.33, 78.52) --
	( 27.08, 78.52);

\path[draw=drawColor,line width= 0.6pt,line join=round] ( 24.33,120.11) --
	( 27.08,120.11);

\path[draw=drawColor,line width= 0.6pt,line join=round] ( 24.33,161.71) --
	( 27.08,161.71);
\end{scope}
\begin{scope}
\path[clip] (  0.00,  0.00) rectangle (231.26,173.45);
\definecolor{drawColor}{RGB}{0,0,0}

\path[draw=drawColor,line width= 0.6pt,line join=round] ( 27.08, 30.69) --
	(225.76, 30.69);
\end{scope}
\begin{scope}
\path[clip] (  0.00,  0.00) rectangle (231.26,173.45);
\definecolor{drawColor}{gray}{0.20}

\path[draw=drawColor,line width= 0.6pt,line join=round] ( 36.11, 27.94) --
	( 36.11, 30.69);

\path[draw=drawColor,line width= 0.6pt,line join=round] ( 73.74, 27.94) --
	( 73.74, 30.69);

\path[draw=drawColor,line width= 0.6pt,line join=round] (111.37, 27.94) --
	(111.37, 30.69);

\path[draw=drawColor,line width= 0.6pt,line join=round] (149.00, 27.94) --
	(149.00, 30.69);

\path[draw=drawColor,line width= 0.6pt,line join=round] (186.63, 27.94) --
	(186.63, 30.69);

\path[draw=drawColor,line width= 0.6pt,line join=round] (224.26, 27.94) --
	(224.26, 30.69);
\end{scope}
\begin{scope}
\path[clip] (  0.00,  0.00) rectangle (231.26,173.45);
\definecolor{drawColor}{gray}{0.30}

\node[text=drawColor,anchor=base,inner sep=0pt, outer sep=0pt, scale=  0.73] at ( 36.11, 19.68) {-1.0};

\node[text=drawColor,anchor=base,inner sep=0pt, outer sep=0pt, scale=  0.73] at ( 73.74, 19.68) {-0.5};

\node[text=drawColor,anchor=base,inner sep=0pt, outer sep=0pt, scale=  0.73] at (111.37, 19.68) {0.0};

\node[text=drawColor,anchor=base,inner sep=0pt, outer sep=0pt, scale=  0.73] at (149.00, 19.68) {0.5};

\node[text=drawColor,anchor=base,inner sep=0pt, outer sep=0pt, scale=  0.73] at (186.63, 19.68) {1.0};

\node[text=drawColor,anchor=base,inner sep=0pt, outer sep=0pt, scale=  0.73] at (224.26, 19.68) {1.5};
\end{scope}
\begin{scope}
\path[clip] (  0.00,  0.00) rectangle (231.26,173.45);
\definecolor{drawColor}{RGB}{0,0,0}

\node[text=drawColor,anchor=base,inner sep=0pt, outer sep=0pt, scale=  0.92] at (126.42,  7.64) {$x$};
\end{scope}
\end{tikzpicture}}
		\vspace{-12pt}
		\subcaption{Hetersokedastic case with $\sigma(x)=0.5\cdot(1+x)$.}
	\end{subfigure}
	\caption[Truncated conditional expectation functions.]{Truncated conditional expectation functions for $m(x)=m_1(x)$ and  $\eta=0.8$.}
	\label{fig:model_trunc}
\end{figure}
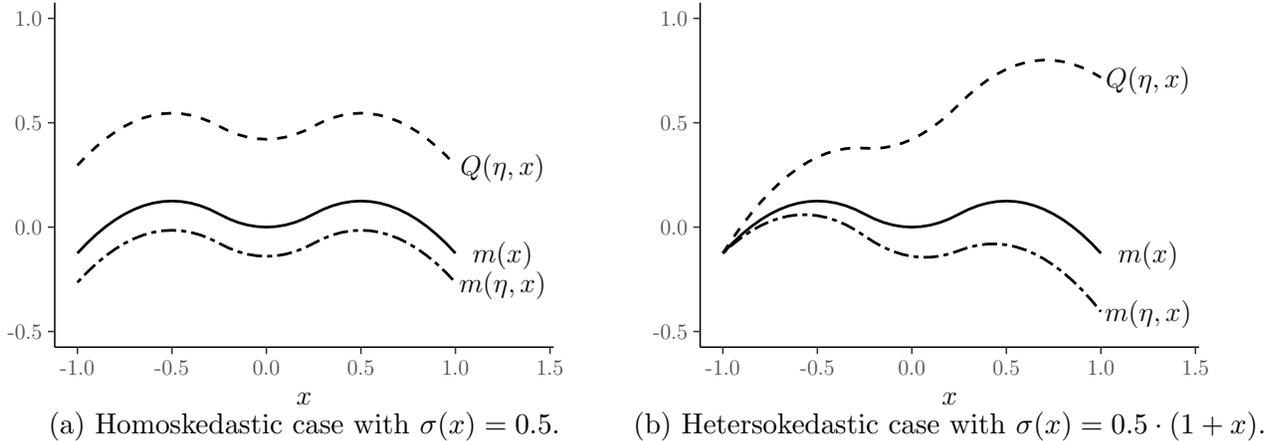

Due to normality of the residuals, the truncated conditional expectation functions have a simple, closed-form expression. It holds that
\begin{equation}
m(\eta,x)=m(x)- \frac{\phi(q_\eta)}{\eta} \sigma(x), 
\end{equation}
where $\phi(\cdot)$ is the density and $q_\eta$ is the $\eta$-quantile of the standard normal distribution, respectively.
With homoskedastic residuals, the truncated conditional expectation functions have the same shape as the respective conditional expectation functions, but they are shifted downwards. 
With heteroskedastic residuals, the slopes change as well, but this type of heteroskedasticity does not affect the curvature. Figure~\ref{fig:model_trunc} illustrates that for $\eta=0.8$ and $m(x)=m_1(x)$. Other cases are analogous.

In all simulations, the sample size is $n=1,000$, and the number of replications is $S=10,000$. The truncated conditional expectation functions are estimated at $x_0=0$ and three quantile levels, $\eta \in \{0.2, 0.5, 0.8\}$. I use the triangular kernel and the EHW variance estimator.

In Table \ref{table:sim1}, I report the root mean squared error (RMSE) of the infeasible estimator $\widetilde{m}$ and the feasible estimator $\wh{m}$, as well as the root mean squared error difference between the two.
The estimators are evaluated with the RMSE-optimal bandwidth chosen for the infeasible estimator using the bandwidth selector of \citet{armstrong2020simple} employing the true smoothness constant ($M=2$). In all cases, the difference between the infeasible and feasible estimators is small compared to their mean squared errors.\footnote{This qualitative conclusion remains the same when using the true bound multiplied or divided by two.} Moreover, the results are very similar in the homoskedastic and heteroskedastic settings, which shows that the estimator adapts to different slopes of the conditional quantile and truncated expectation functions very well. Depending on the specific data generating process, the feasible estimator may perform slightly better or worse than the infeasible estimator, which is due to the higher-order remainder terms in Lemma~\ref{lemma:asy_equivalence}.

\begin{table}[htb]
	\centering
	\caption{RMSE and root mean squared distance to $\wt m$.}
	\scalebox{0.95}{\begin{threeparttable}[hbt]
			\begin{tabular}{llcccccc}
				\hline
				& & \multicolumn{3}{c}{RMSE} & \multicolumn{3}{c}{Distance to $\wt m$}  \\
				& Design for $m_j$: & 1 & 2 & 3 & 1 & 2 & 3 \\ \hline 
				\multicolumn{8}{l}{\textit{Homoskedastic errors} } \\
				\multirow{2}{*}{$\eta=0.2$}						&	Infeasible $\widetilde{m}$& 5.044 & 5.002 & 5.146 & - & - & - \\ 
				& Feasible $\wh{m}$& 5.273 & 5.222 & 4.965 & 0.563 & 0.569 & 0.575 \vspace{4pt} \\
				\multirow{2}{*}{$\eta=0.5$}						&	Infeasible $\widetilde{m}$ & 4.094 & 4.068 & 4.134 & - & - & - \\ 
				& Feasible $\wh{m}$  & 4.202 & 4.174 & 4.041 & 0.277 & 0.280 & 0.282 \vspace{4pt} \\ 
				\multirow{2}{*}{$\eta=0.8$}					&	Infeasible $\widetilde{m}$  & 3.742 & 3.721 & 3.759 & - & - & - \\ 
				& Feasible $\wh{m}$& 3.804 & 3.782 & 3.707 & 0.164 & 0.165 & 0.166 \vspace{4pt}\\ 
				\multicolumn{8}{l}{\textit{Heteroskedastic errors} } \\
				\multirow{2}{*}{$\eta=0.2$}							&	Infeasible $\widetilde{m}$ & 5.095 & 5.032 & 5.177 & - & - & - \\ 
				&	Feasible $\wh{m}$ & 5.306 & 5.236 & 5.006 & 0.548 & 0.551 & 0.556 \vspace{4pt} \\ 
				\multirow{2}{*}{$\eta=0.5$}				&	Infeasible $\widetilde{m}$ &  4.126 & 4.091 & 4.157 & - & - & - \\ 
				&	Feasible $\wh{m}$  & 4.230 & 4.192 & 4.070 & 0.271 & 0.271 & 0.273 \vspace{4pt} \\ 
				\multirow{2}{*}{$\eta=0.8$}							&	Infeasible $\widetilde{m}$ & 3.766 & 3.742 & 3.782 & - & - & - \\ 
				&	Feasible $\wh{m}$ &  3.825 & 3.800 & 3.731 & 0.161 & 0.160 & 0.161 \\ 
				\hline
			\end{tabular}
			\begin{tablenotes}\small
				\item\hspace{-3pt}\textit{Notes:} All values are multiplied by 100. The estimators are evaluated with the RMSE-optimal bandwidth for the infeasible estimator $\wt m$ based on the true smoothness constant. The sample size is $n=1,000$, and the number of simulations is $S=10,000$.
			\end{tablenotes}
	\end{threeparttable}}
	\label{table:sim1}
\end{table}

\begin{table}[htb]
	\centering 
	\caption{Coverage, average bandwidth, and average length of the 95\% CI.}
	\scalebox{0.95}{\begin{threeparttable}[hbt]
			\begin{tabular}{llccccccccc}
				\hline
				& & \multicolumn{3}{c}{Coverage} & \multicolumn{3}{c}{Bandwidth} & \multicolumn{3}{c}{CI length} \\
				& Design for $m_j$:	& 1 & 2 & 3 & 1 & 2 & 3 & 1 & 2 & 3 \\ 
				\hline
				\multicolumn{10}{l}{\textit{Homoskedastic errors} } \\
				
				\multirow{2}{*}{$\eta=0.2$} & Infeasible $\widetilde{m}$ & 92.1 & 92.4 & 96.1 & 0.373 & 0.372 & 0.369 & 0.099 & 0.099 & 0.099 \\ 
				& Feasible $\wh{m}$ & 92.1 & 92.3 & 96.1 & 0.366 & 0.368 & 0.374 & 0.100 & 0.100 & 0.098\vspace{4pt} \\ 
				\multirow{2}{*}{$\eta=0.5$} &	Infeasible $\widetilde{m}$ & 93.5 & 93.7 & 96.0 & 0.334 & 0.334 & 0.333 & 0.080 & 0.080 & 0.080 \\ 
				& Feasible $\wh{m}$ & 93.6 & 93.8 & 95.9 & 0.331 & 0.332 & 0.335 & 0.081 & 0.081 & 0.080\vspace{4pt} \\ 
				\multirow{2}{*}{$\eta=0.8$} &	Infeasible $\widetilde{m}$  & 94.4 & 94.6 & 95.7 & 0.319 & 0.319 & 0.318 & 0.073 & 0.073 & 0.073 \\ 
				& Feasible $\wh{m}$  & 94.4 & 94.5 & 95.9 & 0.318 & 0.318 & 0.320 & 0.074 & 0.074 & 0.073\vspace{4pt} \\ 
				
				\multicolumn{10}{l}{\textit{Heteroskedastic errors} } \\	
				
				\multirow{2}{*}{$\eta=0.2$} & 		Infeasible $\widetilde{m}$ & 92.1 & 92.7 & 96.3 & 0.382 & 0.384 & 0.379 & 0.100 & 0.100 & 0.100 \\ 
				&		Feasible $\wh{m}$ & 92.5 & 93.0 & 96.1 & 0.375 & 0.380 & 0.385 & 0.101 & 0.101 & 0.099\vspace{4pt} \\ 
				\multirow{2}{*}{$\eta=0.5$} &		Infeasible $\widetilde{m}$  & 93.4 & 93.8 & 96.2 & 0.341 & 0.344 & 0.341 & 0.081 & 0.081 & 0.081 \\ 
				&		Feasible $\wh{m}$  & 93.6 & 94.0 & 96.0 & 0.337 & 0.342 & 0.344 & 0.081 & 0.081 & 0.080\vspace{4pt} \\ 
				\multirow{2}{*}{$\eta=0.8$} & 	Infeasible $\widetilde{m}$  & 94.4 & 94.6 & 95.8 & 0.325 & 0.328 & 0.326 & 0.074 & 0.074 & 0.074 \\ 
				&		Feasible $\wh{m}$  & 94.4 & 94.6 & 95.8 & 0.323 & 0.327 & 0.328 & 0.074 & 0.074 & 0.074 \\ 
				\hline
			\end{tabular}
			\begin{tablenotes}\small
				\item\hspace{-3pt}\textit{Notes:}
				The estimators are evaluated with their respective RMSE-optimal bandwidths based on the true smoothness constant. The sample size is $n=1,000$, and the number of simulations is $S=10,000$.
			\end{tablenotes}	
	\end{threeparttable}}
	\label{table:sim2}
\end{table}

In Table \ref{table:sim2}, I present results regarding the bandwidth choice, empirical coverage, and length of 95\% CIs. Here, I also use the true smoothness constant ($M=2$). The bandwidth selector for the feasible estimator chooses virtually the same bandwidth as would be chosen for the infeasible estimator $\wt m$, and the coverage is nearly identical.
I note that even for the infeasible estimator, the CI based on the true smoothness constant can have coverage below the nominal confidence level despite correctly accounting for maximal bias.
The reason for that is that although $Y$ is conditionally normally distributed, the outcome variable $\psi(\eta,Q(\eta,X))$ is not. The non-normality is more pronounced for lower truncation quantile levels.
In Online Appendix~\ref{A:ROT}, I~discuss a rule of thumb for choosing the smoothness constant that performs well in this simulation setting.

\newpage
\section{Applications}\label{sec:Applications}
In this section, the theoretical results from Section~\ref{sec:Asymptotics} are applied to two empirical settings: (i)~sharp regression discontinuity designs with a manipulated running variable and (ii) randomized experiments with sample selection. 
Both applications involve evaluation of truncated conditional expectations at estimated truncation quantile levels that are given by the ratio of two densities and two conditional expectations, respectively.

\subsection{Sharp Regression Discontinuity Designs with Manipulation}\label{sec:RDD}
\citet{gerard2020bounds} study regression discontinuity (RD) designs with a manipulated running variable. They develop a complex estimation approach applicable to fuzzy RD designs, which encompass sharp RD designs as a special case. Their estimation routine involves numerical integration and   inference is based on a bootstrap procedure. I study an approach tailored specifically to sharp RD designs that is easier to implement.

\subsubsection{Partial Identification under Manipulation} In a sharp RD design, units receive a treatment if and only if a special covariate, the running variable, exceeds a fixed cutoff value.
If the distribution of units' potential outcomes varies smoothly with the running variable around the cutoff, then the (local to the cutoff) average treatment effect is identified by the difference in average outcomes of the treated and untreated units whose realization of the running variable is just to the right or just to the left of the cutoff, respectively \citep[see, e.g.,][]{lee2010regression}.
The key identifying assumption,  however, is often questionable if the running variable is not exogenously determined.

To allow for violations of the smoothness assumption, 
\citet{gerard2020bounds} develop a framework where there are two unobservable types of units: \textit{always-assigned} units, for which the realization of the running variable is always to the right of the cutoff, and hence they are assigned the treatment; and \textit{potentially-assigned} units, whose density of the running variable is smooth around the cutoff, and hence they satisfy the standard assumptions of an RD design. 
\citet{gerard2020bounds} show that the average treatment effect for the subpopulation of potentially-assigned units at the cutoff, denoted by $\Gamma$, is partially identified. The bounds are derived as follows.
Under their behavioral model, the share of always-assigned units among the units just to the right of the cutoff, denoted by $\tau$, is identified by the discontinuity in the density of the running variable, denoted by $f_X$, at the cutoff as
\begin{equation*}
\tau = 1- \frac{f_X(x_0^-)}{f_X(x_0^+)},
\end{equation*}
where $x_0$ is the cutoff value.\footnote{For a generic function $g(\cdot)$, I put $g(x_0^+)=\lim_{x\to x_0^+}g(x)$ and $g(x_0^-)=\lim_{x\to x_0^-}g(x)$.} Given $\tau$, the sharp bounds on $\Gamma$ are obtained by considering the `extreme' scenarios in which the always-assigned units have the lowest or the highest outcomes among the units just to the right of the cutoff. The resulting lower and upper bound are given by:
\begin{align*}
& \Gamma^L = \E[Y|X=x_0^+, Y\leq Q(1-\tau,x^+_0)] - \E[Y|X=x_0^-], \\
& \Gamma^U = \E[Y|X=x_0^+, Y\ge Q(\tau,x_0^+)] - \E[Y|X=x_0^-].
\end{align*}

\subsubsection{Estimation and Inference} I discuss the main ingredients of the bounds estimator and its asymptotic properties. The details are given in Appendix \ref{sec:Appendix_inference}.
The bounds $\Gamma^L$ and $\Gamma^U$ involve truncated conditional expectation functions, which I estimate using the estimator $\widehat{m}$ developed in this paper.\footnote{Estimation with truncation from below can be performed using the procedure developed for estimation with truncation from above by taking the negative of the estimator applied to the data $\{X_i, -Y_i\}^{n}_{i=1}$.} Since $\tau$ is the proportion of truncated data, the quantile level $\eta$ in the previous sections corresponds to $1-\tau$, i.e. $\eta$ is the proportion of potentially-assigned units just to the right of the cutoff.
The first step is to estimate $\tau$.  The density limits can be estimated using estimators such as the linear smoother of the histogram 
\citep{cheng1997boundary, mccrary2008manipulation}, the linear smoother of the empirical density function \citep{jones1993simple, LejeuneSarda1992}, or the local quadratic smoother of the empirical distribution function \citep{cattaneo2020simple}.

Under regularity conditions, the resulting estimator of the truncation quantile level, $\widehat{\eta}=1-\widehat{\tau}$, satisfies the high-level assumption of Theorem~\ref{th:esteta}. Moreover, since $\widehat{\eta}$ depends only on the running variable, it is conditionally uncorrelated with the estimators of the truncated conditional expectations with known $\eta$, which simplifies the asymptotic variance formula. The conditional expectation just to the left of the cutoff, $\E[Y|X=x_0^-]$, can be estimated using a~standard local linear estimator. The estimators of the bounds have an asymptotically normal distribution, which can be used to form confidence intervals for the partially identified treatment effect.

\subsubsection{Empirical Application} I evaluate the procedure that I propose by implementing it for the empirical application of \citet{gerard2020bounds}.\footnote{The authors kindly implemented my procedure on their restricted-use data for comparison purposes.}
They investigate the effect of unemployment insurance (UI) benefits on the formal reemployment in Brazil. They exploit the rule that a worker involuntarily laid off from a private-sector firm is eligible for the UI benefit only if there was at least 16 months between the date of her layoff and the date of the last layoff after which she applied for and drew UI benefits. This rule creates a discontinuity in the eligibility for UI benefits, which is reflected in a 70pp increase in the actual take-up of UI benefits. In the following, I focus on an intention-to-treat analysis, where the eligibility for UI benefits is the treatment, and the outcome of interest is the duration without a formal job after the layoff.

\begin{figure}[hbt]
	\captionsetup[subfigure]{justification=centering}
	\begin{subfigure}{0.48\textwidth}	
		\centering
		\includegraphics[width=1\textwidth]{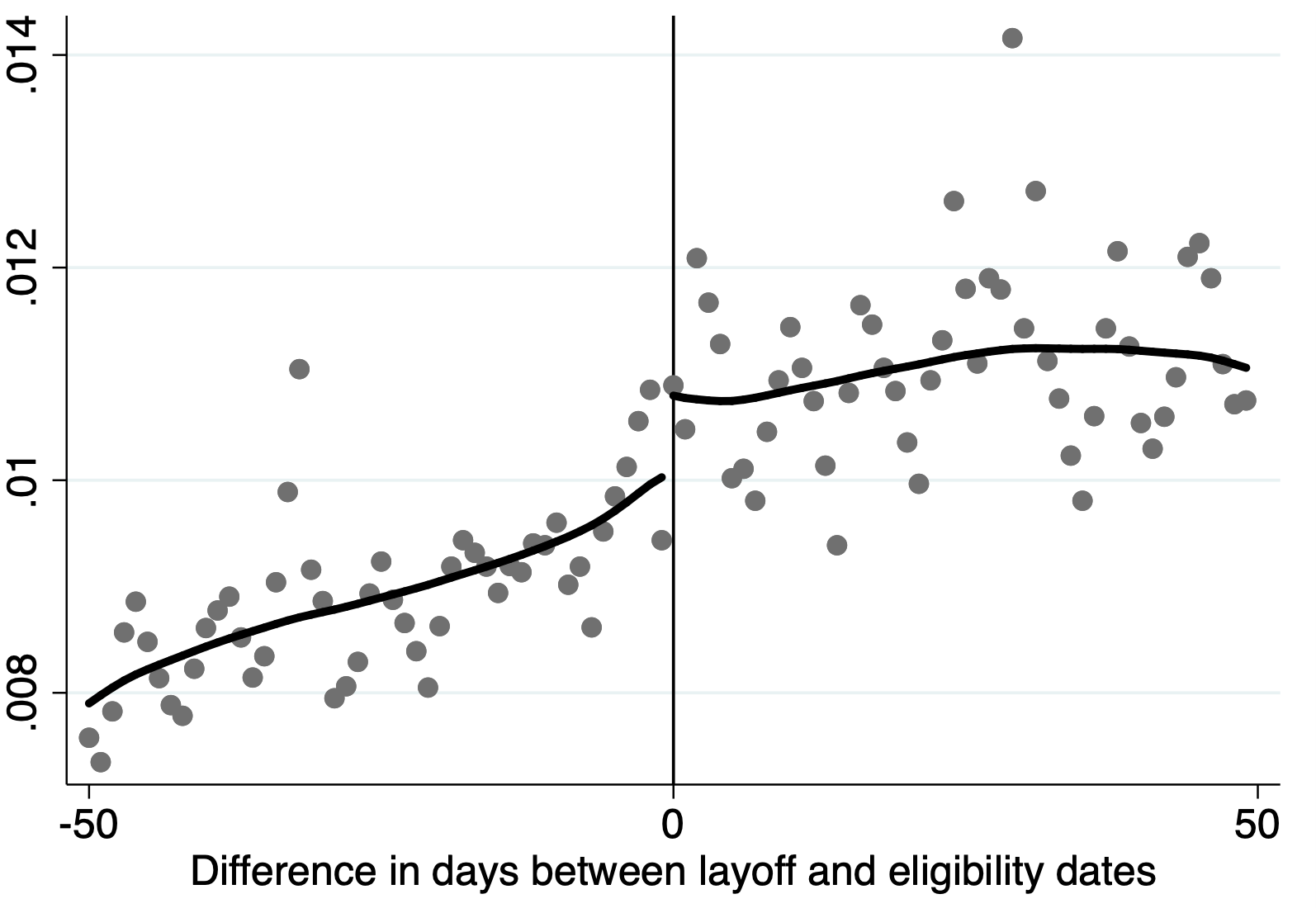}
		\vspace{-4pt}
		\subcaption{Frequency.}
	\end{subfigure}
	\begin{subfigure}{0.48\textwidth}
		\centering	
		\includegraphics[width=1\textwidth]{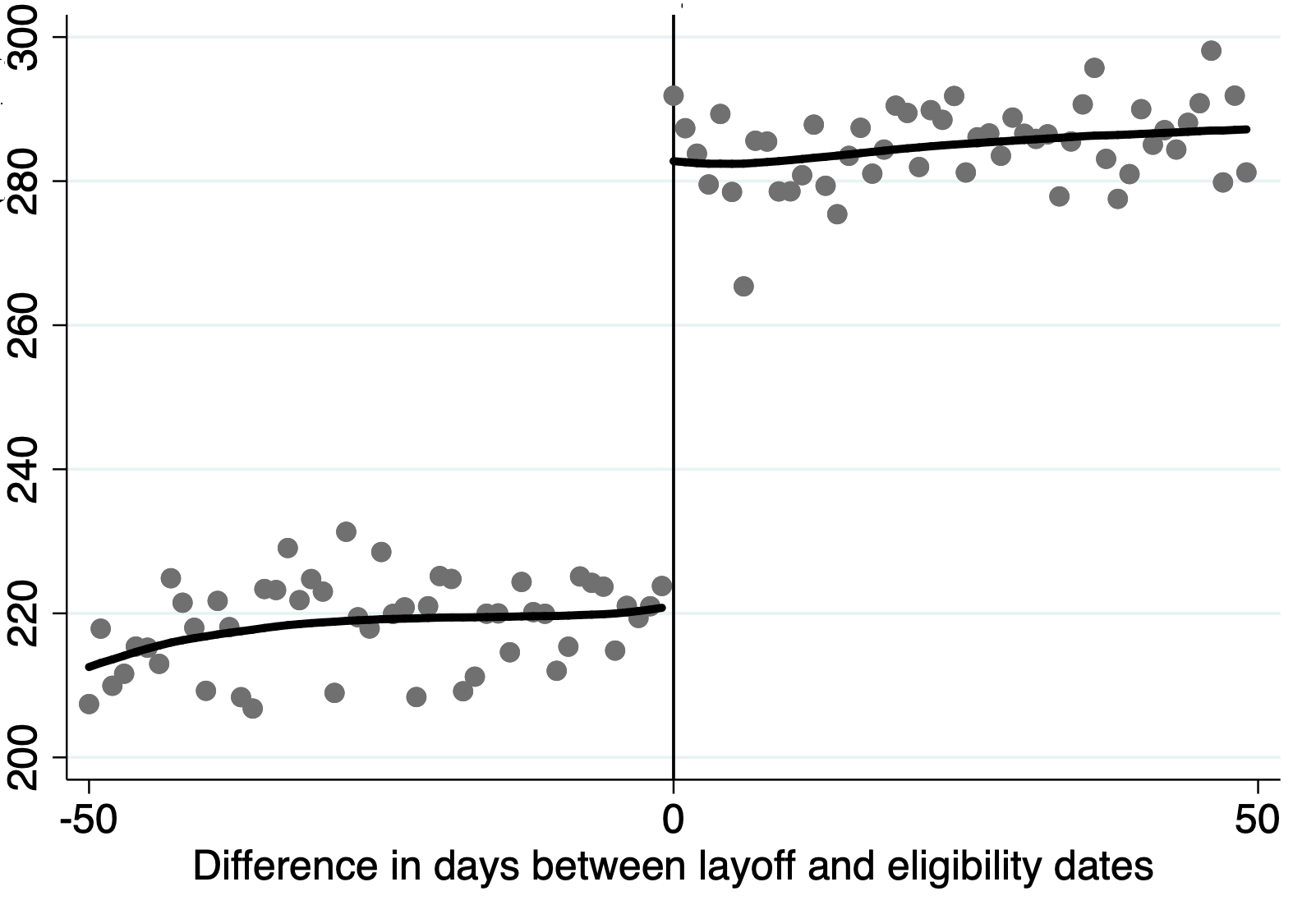}
		\vspace{-4pt}
		\subcaption{Duration without a formal job.}
	\end{subfigure}
	\caption{Graphical evidence for the intention-to-treat analysis.}
	\footnotesize{\textit{Notes:}  The dots represent the frequency (left panel) and the average duration of unemployment censored at 24 months (right panel) by day. The figure is based on 169,575 observations. \textit{Source: \citet{gerard2020bounds}}.}
	\label{fig:evidence}
\end{figure}

Despite the 16-month rule being rather arbitrary, \citet{gerard2020bounds} point out the following ways in which violations of the standard RD assumptions may arise in this setup. Some workers may provoke their layoffs or ask their employers to report their quit as involuntary once they become eligible for a UI benefit. Other workers may have managed to delay their layoff to a~date when they were eligible for the UI benefit. All theses workers are always-assigned units in the manipulation framework outlined in the previous subsection.

Figure \ref{fig:evidence} reproduces the graphical evidence for this RD design. The running variable is the difference in days between the layoff date and the eligibility date, so that the cutoff is at 0. In the left panel, I present the density of the running variable. The share of always-assigned units is estimated to be 6.4\%, which is relatively well separated from zero. This is essential for the good quality of the normal approximation of the asymptotic distribution of $\widehat{\tau}$. In the right panel, the dots represent the average outcome by day (of all observations). There is a marked jump in the mean duration without a formal job at the cutoff. I note that a substantial share, about 12--14\%, of duration outcomes is censored at 24 months. This, however, does not require any adjustment in my estimation and inference procedure.

Following \citet{gerard2020bounds}, I conduct two types of analysis. First, I estimate bounds on $\Gamma$ using an estimated proportion of the always-assigned units to the right of the cutoff. Second, I~conduct a sensitivity analysis, where I report bounds for different levels of potential manipulation. I report my results along with the original estimates of \citet{gerard2020bounds}. Their estimator is based on a local linear estimator of the conditional c.d.f., and they conduct inference via bootstrap. For comparability with \citet{gerard2020bounds}, all estimators use a 30-day bandwidth, and the confidence intervals are formally justified by undersmoothing. 

In Table~\ref{table:results}, I present estimates of the bounds and the 95\% confidence intervals for $\Gamma$ with estimated $\tau$. As a reference point, the point estimate ignoring the possibility of manipulation indicates that the eligibility for UI benefits increases the duration of unemployment by about 62 days. When accounting for manipulation, however, the estimated identified set spans the range from 31 to 81 days.

\begin{table}[th]
	\centering
	\caption[Estimated effects of UI benefits on the duration without a formal job.]{Estimated effects of UI benefits on the duration without a formal job in days.}
	\scalebox{0.95}{\begin{threeparttable}
			\begin{tabular}{lcc@{\hskip 0.3in}cccc}
				\hline
				& \multicolumn{2}{c}{Results of \citet{gerard2020bounds}} & \multicolumn{2}{c}{My results} \\
				& Estimate & 95\% CI & Estimate & 95\% CI  \\ \cline{2-3} \cline{4-5}
				Share of always-assigned units & 0.064 & [0.038; 0.089]  \\  [4pt]
				LATE: Ignoring manipulation & 61.9 & [55.7; 68.1] & 61.9 & [55.5; 68.3] \\ 
				LATE: Bounds for $\Gamma$ & [31.4; 80.9] & [18.9; 89.6]  & [31.4; 80.9] & [19.4; 89.5] \\   \hline
			\end{tabular}
			\begin{tablenotes}
				\small
				\item\hspace{-3pt}\textit{Note:} There are 102,791 observations in the 30-day estimation window.
			\end{tablenotes}
	\end{threeparttable}}
	\label{table:results}
\end{table}

In the second part of the analysis, a certain hypothetical, fixed degree of manipulation in the data is presumed. The results are presented in Figure~\ref{fig:fixed_results}. At $\tau=0$, the results correspond to the `no manipulation case', where the treatment effect is point identified. The bounds become wider as the presumed degree of manipulation increases.
 The vertical black line marks the estimated proportion of always-assigned units among all units just to the right of the cutoff.

The results are nearly identical when using the procedure of \citet{gerard2020bounds} and mine. This similarity, however, is specific to this dataset, where the conditional quantile functions at the truncation quantile levels are flat. I show in Appendix \ref{A:WNW} that compared to my estimator, approaches based on first-stage estimates of the conditional c.d.f. have an additional bias term when the conditional quantile function has a nonzero slope.

\begin{figure}[htb]
	\captionsetup[subfigure]{justification=centering}
	\begin{subfigure}{0.48\textwidth}	
		\centering
		\scalebox{0.95}{
\begin{tikzpicture}[x=1pt,y=1pt]
\definecolor{fillColor}{RGB}{255,255,255}
\path[use as bounding box,fill=fillColor,fill opacity=0.00] (0,0) rectangle (231.26,187.90);
\begin{scope}
\path[clip] (  0.00,  0.00) rectangle (231.26,187.90);
\definecolor{drawColor}{RGB}{255,255,255}
\definecolor{fillColor}{RGB}{255,255,255}

\path[draw=drawColor,line width= 0.6pt,line join=round,line cap=round,fill=fillColor] (  0.00,  0.00) rectangle (231.26,187.90);
\end{scope}
\begin{scope}
\path[clip] ( 58.00, 30.69) rectangle (225.76,182.40);
\definecolor{fillColor}{RGB}{255,255,255}

\path[fill=fillColor] ( 58.00, 30.69) rectangle (225.76,182.40);
\definecolor{drawColor}{gray}{0.87}

\path[draw=drawColor,line width= 0.1pt,line join=round] ( 58.00, 60.57) --
	(225.76, 60.57);

\path[draw=drawColor,line width= 0.1pt,line join=round] ( 58.00,106.54) --
	(225.76,106.54);

\path[draw=drawColor,line width= 0.1pt,line join=round] ( 58.00,152.52) --
	(225.76,152.52);

\path[draw=drawColor,line width= 0.3pt,line join=round] ( 58.00, 37.58) --
	(225.76, 37.58);

\path[draw=drawColor,line width= 0.3pt,line join=round] ( 58.00, 83.56) --
	(225.76, 83.56);

\path[draw=drawColor,line width= 0.3pt,line join=round] ( 58.00,129.53) --
	(225.76,129.53);

\path[draw=drawColor,line width= 0.3pt,line join=round] ( 58.00,175.51) --
	(225.76,175.51);

\path[draw=drawColor,line width= 0.3pt,line join=round] ( 65.63, 30.69) --
	( 65.63,182.40);

\path[draw=drawColor,line width= 0.3pt,line join=round] (114.43, 30.69) --
	(114.43,182.40);

\path[draw=drawColor,line width= 0.3pt,line join=round] (141.88, 30.69) --
	(141.88,182.40);

\path[draw=drawColor,line width= 0.3pt,line join=round] (218.14, 30.69) --
	(218.14,182.40);
\definecolor{drawColor}{RGB}{0,0,0}

\path[draw=drawColor,line width= 1.1pt,dash pattern=on 4pt off 4pt ,line join=round] ( 65.63,109.09) --
	( 84.69,104.25) --
	(103.75, 98.66) --
	(114.19, 95.47) --
	(141.88, 86.55) --
	(218.14, 61.58);

\path[draw=drawColor,line width= 1.1pt,line join=round] ( 65.63,112.01) --
	( 84.69,106.73) --
	(103.75,101.17) --
	(114.19, 98.00) --
	(141.88, 89.12) --
	(218.14, 63.99);

\path[draw=drawColor,line width= 1.1pt,line join=round] ( 65.63,112.01) --
	( 84.69,115.34) --
	(103.75,118.81) --
	(114.19,120.77) --
	(141.88,126.11) --
	(218.14,141.77);

\path[draw=drawColor,line width= 1.1pt,dash pattern=on 4pt off 4pt ,line join=round] ( 65.63,114.94) --
	( 84.69,117.82) --
	(103.75,121.32) --
	(114.19,123.29) --
	(141.88,128.67) --
	(218.14,144.43);

\path[draw=drawColor,line width= 0.6pt,line join=round] (114.43, 30.69) -- (114.43,182.40);
\definecolor{drawColor}{gray}{0.70}

\path[draw=drawColor,line width= 0.6pt,line join=round,line cap=round] ( 58.00, 30.69) rectangle (225.76,182.40);
\end{scope}
\begin{scope}
\path[clip] (  0.00,  0.00) rectangle (231.26,187.90);
\definecolor{drawColor}{gray}{0.30}

\node[text=drawColor,anchor=base east,inner sep=0pt, outer sep=0pt, scale=  0.73] at ( 53.05, 34.55) {-100};

\node[text=drawColor,anchor=base east,inner sep=0pt, outer sep=0pt, scale=  0.73] at ( 53.05, 80.53) {0};

\node[text=drawColor,anchor=base east,inner sep=0pt, outer sep=0pt, scale=  0.73] at ( 53.05,126.50) {100};

\node[text=drawColor,anchor=base east,inner sep=0pt, outer sep=0pt, scale=  0.73] at ( 53.05,172.48) {200};
\end{scope}
\begin{scope}
\path[clip] (  0.00,  0.00) rectangle (231.26,187.90);
\definecolor{drawColor}{gray}{0.70}

\path[draw=drawColor,line width= 0.3pt,line join=round] ( 55.25, 37.58) --
	( 58.00, 37.58);

\path[draw=drawColor,line width= 0.3pt,line join=round] ( 55.25, 83.56) --
	( 58.00, 83.56);

\path[draw=drawColor,line width= 0.3pt,line join=round] ( 55.25,129.53) --
	( 58.00,129.53);

\path[draw=drawColor,line width= 0.3pt,line join=round] ( 55.25,175.51) --
	( 58.00,175.51);
\end{scope}
\begin{scope}
\path[clip] (  0.00,  0.00) rectangle (231.26,187.90);
\definecolor{drawColor}{gray}{0.70}

\path[draw=drawColor,line width= 0.3pt,line join=round] ( 65.63, 27.94) --
	( 65.63, 30.69);

\path[draw=drawColor,line width= 0.3pt,line join=round] (114.43, 27.94) --
	(114.43, 30.69);

\path[draw=drawColor,line width= 0.3pt,line join=round] (141.88, 27.94) --
	(141.88, 30.69);

\path[draw=drawColor,line width= 0.3pt,line join=round] (218.14, 27.94) --
	(218.14, 30.69);
\end{scope}
\begin{scope}
\path[clip] (  0.00,  0.00) rectangle (231.26,187.90);
\definecolor{drawColor}{gray}{0.30}

\node[text=drawColor,anchor=base,inner sep=0pt, outer sep=0pt, scale=  0.73] at ( 65.63, 19.68) {0};

\node[text=drawColor,anchor=base,inner sep=0pt, outer sep=0pt, scale=  0.73] at (114.43, 19.68) {$\hat{\tau}=0.064$};

\node[text=drawColor,anchor=base,inner sep=0pt, outer sep=0pt, scale=  0.73] at (141.88, 19.68) {0.1};

\node[text=drawColor,anchor=base,inner sep=0pt, outer sep=0pt, scale=  0.73] at (218.14, 19.68) {0.2};
\end{scope}
\begin{scope}
\path[clip] (  0.00,  0.00) rectangle (231.26,187.90);
\definecolor{drawColor}{RGB}{0,0,0}

\node[text=drawColor,anchor=base,inner sep=0pt, outer sep=0pt, scale=  0.92] at (141.88,  7.64) {Hypothetical $\tau$};
\end{scope}
\begin{scope}
\path[clip] (  0.00,  0.00) rectangle (231.26,187.90);
\definecolor{drawColor}{RGB}{0,0,0}

\node[text=drawColor,anchor=base west,inner sep=0pt, outer sep=0pt, scale=  0.92] at (  5.50,102.76) {LATE};
\end{scope}
\end{tikzpicture}}
		\vspace{-4pt}
		\subcaption{Procedure of \citet{gerard2020bounds}.}
	\end{subfigure}
	\begin{subfigure}{0.48\textwidth}
		\centering	
		\scalebox{0.95}{
\begin{tikzpicture}[x=1pt,y=1pt]
\definecolor{fillColor}{RGB}{255,255,255}
\path[use as bounding box,fill=fillColor,fill opacity=0.00] (0,0) rectangle (231.26,187.90);
\begin{scope}
\path[clip] (  0.00,  0.00) rectangle (231.26,187.90);
\definecolor{drawColor}{RGB}{255,255,255}
\definecolor{fillColor}{RGB}{255,255,255}

\path[draw=drawColor,line width= 0.6pt,line join=round,line cap=round,fill=fillColor] (  0.00,  0.00) rectangle (231.26,187.90);
\end{scope}
\begin{scope}
\path[clip] ( 58.00, 30.69) rectangle (225.76,182.40);
\definecolor{fillColor}{RGB}{255,255,255}

\path[fill=fillColor] ( 58.00, 30.69) rectangle (225.76,182.40);
\definecolor{drawColor}{gray}{0.87}

\path[draw=drawColor,line width= 0.1pt,line join=round] ( 58.00, 60.57) --
	(225.76, 60.57);

\path[draw=drawColor,line width= 0.1pt,line join=round] ( 58.00,106.54) --
	(225.76,106.54);

\path[draw=drawColor,line width= 0.1pt,line join=round] ( 58.00,152.52) --
	(225.76,152.52);

\path[draw=drawColor,line width= 0.3pt,line join=round] ( 58.00, 37.58) --
	(225.76, 37.58);

\path[draw=drawColor,line width= 0.3pt,line join=round] ( 58.00, 83.56) --
	(225.76, 83.56);

\path[draw=drawColor,line width= 0.3pt,line join=round] ( 58.00,129.53) --
	(225.76,129.53);

\path[draw=drawColor,line width= 0.3pt,line join=round] ( 58.00,175.51) --
	(225.76,175.51);

\path[draw=drawColor,line width= 0.3pt,line join=round] ( 65.63, 30.69) --
	( 65.63,182.40);

\path[draw=drawColor,line width= 0.3pt,line join=round] (114.43, 30.69) --
	(114.43,182.40);

\path[draw=drawColor,line width= 0.3pt,line join=round] (141.88, 30.69) --
	(141.88,182.40);

\path[draw=drawColor,line width= 0.3pt,line join=round] (218.14, 30.69) --
	(218.14,182.40);
\definecolor{drawColor}{RGB}{0,0,0}

\path[draw=drawColor,line width= 1.1pt,dash pattern=on 4pt off 4pt ,line join=round] ( 65.63,109.05) --
	( 84.69,104.22) --
	(103.75, 98.63) --
	(114.19, 95.44) --
	(141.88, 86.51) --
	(218.14, 61.57);

\path[draw=drawColor,line width= 1.1pt,line join=round] ( 65.63,112.01) --
	( 84.69,106.73) --
	(103.75,101.17) --
	(114.19, 98.00) --
	(141.88, 89.12) --
	(218.14, 63.99);

\path[draw=drawColor,line width= 1.1pt,line join=round] ( 65.63,112.01) --
	( 84.69,115.34) --
	(103.75,118.81) --
	(114.19,120.77) --
	(141.88,126.11) --
	(218.14,141.77);

\path[draw=drawColor,line width= 1.1pt,dash pattern=on 4pt off 4pt ,line join=round] ( 65.63,114.97) --
	( 84.69,117.85) --
	(103.75,121.35) --
	(114.19,123.32) --
	(141.88,128.71) --
	(218.14,144.48);

\path[draw=drawColor,line width= 0.6pt,line join=round] (114.43, 30.69) -- (114.43,182.40);
\definecolor{drawColor}{gray}{0.70}

\path[draw=drawColor,line width= 0.6pt,line join=round,line cap=round] ( 58.00, 30.69) rectangle (225.76,182.40);
\end{scope}
\begin{scope}
\path[clip] (  0.00,  0.00) rectangle (231.26,187.90);
\definecolor{drawColor}{gray}{0.30}

\node[text=drawColor,anchor=base east,inner sep=0pt, outer sep=0pt, scale=  0.73] at ( 53.05, 34.55) {-100};

\node[text=drawColor,anchor=base east,inner sep=0pt, outer sep=0pt, scale=  0.73] at ( 53.05, 80.53) {0};

\node[text=drawColor,anchor=base east,inner sep=0pt, outer sep=0pt, scale=  0.73] at ( 53.05,126.50) {100};

\node[text=drawColor,anchor=base east,inner sep=0pt, outer sep=0pt, scale=  0.73] at ( 53.05,172.48) {200};
\end{scope}
\begin{scope}
\path[clip] (  0.00,  0.00) rectangle (231.26,187.90);
\definecolor{drawColor}{gray}{0.70}

\path[draw=drawColor,line width= 0.3pt,line join=round] ( 55.25, 37.58) --
	( 58.00, 37.58);

\path[draw=drawColor,line width= 0.3pt,line join=round] ( 55.25, 83.56) --
	( 58.00, 83.56);

\path[draw=drawColor,line width= 0.3pt,line join=round] ( 55.25,129.53) --
	( 58.00,129.53);

\path[draw=drawColor,line width= 0.3pt,line join=round] ( 55.25,175.51) --
	( 58.00,175.51);
\end{scope}
\begin{scope}
\path[clip] (  0.00,  0.00) rectangle (231.26,187.90);
\definecolor{drawColor}{gray}{0.70}

\path[draw=drawColor,line width= 0.3pt,line join=round] ( 65.63, 27.94) --
	( 65.63, 30.69);

\path[draw=drawColor,line width= 0.3pt,line join=round] (114.43, 27.94) --
	(114.43, 30.69);

\path[draw=drawColor,line width= 0.3pt,line join=round] (141.88, 27.94) --
	(141.88, 30.69);

\path[draw=drawColor,line width= 0.3pt,line join=round] (218.14, 27.94) --
	(218.14, 30.69);
\end{scope}
\begin{scope}
\path[clip] (  0.00,  0.00) rectangle (231.26,187.90);
\definecolor{drawColor}{gray}{0.30}

\node[text=drawColor,anchor=base,inner sep=0pt, outer sep=0pt, scale=  0.73] at ( 65.63, 19.68) {0};

\node[text=drawColor,anchor=base,inner sep=0pt, outer sep=0pt, scale=  0.73] at (114.43, 19.68) {$\hat{\tau}=0.064$};

\node[text=drawColor,anchor=base,inner sep=0pt, outer sep=0pt, scale=  0.73] at (141.88, 19.68) {0.1};

\node[text=drawColor,anchor=base,inner sep=0pt, outer sep=0pt, scale=  0.73] at (218.14, 19.68) {0.2};
\end{scope}
\begin{scope}
\path[clip] (  0.00,  0.00) rectangle (231.26,187.90);
\definecolor{drawColor}{RGB}{0,0,0}

\node[text=drawColor,anchor=base,inner sep=0pt, outer sep=0pt, scale=  0.92] at (141.88,  7.64) {Hypothetical $\tau$};
\end{scope}
\begin{scope}
\path[clip] (  0.00,  0.00) rectangle (231.26,187.90);
\definecolor{drawColor}{RGB}{0,0,0}

\node[text=drawColor,anchor=base west,inner sep=0pt, outer sep=0pt, scale=  0.92] at (  5.50,102.76) {LATE};
\end{scope}
\end{tikzpicture}}
		\vspace{-4pt}
		\subcaption{Estimation with $\wh{m}$.}
	\end{subfigure}
	\caption{Fixed-manipulation inference.}
	\footnotesize{\textit{Notes:} The horizontal axis displays the hypothetical proportion	of potentially-assigned workers. The solid lines present the estimates of the bounds and the	dashed lines mark 95\% confidence intervals. The figures are based on 102,791 observations.}
	\label{fig:fixed_results}
\end{figure}
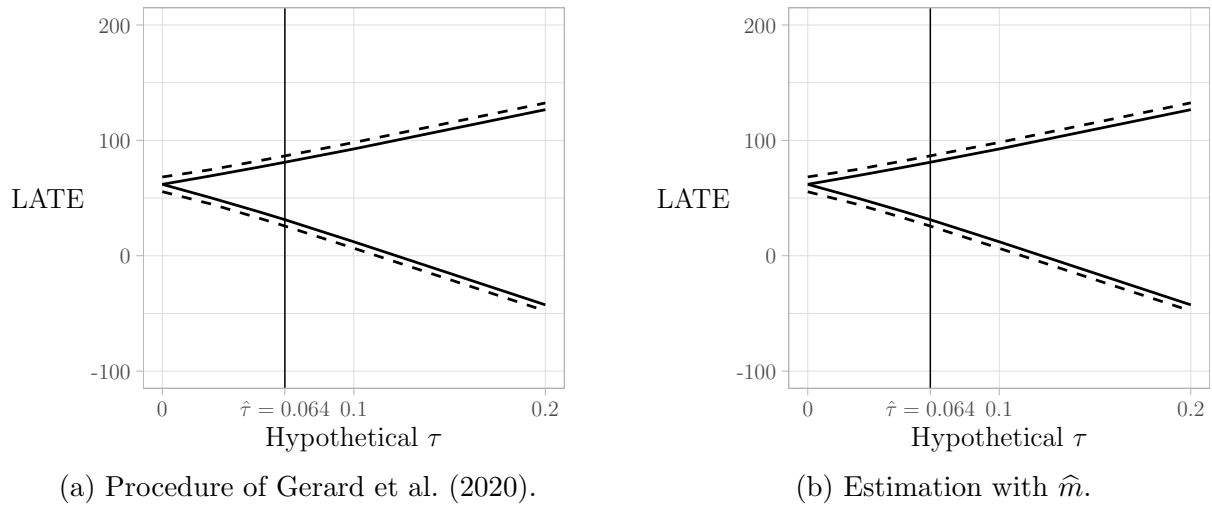

\subsection{Conditional Lee Bounds}\label{sec:Lee}
\citet{Lee2009} studies the effect of a job training program on wage rates. In his analysis, he uses covariates to narrow down the bounds on the unconditional effect \citep[see also][]{semenova2020better}. The conditional treatment effects, however, may be of interest in their own right.

\subsubsection{Partial Identification of the Wage Effect}
In a randomized experiment, units are randomly split into the treatment and control groups. Treatment effects are then typically identified by comparing units in the two treatment arms. In some settings, however, such comparisons are invalid due to sample selection. For example, job training affects not only the wage rates but also the employment status. As a result, individuals whose wages are observed in the treatment and control groups are not comparable even if the treatment was random assigned.

For such settings, \citet{Lee2009} derives bounds on the wage effect for the subpopulation of \textit{always-observed} individuals, i.e. those who would work regardless of whether they obtained the treatment. In the first step, he identifies the proportion of individuals whose employment status is affected by the treatment status. By random assignment to the program, this proportion is given by the difference in the employment rates in the treatment and control group. If the training program weakly encourages to work, then the bounds on the wage rates of the always-observed in the treatment group are obtained by considering the extreme scenarios in which the always-observed individuals have the highest or the lowest wage rates among the employed.\footnote{If the treatment discourages from working, then the control group would need to be truncated.} This reasoning holds unconditionally as well as conditionally on covariates. 

To state these bounds formally, let $D$ be the treatment indicator and $S$ the employment indicator. Further, let $X$ be some additional covariate. The conditional proportion of individuals among the employed in the treatment group who are employed if and only if they are treated is identified as 
\begin{equation*}
p(x)= 1 - \frac{\Prob[S=1|D=0,X=x]}{\Prob[S=1|D=1,X=x]}.
\end{equation*}
The sharp lower and upper bounds on the local average treatment effect on wage rates are given by \citep[Proposition 1b]{Lee2009}
\begin{align*}
\Delta^L(x) =  \E[Y|D=1, S=1, Y  \leq Q_{D=1,S=1}&(1-p(x),x), X=x] - \E[Y|D=0,S=1, X=x], \\
\Delta^U(x) =   \E[Y|D=1, S=1, Y  \ge Q_{D=1,S=1}&(p(x),x), X=x]  - \E[Y|D=0,S=1, X=x],
\end{align*}
where $Q_{D=1,S=1}(u,x)$ denotes the $u$-quantile of $Y$ conditional on $D=1$, $S=1$, and $X=x$.

\subsubsection{Estimation and Inference}
I discuss the main ingredients of the bounds estimator. The details are given in Appendix \ref{A:Lee}.
The conditional probabilities $\Prob[S=1|D=d,X=x]$ can be estimated using a local linear estimator with $S_i$ as the outcome and $X_i$ as a regressor, run on the sample restricted to observations with $D_i=d$ for $d \in \{0,1\}$.
Under regularity conditions, the resulting estimator $\wh{\eta}=1-\wh{p}(x_0)$ satisfies the high-level assumption of Theorem~\ref{th:esteta}. The truncated conditional expectations in the definition of $\Delta^L(x)$ and $\Delta^U(x)$  can be estimated using the estimator proposed in this paper and the conditional expectation function in the control group can be estimated using the standard local linear estimator. Restricting the samples based on the values of indicators $S_i$ and $D_i$ does not cause any complications in the asymptotic analysis.
The estimators of the bounds have an asymptotically normal distribution, which can be used to form confidence intervals for the partially identified treatment effect.

\subsubsection{Empirical Application}
I evaluate the effect of the job training offered under the Job Corps program in the United States. I use data from the National Job Corps Study conducted in mid 90s. 
I follow \citet{Lee2009} closely in terms of the sample definition.
The individuals who applied to the program were followed for four years after random assignment. There are 3599 individuals in the control group and 5546 in the treatment group, giving a total of 9145 observations.
I investigate the effect on wage rates four years after the random assignment, conditioning on the usual weekly earnings at the most recent job reported at the baseline.

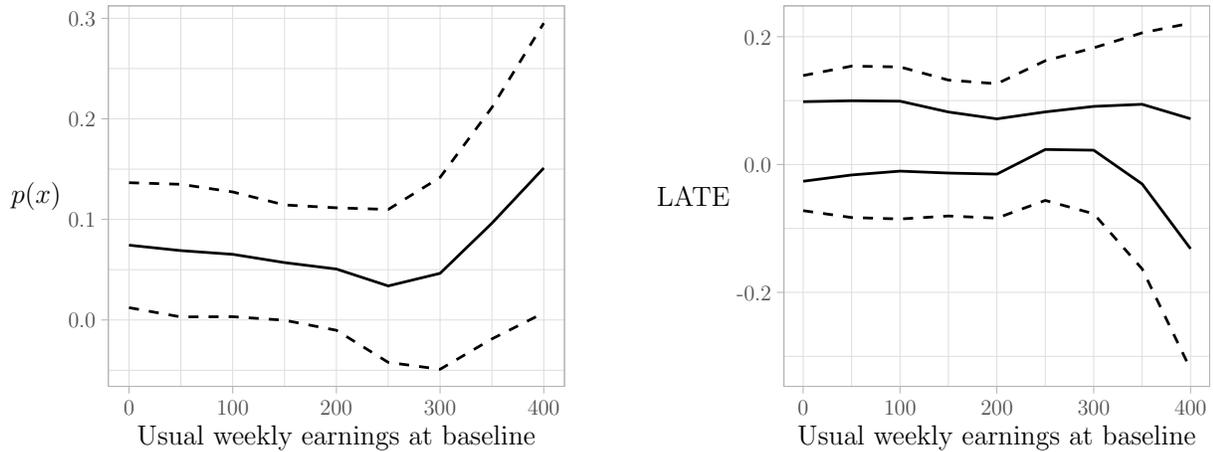
\begin{figure}[htb]
	\captionsetup[subfigure]{justification=centering}
	\begin{subfigure}{0.48\textwidth}	
		\centering
		\scalebox{0.95}{
\begin{tikzpicture}[x=1pt,y=1pt]
\definecolor{fillColor}{RGB}{255,255,255}
\path[use as bounding box,fill=fillColor,fill opacity=0.00] (0,0) rectangle (231.26,187.90);
\begin{scope}
\path[clip] (  0.00,  0.00) rectangle (231.26,187.90);
\definecolor{drawColor}{RGB}{255,255,255}
\definecolor{fillColor}{RGB}{255,255,255}

\path[draw=drawColor,line width= 0.6pt,line join=round,line cap=round,fill=fillColor] (  0.00,  0.00) rectangle (231.26,187.90);
\end{scope}
\begin{scope}
\path[clip] ( 44.11, 30.69) rectangle (225.76,182.40);
\definecolor{fillColor}{RGB}{255,255,255}

\path[fill=fillColor] ( 44.11, 30.69) rectangle (225.76,182.40);
\definecolor{drawColor}{gray}{0.87}

\path[draw=drawColor,line width= 0.1pt,line join=round] ( 44.11, 37.19) --
	(225.76, 37.19);

\path[draw=drawColor,line width= 0.1pt,line join=round] ( 44.11, 77.27) --
	(225.76, 77.27);

\path[draw=drawColor,line width= 0.1pt,line join=round] ( 44.11,117.34) --
	(225.76,117.34);

\path[draw=drawColor,line width= 0.1pt,line join=round] ( 44.11,157.41) --
	(225.76,157.41);

\path[draw=drawColor,line width= 0.1pt,line join=round] ( 73.01, 30.69) --
	( 73.01,182.40);

\path[draw=drawColor,line width= 0.1pt,line join=round] (114.30, 30.69) --
	(114.30,182.40);

\path[draw=drawColor,line width= 0.1pt,line join=round] (155.58, 30.69) --
	(155.58,182.40);

\path[draw=drawColor,line width= 0.1pt,line join=round] (196.86, 30.69) --
	(196.86,182.40);

\path[draw=drawColor,line width= 0.3pt,line join=round] ( 44.11, 57.23) --
	(225.76, 57.23);

\path[draw=drawColor,line width= 0.3pt,line join=round] ( 44.11, 97.30) --
	(225.76, 97.30);

\path[draw=drawColor,line width= 0.3pt,line join=round] ( 44.11,137.37) --
	(225.76,137.37);

\path[draw=drawColor,line width= 0.3pt,line join=round] ( 44.11,177.44) --
	(225.76,177.44);

\path[draw=drawColor,line width= 0.3pt,line join=round] ( 52.37, 30.69) --
	( 52.37,182.40);

\path[draw=drawColor,line width= 0.3pt,line join=round] ( 93.65, 30.69) --
	( 93.65,182.40);

\path[draw=drawColor,line width= 0.3pt,line join=round] (134.94, 30.69) --
	(134.94,182.40);

\path[draw=drawColor,line width= 0.3pt,line join=round] (176.22, 30.69) --
	(176.22,182.40);

\path[draw=drawColor,line width= 0.3pt,line join=round] (217.51, 30.69) --
	(217.51,182.40);
\definecolor{drawColor}{RGB}{0,0,0}

\path[draw=drawColor,line width= 1.1pt,line join=round] ( 52.37, 87.04) --
	( 73.01, 84.87) --
	( 93.65, 83.39) --
	(114.30, 80.09) --
	(134.94, 77.53) --
	(155.58, 70.80) --
	(176.22, 75.81) --
	(196.86, 95.79) --
	(217.51,117.78);

\path[draw=drawColor,line width= 1.1pt,dash pattern=on 4pt off 4pt ,line join=round] ( 52.37,111.93) --
	( 73.01,111.29) --
	( 93.65,108.25) --
	(114.30,103.06) --
	(134.94,101.93) --
	(155.58,101.27) --
	(176.22,114.04) --
	(196.86,141.91) --
	(217.51,175.51);

\path[draw=drawColor,line width= 1.1pt,dash pattern=on 4pt off 4pt ,line join=round] ( 52.37, 62.14) --
	( 73.01, 58.45) --
	( 93.65, 58.54) --
	(114.30, 57.12) --
	(134.94, 53.12) --
	(155.58, 40.34) --
	(176.22, 37.58) --
	(196.86, 49.68) --
	(217.51, 60.06);
\definecolor{drawColor}{gray}{0.70}

\path[draw=drawColor,line width= 0.6pt,line join=round,line cap=round] ( 44.11, 30.69) rectangle (225.76,182.40);
\end{scope}
\begin{scope}
\path[clip] (  0.00,  0.00) rectangle (231.26,187.90);
\definecolor{drawColor}{gray}{0.30}

\node[text=drawColor,anchor=base east,inner sep=0pt, outer sep=0pt, scale=  0.73] at ( 39.16, 54.20) {0.0};

\node[text=drawColor,anchor=base east,inner sep=0pt, outer sep=0pt, scale=  0.73] at ( 39.16, 94.27) {0.1};

\node[text=drawColor,anchor=base east,inner sep=0pt, outer sep=0pt, scale=  0.73] at ( 39.16,134.34) {0.2};

\node[text=drawColor,anchor=base east,inner sep=0pt, outer sep=0pt, scale=  0.73] at ( 39.16,174.41) {0.3};
\end{scope}
\begin{scope}
\path[clip] (  0.00,  0.00) rectangle (231.26,187.90);
\definecolor{drawColor}{gray}{0.70}

\path[draw=drawColor,line width= 0.3pt,line join=round] ( 41.36, 57.23) --
	( 44.11, 57.23);

\path[draw=drawColor,line width= 0.3pt,line join=round] ( 41.36, 97.30) --
	( 44.11, 97.30);

\path[draw=drawColor,line width= 0.3pt,line join=round] ( 41.36,137.37) --
	( 44.11,137.37);

\path[draw=drawColor,line width= 0.3pt,line join=round] ( 41.36,177.44) --
	( 44.11,177.44);
\end{scope}
\begin{scope}
\path[clip] (  0.00,  0.00) rectangle (231.26,187.90);
\definecolor{drawColor}{gray}{0.70}

\path[draw=drawColor,line width= 0.3pt,line join=round] ( 52.37, 27.94) --
	( 52.37, 30.69);

\path[draw=drawColor,line width= 0.3pt,line join=round] ( 93.65, 27.94) --
	( 93.65, 30.69);

\path[draw=drawColor,line width= 0.3pt,line join=round] (134.94, 27.94) --
	(134.94, 30.69);

\path[draw=drawColor,line width= 0.3pt,line join=round] (176.22, 27.94) --
	(176.22, 30.69);

\path[draw=drawColor,line width= 0.3pt,line join=round] (217.51, 27.94) --
	(217.51, 30.69);
\end{scope}
\begin{scope}
\path[clip] (  0.00,  0.00) rectangle (231.26,187.90);
\definecolor{drawColor}{gray}{0.30}

\node[text=drawColor,anchor=base,inner sep=0pt, outer sep=0pt, scale=  0.73] at ( 52.37, 19.68) {0};

\node[text=drawColor,anchor=base,inner sep=0pt, outer sep=0pt, scale=  0.73] at ( 93.65, 19.68) {100};

\node[text=drawColor,anchor=base,inner sep=0pt, outer sep=0pt, scale=  0.73] at (134.94, 19.68) {200};

\node[text=drawColor,anchor=base,inner sep=0pt, outer sep=0pt, scale=  0.73] at (176.22, 19.68) {300};

\node[text=drawColor,anchor=base,inner sep=0pt, outer sep=0pt, scale=  0.73] at (217.51, 19.68) {400};
\end{scope}
\begin{scope}
\path[clip] (  0.00,  0.00) rectangle (231.26,187.90);
\definecolor{drawColor}{RGB}{0,0,0}

\node[text=drawColor,anchor=base,inner sep=0pt, outer sep=0pt, scale=  0.92] at (134.94,  7.64) {Usual weekly earnings at baseline};
\end{scope}
\begin{scope}
\path[clip] (  0.00,  0.00) rectangle (231.26,187.90);
\definecolor{drawColor}{RGB}{0,0,0}

\node[text=drawColor,anchor=base west,inner sep=0pt, outer sep=0pt, scale=  0.92] at (  5.50,103.79) {$p(x)$};
\end{scope}
\end{tikzpicture}}
		\vspace{-4pt}
		\subcaption{The proportion of the employed induced to work by the treatment.}
	\end{subfigure}
	\begin{subfigure}{0.48\textwidth}
		\centering	
		\scalebox{0.95}{
\begin{tikzpicture}[x=1pt,y=1pt]
\definecolor{fillColor}{RGB}{255,255,255}
\path[use as bounding box,fill=fillColor,fill opacity=0.00] (0,0) rectangle (231.26,187.90);
\begin{scope}
\path[clip] (  0.00,  0.00) rectangle (231.26,187.90);
\definecolor{drawColor}{RGB}{255,255,255}
\definecolor{fillColor}{RGB}{255,255,255}

\path[draw=drawColor,line width= 0.6pt,line join=round,line cap=round,fill=fillColor] (  0.00,  0.00) rectangle (231.26,187.90);
\end{scope}
\begin{scope}
\path[clip] ( 56.09, 30.69) rectangle (225.76,182.40);
\definecolor{fillColor}{RGB}{255,255,255}

\path[fill=fillColor] ( 56.09, 30.69) rectangle (225.76,182.40);
\definecolor{drawColor}{gray}{0.87}

\path[draw=drawColor,line width= 0.1pt,line join=round] ( 56.09, 42.76) --
	(225.76, 42.76);

\path[draw=drawColor,line width= 0.1pt,line join=round] ( 56.09, 93.70) --
	(225.76, 93.70);

\path[draw=drawColor,line width= 0.1pt,line join=round] ( 56.09,144.64) --
	(225.76,144.64);

\path[draw=drawColor,line width= 0.1pt,line join=round] ( 83.08, 30.69) --
	( 83.08,182.40);

\path[draw=drawColor,line width= 0.1pt,line join=round] (121.64, 30.69) --
	(121.64,182.40);

\path[draw=drawColor,line width= 0.1pt,line join=round] (160.21, 30.69) --
	(160.21,182.40);

\path[draw=drawColor,line width= 0.1pt,line join=round] (198.77, 30.69) --
	(198.77,182.40);

\path[draw=drawColor,line width= 0.3pt,line join=round] ( 56.09, 68.23) --
	(225.76, 68.23);

\path[draw=drawColor,line width= 0.3pt,line join=round] ( 56.09,119.17) --
	(225.76,119.17);

\path[draw=drawColor,line width= 0.3pt,line join=round] ( 56.09,170.11) --
	(225.76,170.11);

\path[draw=drawColor,line width= 0.3pt,line join=round] ( 63.80, 30.69) --
	( 63.80,182.40);

\path[draw=drawColor,line width= 0.3pt,line join=round] (102.36, 30.69) --
	(102.36,182.40);

\path[draw=drawColor,line width= 0.3pt,line join=round] (140.93, 30.69) --
	(140.93,182.40);

\path[draw=drawColor,line width= 0.3pt,line join=round] (179.49, 30.69) --
	(179.49,182.40);

\path[draw=drawColor,line width= 0.3pt,line join=round] (218.05, 30.69) --
	(218.05,182.40);
\definecolor{drawColor}{RGB}{0,0,0}

\path[draw=drawColor,line width= 1.1pt,line join=round] ( 63.80,112.50) --
	( 83.08,115.01) --
	(102.36,116.54) --
	(121.64,115.77) --
	(140.93,115.34) --
	(160.21,125.16) --
	(179.49,124.88) --
	(198.77,111.42) --
	(218.05, 85.61);

\path[draw=drawColor,line width= 1.1pt,line join=round] ( 63.80,144.21) --
	( 83.08,144.58) --
	(102.36,144.41) --
	(121.64,140.10) --
	(140.93,137.37) --
	(160.21,140.14) --
	(179.49,142.33) --
	(198.77,143.17) --
	(218.05,137.41);

\path[draw=drawColor,line width= 1.1pt,dash pattern=on 4pt off 4pt ,line join=round] ( 63.80,100.79) --
	( 83.08, 98.05) --
	(102.36, 97.47) --
	(121.64, 98.66) --
	(140.93, 97.81) --
	(160.21,104.87) --
	(179.49, 99.55) --
	(198.77, 77.57) --
	(218.05, 37.58);

\path[draw=drawColor,line width= 1.1pt,dash pattern=on 4pt off 4pt ,line join=round] ( 63.80,154.63) --
	( 83.08,158.38) --
	(102.36,158.06) --
	(121.64,152.82) --
	(140.93,151.36) --
	(160.21,160.53) --
	(179.49,165.67) --
	(198.77,171.64) --
	(218.05,175.51);
\definecolor{drawColor}{gray}{0.70}

\path[draw=drawColor,line width= 0.6pt,line join=round,line cap=round] ( 56.09, 30.69) rectangle (225.76,182.40);
\end{scope}
\begin{scope}
\path[clip] (  0.00,  0.00) rectangle (231.26,187.90);
\definecolor{drawColor}{gray}{0.30}

\node[text=drawColor,anchor=base east,inner sep=0pt, outer sep=0pt, scale=  0.73] at ( 51.14, 65.20) {-0.2};

\node[text=drawColor,anchor=base east,inner sep=0pt, outer sep=0pt, scale=  0.73] at ( 51.14,116.14) {0.0};

\node[text=drawColor,anchor=base east,inner sep=0pt, outer sep=0pt, scale=  0.73] at ( 51.14,167.08) {0.2};
\end{scope}
\begin{scope}
\path[clip] (  0.00,  0.00) rectangle (231.26,187.90);
\definecolor{drawColor}{gray}{0.70}

\path[draw=drawColor,line width= 0.3pt,line join=round] ( 53.34, 68.23) --
	( 56.09, 68.23);

\path[draw=drawColor,line width= 0.3pt,line join=round] ( 53.34,119.17) --
	( 56.09,119.17);

\path[draw=drawColor,line width= 0.3pt,line join=round] ( 53.34,170.11) --
	( 56.09,170.11);
\end{scope}
\begin{scope}
\path[clip] (  0.00,  0.00) rectangle (231.26,187.90);
\definecolor{drawColor}{gray}{0.70}

\path[draw=drawColor,line width= 0.3pt,line join=round] ( 63.80, 27.94) --
	( 63.80, 30.69);

\path[draw=drawColor,line width= 0.3pt,line join=round] (102.36, 27.94) --
	(102.36, 30.69);

\path[draw=drawColor,line width= 0.3pt,line join=round] (140.93, 27.94) --
	(140.93, 30.69);

\path[draw=drawColor,line width= 0.3pt,line join=round] (179.49, 27.94) --
	(179.49, 30.69);

\path[draw=drawColor,line width= 0.3pt,line join=round] (218.05, 27.94) --
	(218.05, 30.69);
\end{scope}
\begin{scope}
\path[clip] (  0.00,  0.00) rectangle (231.26,187.90);
\definecolor{drawColor}{gray}{0.30}

\node[text=drawColor,anchor=base,inner sep=0pt, outer sep=0pt, scale=  0.73] at ( 63.80, 19.68) {0};

\node[text=drawColor,anchor=base,inner sep=0pt, outer sep=0pt, scale=  0.73] at (102.36, 19.68) {100};

\node[text=drawColor,anchor=base,inner sep=0pt, outer sep=0pt, scale=  0.73] at (140.93, 19.68) {200};

\node[text=drawColor,anchor=base,inner sep=0pt, outer sep=0pt, scale=  0.73] at (179.49, 19.68) {300};

\node[text=drawColor,anchor=base,inner sep=0pt, outer sep=0pt, scale=  0.73] at (218.05, 19.68) {400};
\end{scope}
\begin{scope}
\path[clip] (  0.00,  0.00) rectangle (231.26,187.90);
\definecolor{drawColor}{RGB}{0,0,0}

\node[text=drawColor,anchor=base,inner sep=0pt, outer sep=0pt, scale=  0.92] at (140.93,  7.64) {Usual weekly earnings at baseline};
\end{scope}
\begin{scope}
\path[clip] (  0.00,  0.00) rectangle (231.26,187.90);
\definecolor{drawColor}{RGB}{0,0,0}

\node[text=drawColor,anchor=base west,inner sep=0pt, outer sep=0pt, scale=  0.92] at (  5.50,102.76) {LATE};
\end{scope}
\end{tikzpicture}}
		\vspace{-4pt}
		\subcaption{Bounds on the LATE for the always observed (log wages).}
	\end{subfigure}
	\caption{Conditional Lee bounds for the Job Corps program.}
	\footnotesize{\textit{Notes:}  The solid lines present the estimates of the bounds on the average treatment effect conditional on usual weekly earnings at baseline. The	dashed lines mark pointwise 95\% confidence intervals.}
	\label{fig:CondLee}
\end{figure}

The results are presented in Figure \ref{fig:CondLee}.  The bandwidth is selected based on smoothness constants calibrated through the procedure described in Online Appendix~\ref{A:ROT}. The point estimates indicate that the treatment encourages taking up employment. The bounds on the treatment effect on wage rates are relatively flat for low weakly earnings at the baseline, where they are very similar to the unconditional estimates of \citet{Lee2009}. I note that there is a mass point in the distribution of the covariate at zero, but this does not invalidate the results.

\section{Conclusions}\label{sec1:Conclusions}
I propose a nonparametric estimator of truncated conditional expectation functions based on an orthogonal conditional moment and local linear methods. When the truncation quantile level is known, I show that the proposed estimator is asymptotically equivalent to its infeasible analog that uses the true conditional quantile function, and I find its asymptotic distribution. I~also consider estimation with an estimated truncation quantile level.
The proposed estimator is applied in two empirical settings: sharp regression discontinuity designs with a manipulated running variable and randomized experiments with sample selection.

\appendix
\section{Proofs}
I define some additional, shorthand notation.
Let $q_j(\eta)= \partial_x^j Q(\eta, x_0)$ and $\widehat{q}_j(\eta;a)=\widehat{q}_j(\eta, x_0;a)$ for $j \in \{0,1\}$, $\widehat{Q}(\eta,x;a)=\widehat{Q}^{ll}(\eta,x;x_0,a)$, $\khi=\kh$, $X_{h,i}=(X_i-x_0)/h$,  $\widetilde{X}_{h,i}=(1, X_{h,i})^\top$, $Q^*(\eta, x)= q_0(\eta)+q_1(\eta)(x-x_0)$, $L_i(b)=b_0+b_1(X_i-x_0)$, $Y_i'(b)=Y_i-L_i(b)$, and $\mathcal{X}_h=\mathcal{X}(x_0,h)$. 
I put $C_f\equiv \sup \{|f_{Y|X}(y,x)|: x \in \mathcal{X} \text{ and } y \in [Q(\eta,x)\pm\epsilon] \} < \infty$, where  $\epsilon$ is as in Assumption~\ref{ass:quantile}.
Two-dimensional vectors are indexed starting with zero, so that, e.g.,  $q(\eta)=(q_0(\eta), q_1(\eta))^\top$.

\subsection{Auxiliary Lemmas}
I show some auxiliary results that are used throughout the proofs.

\begin{lemma}\label{lemma:basics} Suppose that Assumptions \ref{ass:ass1}, \ref{ass:quantile}, and \ref{ass:kernel} hold. Then for $j \in \mathbb{N}$ it holds that
	\begin{enumerate}[label=(\roman*),nosep]
		\item $\displaystyle 	S_{n,j}\equiv\frac{1}{n} \Sum \khi \Xihj = \mu_jf_X(x_0) + o_p(1).  $
		\item $\displaystyle  	\frac{1}{n} \Sum \khi \Xihj \left( \one{Y_i \leq Q(\eta,X_i)} -\eta \right) = O_p((nh)^{-1/2}). $
	\end{enumerate}
\end{lemma}
\begin{proof}
	Standard kernel calculations.
\end{proof}

\begin{lemma}\label{lemma:Bahadur_quantile}
	Suppose that Assumptions~\ref{ass:ass1}, \ref{ass:quantile}, and \ref{ass:kernel} hold. Then
	$$
	a^j(\widehat{q}_j(\eta;a) - q_j(\eta)) = O_p(a^2+(an)^{-1/2}) \text{ for } j\in \{0,1\}.
	$$
\end{lemma}
\begin{proof}
	The lemma follows, e.g., from the proof of Theorem 2 of \citet{fan1994robust} or from the proof of Lemma~\ref{lemma:est_quantile}, where I allow for an estimated truncation quantile level.
\end{proof}

\begin{lemma}\label{lemma:Qhat} Suppose that Assumptions~\ref{ass:ass1}, \ref{ass:quantile}, and \ref{ass:kernel} hold. Then
	\[
	\sup_{x \in \mathcal{X}_h}|\widehat{Q}(\eta,x;a)-Q(\eta,x)| =O_p(w_n),
	\]	
	where $w_n=a^2 + h^2 + (a+h)(a^3n)^{-1/2}$, as defined in Lemma~\ref{lemma:asy_equivalence}.
\end{lemma}
\begin{proof} Using the Taylor expansion of $Q(\eta,x)$ in $x$ with the mean-value form of the remainder, Assumption~\ref{ass:quantile}(a), and the triangle inequality, it follows that
	\begin{align*}
	\sup_{x \in \mathcal{X}_{h}}&|\widehat{Q}(\eta,x;a)-Q(\eta,x)| \\
	&  \leq \left|\widehat{q}_0(\eta;a)-q_0(\eta)\right| + \sup_{x \in \mathcal{X}_{h}} \left|(\widehat{q}_1(\eta;a)-q_1(\eta))(x-x_0)\right| +  \sup_{x, \wt x \in \mc X_h} |( \partial_x^1Q(\eta,\wt x)  - q_1(\eta) )  (x-x_0) | \\
	&= O_p(a^2+(an)^{-1/2} + h(a+(a^3n)^{-1/2}) + h^2)=O_p(w_n). \qedhere
	\end{align*}
\end{proof}

\begin{lemma}\label{lemma:precision1} Suppose that Assumptions~\ref{ass:ass1}, \ref{ass:quantile}, and \ref{ass:kernel} hold, and $\wt Q$ is a,~possibly random, function such that $\sup_{x \in \mathcal{X}_h}|\wt{Q}(\eta,x)-Q(\eta,x)| =O_p(w_n)$. For $j \in \mathbb{N}$ it holds that:
	\begin{enumerate}[wide, labelindent=0pt, label=(\roman*),nosep]
		\item  $ \displaystyle \frac{1}{n} \Sum \khi  \Xihj \left(Y_i-Q(\eta, X_i) \right) \left( \one{Y_i \leq \wt{Q}(\eta,X_i) } - \one{Y_i \leq Q(\eta, X_i)} \right) = O_p\left( w_n^2 \right)$,
		\item  $ \displaystyle  \frac{1}{n} \Sum \khi  \Xihj \left(\wt{Q}(\eta, X_i)  - Q(\eta, X_i) \right)  \left( \one{Y_i \leq \wt{Q}(\eta,X_i) } - \one{Y_i \leq Q(\eta, X_i)} \right) = {O_p\left( w_n^2 \right)},$
		\item $ \displaystyle	\frac{1}{n}\Sum \khi \Xihj \left( \one{Y_i \leq Q(\eta,X_i)} - \one{Y_i \leq \wt{Q}(\eta,X_i)} \right) =  O_p( w_n ).$
	\end{enumerate}
	
\end{lemma}

\begin{proof}
	I prove only part (i). Parts (ii) and (iii) follow analogously.
	The proof is similar to the proof of Lemma~A.3 of \cite{Kato2012}.	For $l>0$, let
	\[
	\mathcal{M}_n(l)=\{g:\,\mathcal{X} \to \mathbb{R} \text{ s.t. } \sup_{x \in \mathcal{X}_h}|g(x)-Q(\eta,x)| \leq lw_n \}.
	\]
	For a function $g: \mathcal{X}\to \mathbb{R}$, let
	\[
	U_n(g) \equiv \bigg| \frac{1}{n} \Sum \khi  \Xihj (Y_i-Q(\eta,X_i)) \left( \one{Y_i \leq g(X_i)} - \one{Y_i \leq Q(\eta,X_i)} \right) \bigg|.
	\]
	It suffices to show that for each fixed $l>0$,
	\begin{equation}\label{claim:Ui}
	\sup_{ g \in \mathcal{M}_n(l)} 	U_n( g ) = O_p(w_n^2).
	\end{equation}
	It holds that
	\begin{align*}
	U_n( g ) & \leq \frac{1}{n} \Sum \khi |\Xihj| (Y_i-Q(\eta,X_i)) \one{Q(\eta,X_i)<Y_i\leq g(X_i)} \\
	&\quad + \frac{1}{n} \Sum \khi |\Xihj|  (Q(\eta,X_i)-Y_i) \one{g(X_i)<Y_i \leq Q(\eta,X_i)}.
	\end{align*}
	Let $U_{n,1}(g)$ and $U_{n,2}(g)$ denote the first and the second element in the above sum, respectively. They are both nonnegative. It holds that
	\begin{align*}
	\sup_{ g \in \mathcal{M}_n(l)} U_{n,1}( g )
	= \frac{1}{n} \Sum \khi |\Xihj|  (Y_i-Q(\eta,X_i)) \one{Q(\eta,X_i)<Y_i\leq Q(\eta,X_i)+l w_n } \equiv 
	\bar{U}_{n,1}.
	\end{align*}
	Further, for $n$ large enough,
	\begin{align*}
	E\left[ \bar{U}_{n,1}\right] & \leq \E\left[  k_h(X-x_0) |X_h^j|  lw_n \one{Q(\eta,X)<Y \leq Q(\eta,X)+l w_n}\right] \\
	& \leq C_f l^2 w_n^2 \int k_h(x-x_0) f(x)  dx = O(w_n^2).
	\end{align*}
	Since $\bar{U}_{n,1}$ is nonnegative, it follows from Markov's inequality that $ \bar{U}_{n,1}=O_p(w_n^2)$. Applying the same reasoning to $U_{n,2}(g)$ yields \eqref{claim:Ui}.
\end{proof}

\subsection{Proofs of Lemma~\ref{lemma:asy_equivalence} and Theorem~\ref{th:asy_distribution}}

\begin{proof}[Proof of Lemma~\ref{lemma:asy_equivalence}]
	By standard algebra of the weighted least squares estimator, it holds that
	\[
	\widehat{m}(\eta,x_0;a,h) - \widetilde{m}(\eta,x_0;h) = \frac{S_{n,2}(\wh\Psi_{n,0}(a)  - \widetilde{\Psi}_{n,0}) - S_{n,1} (\wh\Psi_{n,1}(a)- \widetilde{\Psi}_{n,1})}{S_{n,2} S_{n,0} - S_{n,1}^2 }.
	\]
	where  $\wh\Psi_{n,j}(a) =\frac{1}{n} \Sum \khi \Xihj \psi_i(\eta, \widehat{Q}(\eta, X_i;a))$, $\widetilde{\Psi}_{n,j}=\frac{1}{n} \Sum \khi \Xihj \psi_i(\eta, Q(\eta, X_i))$, and $S_{n,j}$ is defined in Lemma \ref{lemma:basics}.
	The denominator converges in probability to a positive number. I~consider the terms in the numerator.	For $j \in \{0,1\}$, it holds that
	\begin{align*}
	\wh\Psi_{n,j}(a)  - \widetilde{\Psi}_{n,j} & = \frac{1}{n} \Sum \khi \Xihj \bigg( \frac{1}{\eta} Y_i \left( \one{Y_i \leq \widehat{Q}(\eta,X_i;a)} - \one{Y_i \leq Q(\eta, X_i)} \right) \\
	& \quad- \frac{1}{\eta} \widehat{Q}(\eta,X_i;a) \one{Y_i \leq \widehat{Q}(\eta,X_i;a)} + \frac{1}{\eta} Q(\eta,X_i) \one{Y_i \leq Q(\eta,X_i)} \\
	& \quad \pm \frac{1}{\eta} \widehat{Q}(\eta,X_i;a) \one{Y_i \leq Q(\eta,X_i)} -  (Q(\eta,X_i) -  \widehat{Q}(\eta,X_i;a)  ) \bigg)\\
	& =  \frac{1}{n} \Sum  \khi \Xihj   \frac{1}{\eta} \left( Q(\eta,X_i) -\widehat{Q}(\eta,X_i;a) \right) \left( \one{Y_i \leq Q(\eta,X_i)} -\eta \right)    + O_p(w_n^2),
	\end{align*}
	where the last equality follows from Lemma \ref{lemma:precision1}.	
	Further,
	\begin{align*}
	\frac{1}{n} \Sum & \khi \Xihj  \frac{1}{\eta} \left( Q(\eta,X_i) -\widehat{Q}(\eta,X_i;a) \right) \left( \one{Y_i \leq Q(\eta,X_i)} -\eta \right)   \\
	& = \frac{1}{\eta} (  q_0(\eta)-\widehat{q}_0(\eta;a) ) \frac{1}{n} \Sum \khi \Xihj   \left( \one{Y_i \leq Q(\eta,X_i)} -\eta \right) \\
	& \quad + \frac{1}{\eta} h(  q_1(\eta)-\widehat{q}_1(\eta;a) ) \frac{1}{n} \Sum \khi X_{h,i}^{j+1}  \left( \one{Y_i \leq Q(\eta,X_i)} -\eta \right) \\
	& \quad +  \frac{1}{n} \Sum \khi \Xihj \frac{1}{\eta} ( Q(\eta,X_i) - q_0(\eta) - q_1(\eta)(X_i-x_0) ) \left( \one{Y_i \leq Q(\eta,X_i)} -\eta \right) \\
	&\equiv L_1+L_2+L_3.
	\end{align*}
	By Lemmas~\ref{lemma:basics} and~\ref{lemma:Bahadur_quantile}, it holds that
	$L_1= O_p( a^2+(na)^{-1/2})O_p((nh)^{-1/2})$ and $L_2= h/aO_p( a^2+(na)^{-1/2})O_p((nh)^{-1/2})$. Further, $\E[L_3|X_1,...,X_n]=0$ and $\Var[L_3|X_1,...,X_n]=O_p(h^4(nh)^{-1})$, which implies that $L_3=O_p(h^2(nh)^{-1/2})$.
	In total, 
	\begin{align*}
	\wh\Psi_{n,j}(a)  - \widetilde{\Psi}_{n,j} & = O_p( a^2+(na)^{-1/2}+ h(a+(a^3n)^{-1/2}) + h^2  )O_p((nh)^{-1/2}) + O_p(w_n^2)\\
	& = O_p(w_n (nh)^{-1/2} + w_n^2),
	\end{align*}
	which concludes the proof.
\end{proof}

\begin{remark}
	The proof of Lemma~\ref{lemma:asy_equivalence} does not explicitly use the orthogonality condition stated in equation \eqref{eq:orthogonality}. However, this property is the reason why the terms involving $ \widehat{q}(\eta;a)$ are negligible in the expansion of $\wh\Psi_{n,j}(a)  - \widetilde{\Psi}_{n,j}$.
\end{remark}

\begin{proof}[Proof of Theorem~\ref{th:asy_distribution}] The first step of the proof is to show that under the assumptions made on the bandwidths, the remainder in Lemma~\ref{lemma:asy_equivalence} is of order $o_p(h^2 + (nh)^{-1/2} )$.
	Recall that $w_n= a^2 + h^2 + (a+h)(a^3n)^{-1/2}$.
	By Assumption~\ref{ass:kernel}(b), it holds that
	\begin{align*}
	O_p(w_n(nh)^{-1/2}  + w_n^2) & = O_p\left(	 w_n(nh)^{-1/2} + a^4 + h^4 + (a^2+h^2)(a^3n)^{-1} \right) \\ 
	& = O_p\left(  h(a^3n)^{-1/2} (nh)^{-1/2} + a^4 + (an)^{-1} + h^2(a^3n)^{-1}\right) + o_p(h^2+(nh)^{-1/2}).
	\end{align*}
	Further, the following equivalence statements hold
	\begin{equation*}
	\begin{aligned}	
	h^2/(a^3n) \to 0 &\iff (nh)^{-1}h \prec a,   \\     
	(nh)^{1/2}/(an) \to 0 &\iff (nh)^{-1/2}h \prec a,  
	\end{aligned}\qquad
	\begin{aligned}
	a^4/h^2 \to 0 & \iff a \prec \sqrt{h}, \\
	(nh)^{1/2} h^2/(a^3n) \to 0 & \iff (nh)^{-1/6} h \prec a. 
	\end{aligned}
	\end{equation*}
	The conditions on the right-hand sides hold under the assumptions made.
	
	The asymptotic first-order distribution of $\wh m(\eta,x_0;a,h)$ is hence the same as that of $\wt m(\eta,x_0;h)$. The variance is derived as follows:
	\begin{align*}
	\Var[\psi(\eta, Q(\eta,X)) |X=x_0] & = \E\left[ \left(\psi(\eta, Q(\eta,X)) - m(\eta,X)\right)^2 |X=x_0\right] \\
	& = \frac{1}{\eta^2} \E\left[ \left( (Y-m(\eta,X)) \one{Y \leq Q(\eta,X)} \right. \right. \\
	& \quad - \left. \left. (Q(\eta,X) - m(\eta,X))(\eta - \one{Y \leq Q(\eta,X)} ) \right)^2\big|X=x_0 \right] \\
	&=\frac{1}{\eta}\Var[Y|Y\leq Q(\eta,X), X=x_0] + \frac{(1-\eta)}{\eta}(Q(\eta,x_0)-m(\eta,x_0))^2. \qedhere
	\end{align*}\end{proof}

\subsection{Proof of Theorem~\ref{th:esteta}}
The main burden of the proof lies in studying the properties of the local linear quantile estimator with estimated quantile level, $\wh q(\wh\eta;h)$. Under the assumptions made, it has the same rate of convergence as the local linear quantile estimator with a known quantile level. 	Since $a \asymp h$, $w_n=h^2+(nh)^{-1/2}\equiv r_n$.

\begin{lemma}\label{lemma:est_quantile}
Under the assumptions of Theorem~\ref{th:esteta},
$h^j(\widehat{q}_j(\widehat{\eta};a) - q_j(\eta)) =O_p(r_n) \text{ for } j \in \{0,1\}$.	
\end{lemma}

\begin{proof}
	See Section~\ref{sec:proof_lemma_esteta}.
\end{proof}

\begin{proof}[Proof of Theorem~\ref{th:esteta}]
 Using Lemma~\ref{lemma:est_quantile}, the proof of Lemma~\ref{lemma:asy_equivalence} immediately implies that
	\[
	\frac{1}{n} \Sum \khi \Xihj  \psi_i( \eta, \widehat{Q}(\widehat{\eta},X_i;a)) -   \frac{1}{n} \Sum \khi \Xihj \psi_i(\eta, Q(\eta,X_i)) = O_p(r_n^2) \text{ for } j \in \{0,1\}.
	\]
	 Moreover,
	\begin{align*}
	\frac{1}{n} \Sum \khi \Xihj & \psi_i(\widehat{\eta}, \widehat{Q}(\widehat{\eta},X_i;a)) - \frac{1}{n} \Sum \khi \Xihj \psi_i(\eta, \widehat{Q}(\widehat{\eta},X_i;a)) \\
	& = \frac{1}{n} \Sum \khi \Xihj (Y_i - \widehat{Q}(\widehat{\eta}, X_i;a)) \one{Y_i \leq \widehat{Q}(\widehat{\eta}, X_i;a)} \left( \frac{1}{\widehat{\eta}} - \frac{1}{\eta} \right) \\
	&=\left(\frac{1}{n} \Sum \khi \Xihj (Y_i - Q(\eta, X_i)) \one{Y_i \leq Q(\eta, X_i)}+ O_p(r_n)\right) \left( \frac{1}{\widehat{\eta}} - \frac{1}{\eta} \right),
	\end{align*}
	where the second equality follows from Lemma~\ref{lemma:precision1} and the convergence rate of $\wh q(\wh \eta;a)$.
	Further, using the convergence rate of the local linear estimator, it follows that
	\[
	\wh m(\wh\eta,x_0;a,h) = \wt m(\eta,x_0;h) + \eta (m(\eta,x_0)-Q(\eta,x_0) + O_p(r_n)) \left( \frac{1}{\widehat{\eta}} - \frac{1}{\eta} \right).
	\]
	The proof is concluded by noting that
	\[
	 \frac{1}{\widehat{\eta}} - \frac{1}{\eta}  = -\frac{1}{\eta^2}\left(\widehat{\eta}-\eta\right) + O_p(r_n^2).\qedhere
	\]
\end{proof}

\subsection{Proof of Lemma~\ref{lemma:est_quantile}}\label{sec:proof_lemma_esteta}
To use the conventional notation, I write $h$ instead of $a$ in this proof. To begin with, I decompose the expression of interest as:
\begin{equation}\label{eq:decomposition}
	h^j(\wh q_j(\wh\eta;h) - q_j(\eta)) = h^j(\wh q_j(\wh\eta;h) - q^*_j(\wh\eta;h)) + h^j(q^*_j(\widehat{\eta};h)-q^*_j(\eta;h)) + h^j(q^*_j(\eta;h) - q_j(\eta)),	
\end{equation}
where 
\begin{equation}\label{eq:LLQestimand}
	(q^*_0(u;h), q^*_1(u;h))^\top = \argmin_{(b_0,b_1) \in \mathbb{R}^2} \E\left[ \rho_u(Y_i - b_0 - b_1(X-x_0)) \kX \right]
\end{equation}
is the bandwidth-dependent estimand of the local linear quantile estimator.
Under the assumptions made, $q^*(u;h)$ is uniquely defined for $u$ in a sufficiently small neighborhood of $\eta$ and $h$ small enough; see the proof of Lemma A.1 of \citet{guerre2012uniform}.\footnote{\citet{guerre2012uniform} assume that $f_{Y|X}(y|x)$ is positive on $\mathbb R \times \mc X$, but the asymptotic results for $q^*(u;h)$ rely on $f_{Y|X}(y|x)$ being positive on a neighborhhod of $(x_0,Q(u,x_0))$.}

Lemma~\ref{lemma:est_quantile} is proven in Lemmas~\ref{lemma:q_star_eta_hat} and~\ref{lemma:Bahadur_est_quantile}, where I analyze the three summands on the right-hand side of \eqref{eq:decomposition}. Let $v_n=(nh)^{-1/2}$ and $Q^*(u,x;h) = q^*_0(u;h) + q^*_1(u;h)(x-x_0)$.

\begin{lemma}\label{lemma:q_star_eta_hat}
	Suppose that the assumptions of Theorem~\ref{th:esteta} hold. Then for $j\in\{0,1\}$,
	\[
	(i)\, h^j(q^*_j(\eta;h) - q_j(\eta))=O(h^2) \text{ and } (ii)\, h^j(q^*_j(\widehat{\eta};h)-q^*_j(\eta;h)) = O(h^2)+O_p(v_n).
	\]
\end{lemma}

\begin{proof}
	The proof follows the lines of the proofs of Theorem~1 and Lemma~A.1 of \citet{guerre2012uniform}.\footnote{The proofs of \citet{guerre2012uniform} are more involved as they provide results uniform in the evaluation point, bandwidth, and quantile level. In the setting considered here, $x_0$ is fixed, and $h$ is a fixed sequence.} I outline only the main steps.
	
	The first-order condition of the population minimization problem in~\eqref{eq:LLQestimand} reads
	\[
	\int k(v) \begin{bmatrix}
		1 \\ v
	\end{bmatrix} \left( F_{Y|X}( q_0^*(u;h) + hq_1^*(u;h) v  |x_0+vh) - u \right)f_X(x_0+vh)dv = 0,
	\]
	which is well-defined also for $h=0$. Note that
	\[
	\int k(v) \begin{bmatrix}
		1 \\ v
	\end{bmatrix} \left( F_{Y|X}( q_0(u) |x_0) - u \right)f_X(x_0)dv = 0.
	\]
	Given uniqueness of $q^*(u;h)$, continuity of $F_{Y|X}(y|x)$ and $f_X(x)$ implies continuity of $q^*(u;h)$ in $u$ and $h$.\footnote{\citet{guerre2012uniform} invoke the implicit function theorem and differentiability of  $F_{Y|X}(y|x)$ and $f_X(x)$ to prove this claim, but continuity of these functions is sufficient at this point of the proof.} It therefore follows 	
	 that $q_0^*(u;h)\to q_0(u)$ and $hq_1^*(u;h)\to0$ as $h\to0$ uniformly over $u$ in a sufficiently small neighborhood of $\eta$.
	
Further,  it follows that $q^*_j(u;h)$ is  differentiable in $u$ with
	\[
	\begin{bmatrix}
		\partial_{u}^1q^*_0(u;h) \\
		h\partial_{u}^1q^*_1(u;h)
	\end{bmatrix} =
	\E \left[ \khX  f_{Y|X}(Q^*(u,X;h)|X) \widetilde{X}_h \widetilde{X}_h^\top \right]^{-1} \E\left[ \khX \widetilde{X}_h\right],
	\]
	which is bounded uniformly over $u$ in a sufficiently small neighborhood of $\eta$ and $h$ small enough.  Hence, part~(ii) follows using the mean value theorem. Part~(i) follows along the lines of the proof of Theorem~1.
\end{proof}

Next, I prove two stochastic equicontinuity results that are then used to show convergence of the criterion function of the local linear quantile estimator with an estimated quantile level. I~introduce the following additional notation. Let $\mathcal{M}_n(q,l)= \{b: |b_0-q_0|\leq l_0v_n \text{ and } h|b_1-q_1|\leq l_1v_n \}$. For a vector $l=(l_0,l_1)^\top$, I put $|l|\equiv ||l||_1=|l_0|+|l_1|$.

\begin{lemma}\label{lemma:equi1} Suppose that Assumptions~\ref{ass:ass1}, \ref{ass:quantile}, and \ref{ass:kernel} hold. Let $A_{n,i}=v_n \widetilde{X}_{h,i}^\top \theta$ for some $\theta$ and
	\begin{align*}
		T(b) & = \Sum \kXi (Y_i'(b) - A_{n,i}) \left( \one{ Y_i'(b) \leq A_{n,i} } - \one{ Y_i'(b) \leq 0}  \right),\\
		\bar{T}(b) & = T(b) - \E[T(b)].
	\end{align*}
	For any sequence $q_n \to q(\eta)$ and constant $M>0$, it holds that
	\[
	\sup_{ b \in \mathcal{M}_n(q_n,M)  } | \bar{T}(b) | = o_p(1).
	\]
\end{lemma}

\begin{proof}
	I will show that \textit{(i)} $\bar{T}(q_n)=o_p(1)$ and \textit{(ii)} $\sup_{ b \in \mathcal{M}_n(q_n,M)  } | \bar{T}(b) - \bar{T}(q_n)| = o_p(1)$.
	
	\noindent\textit{Part (i).} Note that
	\[
	T(b) = \Sum \kXi  (Y_i'(b) - A_{n,i}) \left( \one{0 < Y'_i(b) \leq A_{n,i} } - \one{ A_{n,i}  < Y'_i(b) \leq 0} \right).
	\]
	Using the fact that $f_{Y|X}(y|x)$ is bounded over $(x,y)$ in a sufficiently small neighborhood of $(x_0,Q(\eta,x_0))$, I obtain that
	\[
	\Var[T(q_n)] \leq  \Sum \E\big[ k(X_{h,i})^2 A_{n,i}^2 \one{ |Y'_i(q_n)| \leq |A_{n,i}| } \big] = O(nh v_n^3 )=o(1).
	\]
	Hence, $\bar{T}(q_n)=o_p(1)$.
	
	\noindent\textit{Part (ii).} I follow the lines of the proof of Lemma 4.1 of \citet{Bickel1975}. A similar claim has been also shown in Lemma~A.4 of \citet{RuppertCarroll1978}. Let $\Delta_i(q,b) \equiv Y_i'(q) - Y_i'(b)=L_i(b-q)$. It holds that
	\begin{align*}
		T&(q)  - T(b)  \\
		& =   \Sum \kXi  \big( (Y_i'(q) - Y_i'(b)) \left( \one{0 < Y'_i(q) \leq A_{n,i} } - \one{ A_{n,i}  < Y'_i(q)  \leq 0} \right) \\
		&\quad + (Y_i'(b)-A_{n,i}) \left(  \one{ Y_i'(q) \leq A_{n,i} } - \one{ Y_i'(q) \leq 0}  -  \one{ Y_i'(b) \leq A_{n,i} } + \one{ Y_i'(b) \leq 0} \right) \big) \\
		& =   \Sum \kXi \big( \Delta_i(q,b) \left( \one{0 < Y'_i(q) \leq A_{n,i}} - \one{ A_{n,i}  < Y'_i(q)  \leq 0} \right) \\
		&\quad + (Y_i'(q)-A_{n,i}-\Delta_i(q,b)) \left(  \one{ \Delta_i(q,b) < Y_i'(q) - A_{n,i} \leq 0 } -  \one{ 0 < Y_i'(q) - A_{n,i} \leq \Delta_i(q,b) } \right)\\
		&\quad + (Y_i'(q)-A_{n,i}-\Delta_i(q,b)) \left( \one{ 0 < Y_i'(q) \leq \Delta_i(q,b)} - \one{\Delta_i(q,b) < Y_i'(q) \leq 0} \right) \big).
	\end{align*}
	For $l=(l_0,l_1)$, let $b_{n,0}(l)=q_{n,0} + l_0 v_n$ and $b_{n,1}(l)=q_{n,1} + l_1v_n/h$. Note that for indices $i$ such that  $X_i \in \mathcal{X}_h$, it holds that $|\Delta_i(q_n,b_n(l))| \leq v_n |l|$. Therefore,
	\begin{align*}
		\Var[T(b_n(l))-T(q_n)]  \leq & 3 \Sum  \E\big[ k(X_{h,i})^2(v_n|l|)^2 \one{ | Y'_i(q_n)| \leq |A_{n,i}|}  \\
		& +   k(X_{h,i})^2 (v_n |l|)^2 \one{ | Y_i'(q_n) - A_{n,i}| \leq v_n|l| }  \\ 
		& +   k(X_{h,i})^2 (v_n |l| + |A_{n,i}| )^2\one{ | Y_i'(q_n) | \leq v_n|l| }  \big]\\ 
		= & O(nh v_n^3).
	\end{align*}
	Hence, for any fixed $l$,
	\begin{equation}\label{eq:conv_seq}
		\bar{T}(b_n(l)) - \bar{T}(q_n) = o_p(1).
	\end{equation}
	For a fixed $\delta \in (0,1)$, decompose $\mathcal{M}_n(q_n,M)$ as the union of cubes with vertices on the grid $J_n(\delta)=\{ q_n +  \delta M v_n (j_0,j_1/h)^\top:  j_i \in \{0,\pm 1, ..., \pm \lceil 1/\delta \rceil \} \text{ for } i=0,1 \} $, where $\lceil \cdot \rceil$ is the ceiling function. For $b \in \mathcal{M}_n(q_n,M)$, let $V_n(b)$ be the lowest vertex of the cube containing~$b$.
	It follows from \eqref{eq:conv_seq} that  
	\[
	\max \left\{ |\bar{T}(V_n(b)) - \bar{T}(q_n ) | : b \in \mathcal{M}_n(q_n,M)  \right\} =  o_p(1).
	\]
	Next, I consider the behavior on a cube. Note that for $X_i \in \mathcal{X}_h$, it holds that 
	$$\sup\{ |\Delta_i(V_n(b),b)|: b \in \mathcal{M}_n(V_n(b), \delta M) \} =2\delta Mv_n. $$
	Further,
	\begin{align*}
		|T(V_n(b)) - T(b)| & \leq  \Sum \kXi  \big( 2\delta Mv_n \one{ |Y'_i(V_n(b))| \leq |A_{n,i}| }\\
		&\quad+ 2\delta Mv_n  \one{ |Y_i'(V_n(b)) - A_{n,i}| \leq 2\delta Mv_n } \\
		&\quad+ (2 \delta Mv_n + |A_{n,i}|) \one{ | Y_i'(V_n(b))| \leq 2\delta Mv_n } \big)\\
		&\equiv \widetilde{T}(V_n(b),\delta).
	\end{align*}
	The reasoning leading to \eqref{eq:conv_seq} yields also that
	\[
	\max_{b \in J_n(\delta)} |\widetilde{T}(b, \delta)-\E[\widetilde{T}(b,\delta)] | = o_p(1).
	\]
	Moreover,
	\[
	\max_{b \in J_n(\delta)} \E[\widetilde{T}(b,\delta)] \leq \delta O(1).
	\]	
	uniformly in $\delta \in (0,1)$, which concludes the proof. 
\end{proof}

\begin{lemma}\label{lemma:equi2} Suppose that Assumptions~\ref{ass:ass1}, \ref{ass:quantile}, and \ref{ass:kernel} hold. Let
	\begin{align*}
		S(b) & = \frac{1}{\sqrt{nh}} \Sum \kXi X_{h,i}^j \one{Y_i'(b)\leq 0} ,\\
		\bar{S}(b)&= S(b)-\E[S(b)].
	\end{align*}
	For any sequence $q_n=(q_{n,0},q_{n,1})^\top \to q(\eta)$ and constant $M>0$, it holds that
	\begin{align*}
		\sup_{ b \in \mathcal{M}_n(q_n,M)  } |\bar{S}(b) - \bar{S}(q_n)| = o_p(1).
	\end{align*}
\end{lemma}

\begin{proof}
	The proof is similar to, but simpler than, the proof of Lemma \ref{lemma:equi1}. I am using the notation defined therein. I note that
	\begin{align*}
		S(q)-S(b)=  \frac{1}{\sqrt{nh}} \Sum k(X_{h,i}) X_{h,i}^j\left( \one{Y_i'(q) \leq \Delta_i(q,b) } - \one{Y_i'(q) \leq 0 } \right).
	\end{align*}
	
	For any fixed $l$, it holds that
	\begin{align*}
		\Var[S(b_n(l))-S(q_n)] = O_p(v_n)=o_p(1).
	\end{align*}
	Hence,
	\[
	\max \left\{ |\bar{S}(V_n(b)) - \bar{S}(q_n ) | : b \in \mathcal{M}_n(q_n,M)  \right\} =  o_p(1).
	\]
	Moreover,
	\begin{align*}
		|S(V_n(b))-S(b)| & \leq \frac{1}{\sqrt{nh}} \Sum k(X_{h,i}) |X_{h,i}| \one{ |Y_i'(V_n(b))| \leq 2\delta M v_n } \\
		& \equiv \widetilde{S}(V_n(b), \delta)
	\end{align*}
	It holds
	\[
	\max_{b \in J(\delta)} |\widetilde{S}(b, \delta)-\E[\widetilde{S}(b,\delta)] | = O_p(v_n)=o_p(1).
	\]
	Finally,
	\[
	\max_{b \in J(\delta)} \E[\widetilde{S}(b,\delta)] \leq \delta O_p(1).
	\]
	uniformly in $\delta$, which concludes the proof.
\end{proof}

\begin{lemma}\label{lemma:Bahadur_est_quantile}
	Suppose that the assumptions of Theorem~\ref{th:esteta} hold. Then
	$$
	h^j(\widehat{q}_j(\widehat{\eta};h) - q^*_j(\wh\eta;h))  = O_p(v_n) \text{ for } j \in \{0,1\}. 
$$
\end{lemma}

\begin{proof}
The proof proceeds similarly to the proof of Theorem~2 of \citet{fan1994robust}, but I allow for an estimated quantile level. Recall that $\rho_u(v)=v(u-\one{v\leq 0})$ and
	\[
	\widehat{q}(u;h) = \argmin_{(b_0,b_1) \in \mathbb{R}^2} \Sum \rho_u(Y_i - b_0 - b_1(X_i-x_0)) \kXi.
	\]
	Let $\widehat{\theta}_n(u) = \sqrt{nh} \left(\widehat{q}_0(u;h) - q_0^*(u;h), \, h(\widehat{q}_1(u;h) - q_1^*(u;h)) \right)^\top$. For a given $u$, the vector $\widehat{\theta}_n(u)$ minimizes the function
	
	\[
	G_n(u,\theta) =  \Sum \left[ \rho_u(Y_{i}^*(u;h)- v_n \theta^\top \widetilde{X}_{h,i} ) - \rho_u(Y_{i}^*(u;h)) \right] \kXi,
	\]
	where $Y_i^*(u;h)=Y_i- Q^*(u,X_i;h)$.	
	Let
	\begin{align*}
		W_n(u) & = v_n \Sum  \kXi  \widetilde{X}_{h,i} ( u - \one{Y_i^*(u;h) \leq 0 } ),\\
		T_n(u,\theta) & = - \Sum \kXi (Y_i^*(u;h)- v_n \theta^\top \widetilde{X}_{h,i} ) \left( \one{Y_i^*(u;h) - v_n \theta^\top \widetilde{X}_{h,i} < 0} - \one{Y_i^*(u;h)<0} \right).
	\end{align*}
	It holds that $G_n(u,\theta)=T_n(u,\theta) - \theta^\top W_n(u)$.
	
	By the mean value theorem,
	\begin{align*}
		\E[T_n(u,\theta)] & = - \E \Bigg[  \Sum \kXi \int_0^{v_n\theta^\top \widetilde{X}_{h,i} } \big(y  - v_n\theta^\top \widetilde{X}_{h,i}\big) f_{Y^*(u;h)|X}(y|X_i)dy \Bigg] \\
		& = \frac{1}{2}  \E \Bigg[ \Sum \kXi  f_{Y^*(u;h)|X}( \widetilde{z}_i(u)|X_i) (v_n\theta^\top \widetilde{X}_{h,i})^2 \Bigg] 
	\end{align*}
	for some $|\widetilde{z}_i(u)|\leq v_n|\theta|$. It follows that
	\[
	\E[T_n(u,\theta)] =  \frac{1}{2} \E\left[ k_h(X-x_0) (\theta^\top \widetilde{X}_{h} )^2\right] f_{Y|X}(q_0(u)|x_0)  + o(1)
	\]
uniformly for $u$ in a sufficiently small neighborhood of $\eta$.
	Using Lemma \ref{lemma:equi1}, it follows that
	\[
	T_n(\widehat{\eta},\theta) = \theta^\top S \theta + o_p(1),
	\]
	where
	\[
	S=f_{Y|X}(q_0(\eta)|x_0) f_X(x_0) \begin{bmatrix}
		\mu_0 & \mu_1 \\
		\mu_1 & \mu_2
	\end{bmatrix}.
	\]
	The convex, random function $\widehat{T}_n(\theta)\equiv T_n(\widehat{\eta},\theta)$ converges pointwise in $\theta$ to the convex function $\theta^\top S \theta$. By the convexity lemma of \citet{pollard1991asymptotics}, this convergence is uniform on any compact set.
	The function $\frac{1}{2}\theta^\top S \theta - \theta^\top W_n(\widehat{\eta}) $ is minimized at $S^{-1}W_n(\widehat{\eta})$. 
	Since by construction $\E[W_n(u)]=0$, Lemma \ref{lemma:equi2} implies that
	\[
	W_n(\widehat{\eta})=W_n(\eta_n)+o_p(1)=O_p(1).
	\]
	Using convexity again, the consistency argument of \citet[][proof of Theorem~1]{pollard1991asymptotics} implies that $\widehat{\theta}_n(\widehat{\eta})=S^{-1}W_n(\widehat{\eta})+o_p(1)$, which concludes the proof.
\end{proof}

\section{Estimation and Inference Details for Applications}
I formally introduce the estimators discussed in Sections~\ref{sec:RDD} and~\ref{sec:Lee}. Their asymptotic distributions follow easily from Theorems~\ref{th:asy_distribution} and~\ref{th:esteta}, and hence I state them without further proofs. For estimation of truncated conditional expectations truncated from above or from below, I use the following notation:
$\psi_i^L( u, q)= \psi_i(u,q)$ and $\psi_i^U( u, q)= \frac{1}{u}Y_i \one{ q \leq Y_i } -  \frac{1}{u}q(\one{ q \leq Y_i } - u)$.

\subsection{Regression Discontinuity Designs with Manipulation}\label{sec:Appendix_inference} 
In the following, I assume that if a function is continuous on an open interval, then its limits at the boundary points exist and are finite. I define $k^-_h(v)=\one{v < 0}k_h(v)$ and $k^+_h(v)=\one{v \ge 0}k_h(v)$.

The one-sided density limits  $f_X(x_0^-)$ and $f_X(x_0^+)$ can be estimated using, e.g., the `linear' boundary kernel \citep{jones1993simple} as
\begin{align}\label{cond:f1}
	\wh{f}_X^\star = \frac{1}{n} \Sum k^\star_h(X_i-x_0) \frac{\bar{\mu}_2  - \bar{\mu}_1 |(X_i-x_0)/h| }{\bar{\mu}_2\bar{\mu}_0 - \bar\mu_1^2} \text{ for } \star \in\{+,-\},
\end{align}
where $\bar{\mu}_{j}= \int_{0}^{\infty}v^j k(v)dv$ as defined in Section~\ref{sec:Asymptotics}. 
The share of always-assigned units among all units just to the right of the cutoff is then estimated as:
$\widehat{\tau} = \max\big\{ 1- \widehat{f}_X^-/\widehat{f}_X^+, 0\big\}$.  In the notation from the main text, $\widehat{\eta}= 1-\wh\tau=\min\big\{\widehat{f}_X^-/\widehat{f}_X^+,1 \big\}$.
To analyze this estimator, I impose smoothness assumptions on the density of the running variable.

\begin{assumption}\label{ass:smooth_density}
	$f_X(x)$ is bounded away from zero and twice continuously differentiable on $(x_0-\epsilon,x_0) \cup (x_0, x_0+\epsilon)$ for some $\epsilon>0$. 
\end{assumption}

Lemma \ref{lemma:eta} states the asymptotically linear representation of $\wh{\eta}$.

\begin{lemma}\label{lemma:eta}
	Suppose that Assumptions~\ref{ass:ass1}, \ref{ass:kernel}, and \ref{ass:smooth_density} hold. Then
	\begin{align*}
	\frac{1}{\eta}\big(\wh{\eta} - \eta \big) = \frac{   \wh{f}_X^- -f_X(x_0^-)}{f_X(x_0^-)}  - \frac{ \wh{f}_X^+ - f_X(x_0^+)}{f_X(x_0^+)}  + o(h^2) + o_p((nh)^{-1/2})=O(h^2) + O_p( (nh)^{-1/2} ).
	\end{align*}
\end{lemma}
Lemma~\ref{lemma:eta} implies that $\wh \eta$ satisfies the high-level assumption of Theorem~\ref{th:esteta}.
I note that the asymptotic bias and variance of $\frac{1}{\eta}\big(\wh{\eta} - \eta \big)$ are given by
\[
A_\eta= \frac{1}{2}\bar{\mu} \left( \frac{f_X''(x_0^-) }{f_X(x_0^-)} - \frac{f_X''(x_0^+)}{f_X(x_0^+)} \right) h^2 + o(h^2) \text{ and } W_\eta=\bar{\kappa} \left( \frac{1}{f_X(x_0^+)}+\frac{1}{f_X(x_0^-)} \right).
\]
These quantities appear in the asymptotic distribution of the bounds.

Let $m^L(\eta,x) = \E[Y|X=x, Y\leq Q(\eta,x)]$, $m^U(\eta,x) = \E[Y|X=x, Y\ge Q(1-\eta,x)]$, and $m(x) = \E[Y|X=x]$.
The truncated conditional expectations $m^L(\eta,x_0^+)$ and $m^U(\eta,x_0^+)$ are estimated as
\begin{align*}
\widehat{m}^{L}_{+}(\widehat{\eta},x_0) & = e_1^\top \argmin_{ \beta_0, \beta_1 } \Sum k_h^+(X_i-x_0)  \left(  \psi_i^L( \widehat{\eta}, \widehat{Q}^{ll}_+(\widehat{\eta},X_i;x_0,h))  -\beta_0 - \beta_1 (X_i-x_0)\right)^2 , \\
\widehat{m}^{U}_{+}(\widehat{\eta},x_0) & = e_1^\top \argmin_{ \beta_0, \beta_1 } \Sum k_h^+(X_i-x_0) \left(  \psi_i^U( \widehat{\eta}, \widehat{Q}^{ll}_+(1-\widehat{\eta},X_i;x_0,h))-\beta_0 - \beta_1(X_i-x_0)\right)^2,
\end{align*}
where the estimated quantile function $\widehat{Q}^{ll}_+$ is defined as in Section \ref{sec:TLLR}, except that it uses only observations to the right of the cutoff.

The conditional expectation $m(x_0^-)$ is estimated using the local linear estimator:
\[
\widehat{m}_-(x_0) =  e_1^\top  \argmin_{ \beta_0, \beta_1 } \Sum  k_h^-(X_i-x_0)  (Y_i-\beta_0 - \beta_1 (X_i-x_0))^2.
\]
The final estimators of the bounds on $\Gamma$ are defined as
\begin{align*}
\widehat{\Gamma}^L = \widehat{m}^{L}_{+}(\widehat{\eta},x_0) - \widehat{m}_-(x_0) \text{ and } \widehat{\Gamma}^U = \widehat{m}^{U}_{+}(\widehat{\eta},x_0) - \widehat{m}_-(x_0).
\end{align*}

I impose standard assumptions for the analysis of $\widehat{m}_-$.

\begin{assumption}\label{ass:left}  For some $\epsilon  > 0$, the following hold on $(x_0-\epsilon, x_0)$.
	\begin{enumerate}[label=(\alph*), nosep]
		\item $m(x)$ is twice continuously differentiable in $x$;
		\item $\Var[Y|X=x]$ is continuous, bounded, and bounded away from zero;
		\item There exists $\xi >0$ such that $\E\big[ |Y|^{2+\xi}\big|X=x \big]$ is uniformly bounded.
	\end{enumerate}
\end{assumption}

Corollary~\ref{th:Gamma} establishes joint asymptotic normality of the bounds estimators.
\begin{corollary}\label{th:Gamma}
	Suppose that Assumptions \ref{ass:ass1}--\ref{ass:kernel}, \ref{ass:smooth_density}, and \ref{ass:left} hold, mutatis mutandis. Then
	\[
	\sqrt{nh} \begin{bmatrix}
	\widehat{\Gamma}^L - \Gamma^L - (B_+^{L} - B_-) \\
	\widehat{\Gamma}^U - \Gamma^U - (B_+^{U} - B_-)
	\end{bmatrix}
	\xrightarrow{d} \mathcal{N} \left(0, \begin{bmatrix}
	V_+^{L} + V_- & Cov_+ + V_-\\ 
	Cov_+ + V_- & V_+^{U} + V_- 
	\end{bmatrix} \right),
	\]	
	where for $* \in \{L,U\}$
	
	\begin{align*}
	B^{*}_+ & = 	\frac{1}{2} \bar{\mu} \partial^2_x m^*(\eta,x_0^+)h^2   +  D^{*}_+ A_\eta + o_p(h^2), & & B_- = \frac{1}{2}  \bar{\mu} \partial^2_x m(x_0^-)h^2 + o_p(h^2), \\	
	V^{*}_+ & =   \frac{\bar{\kappa} }{ f_X(x_0^+)}  \Var[ \psi^*|X=x_0^+] + ( D^{*}_+)^2  W_\eta, & &	V_- =   \frac{\bar{\kappa}}{ f_X(x_0^-)}\Var[Y|X=x_0^-],\\
	Cov_+&= \frac{  \bar{\kappa}   }{f_X(x_0^+) } \Cov[ \psi^L, \psi^U  |X=x_0^+]  +  D^L_+ D^U_+  W_\eta & &
	\end{align*}
	with $\psi^L \equiv \psi^L( \eta,Q(\eta,X) )$, $\psi^U \equiv \psi^U( \eta,Q(1-\eta,X) )$, $D^{L}_+\equiv Q(\eta,x_0^+) -m^L(\eta,x_0^+)$, and $D^{U}_+ \equiv Q(1-\eta,x_0^+) -m^U(\eta,x_0^+)$.
\end{corollary}
Since $\widehat{\eta}$ is obtained based only on realizations of the covariate, there is no asymptotic covariance between $\widehat{\eta}$ and the estimators of the three  conditional expectations. The component in the asymptotic covariance due to estimation of $\eta$ is negative since $D^L_+ > 0$ and $D^U_-< 0$.

\subsection{Conditional Lee Bounds}\label{A:Lee}
For $d\in \{0,1\}$, let $s_d(x)=\Prob(S=1|D=d,X=x)$. The probability $s_d(x_0)$ can be estimated using the standard local linear estimator with the sample restricted to observations with $D_i=d$,
\begin{equation*}
\wh{s}_d(x_0)= e_1^\top \argmin_{ \beta_0, \beta_1 } \Sum k_h (S_i -\beta_0 - \beta_1(X_i-x_0))^2\one{D_i=d} \text{ for } d \in \{0,1\}.
\end{equation*}
Let $\wh{\eta}(x_0) = \wh{s}_0(x_0)/\wh{s}_1(x_0)$.
\begin{assumption}\label{ass:smooth_s}\,
	\begin{enumerate}[label=(\alph*), nosep]
		\item $s_d(x)$ is twice continuously differentiable on $\mc X$ for $d \in \{0,1\}$;
		\item $\E[D|X=x]$ is continuous in $x$ on $\mc X$.
	\end{enumerate}
\end{assumption}

Lemma \ref{lemma:eta_Lee} states the asymptotically linear representation of $\wh{\eta}(x_0)$.

\begin{lemma}\label{lemma:eta_Lee}
	Suppose that Assumptions \ref{ass:ass1}, \ref{ass:quantile}(b), \ref{ass:kernel}, and \ref{ass:smooth_s} hold. Then
	\begin{align*}
	\frac{1}{\eta}\big(\wh{\eta}(x_0) - \eta(x_0) \big) = \frac{   \wh{s}_0(x_0) -s_0(x_0)}{s_0(x_0)}  - \frac{ \wh{s}_1(x_0) - s_1(x_0)}{s_1(x_0)}  +  o_p(h^2+(nh)^{-1/2}) = O_p(h^2 + (nh)^{-1/2}).
	\end{align*}
\end{lemma}

I note that the asymptotic bias and variance of $\frac{1}{\eta}\big(\wh{\eta}(x_0) - \eta(x_0) \big)$ are given by
\begin{align*}
& A_\eta^{Lee}= \frac{1}{2}\mu \left( \frac{s_0''(x_0) }{s_0(x_0)} - \frac{s_1''(x_0)}{s_1(x_0)} \right) h^2 + o_p(h^2), \\ 
& W_\eta^{Lee}=\frac{\kappa}{f_X(x_0)} \left( \frac{s_0(x_0)(1-s_0(x_0))}{\Prob[D=0|X=x_0] s_0(x_0)^2}+\frac{s_1(x_0)(1-s_1(x_0))}{\Prob[D=1|X=x_0]s_1(x_0)^2} \right).
\end{align*}
These quantities appear in the asymptotic distribution of the bounds.

Let $Q_1(u,x)=Q_{D=1,S=1}(u,x)$, $m^L_1(\eta,x) = \E[Y|X=x, Y\leq Q_1(\eta,x), D=1,S=1]$, $m^U_1(\eta,x) = \E[Y|X=x, Y\ge Q_1(1-\eta,x), D=1, S=1]$, and $m_0(x) = \E[Y|X=x,D=0,S=1]$.
The truncated conditional expectations $m^L_1(\eta,x_0)$ and $m^U_1(\eta,x_0)$ are estimated as
\begin{align*}
&\widehat{m}_1^L(\eta,x_0)= e_1^\top \argmin_{ \beta_0, \beta_1 } \Sum  \kh S_iD_i  \left(  \psi_i^L( \eta, \widehat{Q}^{ll}_1(\eta,X_i;x_0,h))  -\beta_0 - \beta_1(X_i-x_0)\right)^2 , \\
&\widehat{m}_1^U(\eta,x_0) = e_1^\top \argmin_{ \beta_0, \beta_1 } \Sum  \kh S_iD_i \left(  \psi_i^U( \eta, \widehat{Q}^{ll}_1(1-\eta,X_i;x_0,h))-\beta_0 - \beta_1(X_i-x_0)\right)^2 ,
\end{align*}
where  the estimated quantile function $\widehat{Q}^{ll}_1$ is defined as in Section~\ref{sec:TLLR}, except that it uses only observations with $D_i=1$ and $S_i=1$. 
The conditional expectation $m_0(x_0)$ is estimated using the local linear estimator:
\[
\widehat{m}_0(x_0) =  e_1^\top  \argmin_{ \beta_0, \beta_1 } \Sum  \kh S_i(1-D_i)  (Y_i-\beta_0 - \beta_1(X_i-x_0))^2 .
\]

The final estimators of the bounds on the conditional average treatment effect are defined as
\[
\widehat{\Delta}^L(x_0) = \widehat{m}^{L}_1(\wh\eta(x_0), x_0) - \widehat{m}_0(x_0) \text{ and } \widehat{\Delta}^U(x_0) = \widehat{m}^{U}_1(\wh\eta(x_0),x_0) - \widehat{m}_0(x_0).
\]
I impose standard assumptions for the analysis of $\widehat{m}_0(x_0)$.

\begin{assumption}\label{ass:control}\,
	\begin{enumerate}[label=(\alph*), nosep]
		\item $m_0(x)$ is twice continuously differentiable with respect to $x$ on $\mc X$;
		\item $\Var[Y|X=x,S=1,D=0]$ is continuous in $x$ on $\mc X$;
		\item $\E\left[ |Y|^{2+\xi}|X=x, S=1, D=0 \right]$ is bounded uniformly over $x\in\mc X$ for some $\xi >0$.
	\end{enumerate}
\end{assumption}

Corollary~\ref{th:Delta} establishes joint asymptotic normality of the bounds estimators. The dependence on $x_0$ is dropped to ease the notation.
\newpage
\begin{corollary}\label{th:Delta}
Suppose that Assumptions \ref{ass:ass1}--\ref{ass:kernel}, \ref{ass:smooth_s}, and \ref{ass:control} hold, mutatis mutandis. If $h=O(n^{-1/5})$, then
	\[
	\sqrt{nh} \begin{bmatrix}
	\widehat{\Delta}^L - \Delta^L - (B_1^{L} - B_0) \\
	\widehat{\Delta}^U - \Delta^U - (B_1^{U} - B_0)
	\end{bmatrix}\\
	\xrightarrow{d} \mathcal{N} \left(0, \begin{bmatrix}
	V_1^{L} +  V_0 & Cov_1 +  V_0\\ 
	Cov_1 +  V_0 & V_1^{U} +  V_0 
	\end{bmatrix} \right),
	\]	
	where for $* \in \{L,U\}$
	\begin{align*}
	B^{*}_1 & = 	\frac{1}{2} \mu \partial^2_x m_1^*(\eta,x_0)h^2  + D^{*}_1 A^{Lee}_\eta + o_p(h^2), &	B_0 & = \frac{1}{2}  \mu \partial^2_x m_0(x_0)h^2 + o_p(h^2), \\
	V^{*}_1 & =   \frac{ \kappa \Var[\psi^*|X=x_0, S=1, D=1] }{ f_X(x_0) \E[SD|X=x_0] }   +   (D^{*}_1)^2  W^{Lee}_\eta, & V_0 & =   \frac{ \kappa \Var[Y|X=x_0,S=1,D=0] }{ f_X(x_0)\E[S(1-D)|X=x_0]},\\
	Cov_1&= \frac{  \kappa  \Cov[ \psi^L, \psi^U  |X=x_0,S=1,D=1]  }{f_X(x_0)\E[SD|X=x_0] }  + D^{L}_1 D^{U}_1 W_\eta^{Lee} & & 
	\end{align*}
	with $ \psi^L \equiv  \psi^L( \eta(X),Q_1(\eta(X),X) ) $, 
	$\psi^U \equiv \psi^U( \eta(X),Q_1(1-\eta(X),X) ) $,
	$D_1^{L} \equiv Q_1(\eta,x_0) -m^L_1(\eta,x_0)$, 
	and $D_1^{U} \equiv Q_1(1-\eta,x_0) -m^U_1(\eta,x_0)$.
\end{corollary}

\bibliographystyle{apalike}
\bibliography{References}

\newpage
\pagenumbering{arabic}
\renewcommand*{\thepage}{OA--\arabic{page}}

\begin{center}
\Large\textbf{Online Appendix: Nonparametric Estimation of Truncated Conditional Expectation Functions}

\vspace{0.5cm}
\large Tomasz Olma\vspace{0.2cm}

\today
\end{center}

\section{Extensions}
It is straightforward to generalize the results from the main text to allow for a vector of covariates and to use an arbitrary order of polynomials. I discuss extensions in these two directions separately to avoid cumbersome notation, and to highlight different orders of the remainder term in the respective asymptotic equivalence results. These results can be proven using exactly the same steps as in the proof of Lemma~\ref{lemma:asy_equivalence}, and I therefore omit the proofs.

\subsection{Multivariate Case}\label{A:multivariate}
Let $d$ be the dimension of $X$, and let $a=(a_1,...,a_d)$ and $h=(h_1,...,h_d)$ be vectors of bandwidths. Let $k(v)= \prod_{j=1}^{d}\mathcal{K}(v_j)$ be a $d$-dimensional product kernel built from the univariate kernel function $\mathcal{K}(\cdot)$. I put $|h|=\prod_{j=1}^{d}h_j$ and $k_h(v)= \prod_{j=1}^{d}\mathcal{K}(v_j/h_j)/h_j$, and similarly for $a$.

In the first step, I run a multivariate local linear quantile regression,
\begin{equation*}
(\widehat{q}_{0}(\eta,x_0;a),
\widehat{q}_{1}(\eta,x_0;a)^\top)^\top
= \argmin_{\beta_0,\beta_1} \Sum \rho_\eta(Y_i-\beta_0-\beta_1^\top(X_i-x_0))k_a(X_i-x_0).
\end{equation*}
Further,
\begin{equation*}
\widehat{Q}^{ll}(\eta,x;x_0,a)= \widehat{q}_{0}(\eta,x_0;a) +  \widehat{q}_{1}(\eta,x_0;a)^\top (x-x_0).
\end{equation*}
Finally,
\begin{equation*}
\widehat{m}(\eta, x_0; a, h) = e_1^\top \argmin_{ \beta_0, \beta_1 } \Sum k_h(X_i-x_0)( \psi_i(\eta, \widehat{Q}^{ll}(\eta,X_i;x_0,a)) -\beta_0 - \beta_1^\top(X_i-x_0))^2,
\end{equation*}
where $e_1=(1,0,...,0)^\top$ is a $(d+1)$-dimensional vector. Likewise, the infeasible estimator $\widetilde{m}(\eta, x_0; h)$ is defined as above but with $\psi_i(\eta, Q(\eta,X_i))$ as the outcome variable.

I maintain the smoothness assumptions on $Q(\eta, \ccdot)$ on the understanding that for boundary points the derivatives exist in the directions in which $x$ can be perturbed within $\mathcal{X}$. The assumptions on the kernel and the bandwidths are as follows.

\begin{assbis}{ass:kernel}\label{ass:kernelmulti}\,
	\begin{enumerate}[label=(\alph*), nosep]
		\item Kernel: $\mathcal{K}$ is a continuous, symmetric density function with compact support, say $[-1,1]$;
		\item As $n \to \infty$, $\max_j h_j \to 0$, $\max_j a_j  \to 0$, $n|h| \to \infty$, and $n|a| \to \infty$.
	\end{enumerate}
\end{assbis}

Lemma~\ref{th:main_multi} is the multivariate version of Lemma~\ref{lemma:asy_equivalence}.

\begin{lemma}\label{th:main_multi} Suppose that Assumptions \ref{ass:ass1}, \ref{ass:quantile}, and \ref{ass:kernelmulti} hold, $h_j \asymp a_j$ for $j \in \{1,...,d\}$, and $\mathcal{X}$ is a convex set. Then
	\[
	\widehat{m}(\eta,x_0;a,h) = \widetilde{m}(\eta,x_0;h) + O_p\Big( \sum_j h_j^4 + (n|h|)^{-1} \Big).
	\]
\end{lemma}
For $d>1$ the variance component of the remainder in Lemma~\ref{th:main_multi} is of larger order than it is in Lemma~\ref{lemma:asy_equivalence}. However, this result can still be used to obtain asymptotic normality because the infeasible estimator has a bias of order $O_p(\sum h_j^2 )$ and variance of order $O((n|h|)^{-1/2})$, which are smaller than the remainder in Lemma~\ref{th:main_multi}.

\subsection{Higher-Order Polynomials and Derivatives}\label{subsec:porder}
I introduce notation analogous to that in Section \ref{sec:TLLR}, making the dependence on $p$ explicit. The local polynomial quantile estimates are given by
\begin{equation*}
\widehat{q}^\top(\eta,x_0;a,p) = \argmin_{(\beta_0,...,\beta_p)^\top} \Sum k_h(X_i-x_0) \rho_\eta\Big(Y_i - \sum_{j=0}^p \frac{1}{j!}\beta_j (X_i-x_0)^j \Big).
\end{equation*}
I define the estimated $p$-th order approximation of $Q(\eta,\ccdot)$ as
\begin{equation*}
\widehat{Q}(\eta,x;x_0,a,p)=  \sum_{j=0}^p \frac{1}{j!} \widehat{q}_{j}(\eta,x_0;a,p) (x-x_0)^j.
\end{equation*}
The estimator of the $r$-th derivative of $m(\eta,x)$ with respect to $x$ at $x_0$, $\partial^r_{x} m(\eta,x_0)$, is defined as
\begin{align*}
\widehat{m}_r(\eta,x_0;a,h,p)= e_{r+1}^\top \argmin_{ \beta } \Sum k_h(X_i-x_0)\Big(\psi_i(\eta, \widehat{Q}(\eta,X_i;x_0,a,p) ) - \sum_{j=0}^p \frac{1}{j!}\beta_j (X_i-x_0)^j \Big)^2,
\end{align*}
where $e_{r+1}$ is a $(p+1)$-dimensional vector with $1$ at the $(r+1)$-th position and 0 otherwise.
Likewise, the infeasible estimator $\widetilde{m}_r(\eta, x_0; h,p)$ is defined as above but with $\psi_i(\eta, Q(\eta,X_i))$ as the outcome variable.

In order to prove an analog of Lemma~\ref{lemma:asy_equivalence}, one natural modification of Assumption~\ref{ass:quantile} is required.

\begin{assbis}{ass:quantile}\label{ass:quantilenis} $Q(\eta,x)$ is $p$ times differentiable with respect to $x$ on $\mc X$ and $\partial^p_x Q(\eta,x)$ is Lipschitz continuous in $x$. Moreover, Assumptions 2(b) and 2(c) hold.
\end{assbis}

\begin{lemma}\label{th:p}
	Suppose that Assumptions~\ref{ass:ass1}, \ref{ass:quantilenis}, and \ref{ass:kernel} hold, and $h \asymp a$. Then
	\[
	\widehat{m}_r(\eta,x_0;a,h,p) = \widetilde{m}_r(\eta,x_0;h,p)  + O_p( h^{-r}(h^{2(p+1)}+ (nh)^{-1}) ).
	\]
\end{lemma}
With this result, under modified Assumption \ref{ass:mean}, asymptotic normality follows e.g. from the results of \citet{Hong2003}. The stochastic part of $h^r(\widetilde{m}_r(\eta,x_0;h,p) -  \partial^r_{x} m(\eta,x_0))$ is of order $O_p( (nh)^{-1/2} )$, and its leading bias is of order $O_p(h^{p+1})$ or $O_p(h^{p+2})$. Lemma~\ref{th:p} allows to characterize the leading bias for all orders $p$ and derivatives $r\leq p$, both for interior and boundary points, except for the local constant estimator for interior points. Its leading bias is of order $O_p(h^2)$, which is the same as the order of the remainder in the above theorem. This case is discussed by \citet{Kato2012}.

\section{Alternative Approaches}\label{A:Alternative}
I contrast the estimator proposed in the main text with the weighted Nadaraya-Watson estimator of \citet{Kato2012} and two further approaches based on local linear methods that have not been formally studied in the literature so far. As reference points for the last two approaches, I also present the asymptotic distributions of the corresponding infeasible estimators employing the true conditional quantile function. 

\subsection{Weighted Nadaraya-Watson Estimator}\label{A:WNW}
\citet{Kato2012} proposes an estimator of the conditional expected shortfall based on the weighted Nadaraya-Watson (WNW) estimator of the conditional c.d.f. For interior points, the WNW regression estimator is asymptotically equivalent to the local linear estimator. Additionally, the WNW estimator of $F_{Y|X}(y|x_0)$ is monotone in $y$, and it lies between zero and one. Both these properties are not shared by the local linear estimator.
I emphasize that the WNW estimator is not defined for boundary points, but for interior points the estimator of \citet{Kato2012} bears some similarity with the approach developed in this paper. 

In the first step, \citet{Kato2012} estimates the conditional c.d.f. as
\begin{equation}
\widehat{F}^{\sss WNW}_{Y|X}(y|x_0;h)=\frac{\Sum p_i(x_0)\kh \one{Y_i \leq y} }{\Sum p_i(x_0) \kh},
\end{equation}
where $p_i(x_0)\ge 0$ are the empirical likelihood weights, which maximize $\Sum \log(p_i(x_0))$ subject to the constraints
$\Sum p_i(x_0)=1$ and  $\Sum p_i(x_0)(X_i-x_0)k_h(X_i-x_0)=0$.\footnote{If $x_0$ lies on the boundary, such that all $X_i-x_0$ have the same sign, then it is not possible to find non-negative weights satisfying the last constraint.}
He estimates $Q(\eta,x_0)$ as $\widehat{Q}^{\sss WNW}(\eta,x_0;h)=\inf\{y: \eta \leq \widehat{F}^{\sss WNW}_{Y|X}(y|x_0;h) \}$, and $m(\eta,x_0)$ as
\begin{align*}
\widehat{m}^{\sss WNW}(\eta,x_0;h) & =  \frac{ \Sum p_i(x_0)\kh Y_i \one{Y_i \leq \widehat{Q}^{\sss WNW}(\eta,x_0;h)} }{\Sum p_i(x_0)\kh \one{Y_i \leq \widehat{Q}^{\sss WNW}(\eta,x_0;h)} },
\end{align*}
which is essentially the WNW estimator with $\frac{1}{\eta}Y_i\one{Y_i \leq \widehat{Q}^{\sss WNW}(\eta,x_0;h)}$ as the outcome variable. \citet{Kato2012} shows that, under suitable assumptions, the estimator $\widehat{m}^{\sss WNW}$ is asymptotically equivalent to the WNW estimator (and hence to the local linear estimator) with $\psi_i(\eta, Q(\eta,x_0))$ as the outcome variable. In consequence, it is asymptotically normal with asymptotic variance $V(\eta,x_0)$ defined in  Theorem~\ref{th:asy_distribution}.\footnote{\citet{Kato2012} considers time series data, but, under the $\alpha$-mixing condition he imposes, the asymptotic variance of his estimator is the same as for i.i.d. data.} Its bias is of order $h^2$ and the leading bias constant is given by
\begin{align*}
B^{\sss WNW}(\eta, x_0) &= \frac{1}{2} \mu \frac{d^2}{dx^2}  \E[ \psi( \eta, Q(\eta,x_0)|X=x]|_{x=x_0}.
\end{align*}
The difference between the WNW approach and my approach, for interior points, results from the fact that they effectively estimate different curves, which coincide only at the evaluation point $x_0$. The two approaches have the same asymptotic variance but their biases are different, as shown in Proposition \ref{prop:WNW}.
\begin{proposition}\label{prop:WNW}
	Suppose that $\partial^2_x f_{Y|X}(y|x)$ is continuous and $\partial^2_x Q(\eta,x)$ is continuous in $x$. Then
	\[
	B^{\sss WNW}(\eta,x_0) = B(\eta,x_0) - \frac{1}{2 \eta} \mu f_{Y|X}(Q(\eta,x_0)|x_0) (\partial_x Q(\eta,x_0))^2,
	\]
	where $B(\eta,x_0) = \frac{1}{2} \mu\, \partial_x^2 m(\eta, x_0)$.
\end{proposition}
The second term of the difference on the right-hand side is always non-negative, so that $B^{\sss \sss WNW}(\eta,x_0) \leq  B(\eta,x_0) $.
Which of the two biases is larger in absolute value, depends on the specific data generating process. 
For example, it is possible that $B^{\sss WNW}(\eta,x_0)=0$ and $B(\eta,x_0)>0$, or that $B^{\sss WNW}(\eta,x_0)<0$ and $B(\eta,x_0)=0$.
However, I remark that in a simple location-scale model with a linear conditional expectation function and homoskedastic residuals, my estimator has no bias, whereas $|B^{\sss WNW}(\eta,x_0)|$ can be arbitrarily large.

\subsection{Local Linear Estimator Based on a Non-Orthogonal Moment}\label{A:NM}

\subsubsection{Definition and Discussion}
The non-orthogonal conditional moment (NM) in \eqref{eq:def2} motivates running a local linear regression without the second term included in the generated outcome variable $\psi_i(\eta, \widehat{Q}^{ll}(\eta,X_i;x_0,a))$. Let
\begin{align*}
\widehat{m}^{\sss NM}(\eta,x_0;a,h)= e_1^\top \argmin_{ (\beta_0,\beta_1) } \Sum k_h(X_i-x_0)\left( \frac{1}{\eta} Y_i \one{Y_i \leq \widehat{Q}^{ll}(\eta,X_i;x_0,a)} -\beta_0 - \beta_1(X_i-x_0)\right)^2. 
\end{align*}
Under assumptions, this estimator is consistent and asymptotically normal. However, it has one unappealing property---it is not translation invariant.  Adding a constant to all outcomes and subtracting it from the result can yield a different estimate than applying the estimator to the original data.\footnote{This difference is asymptotically very small in the case when the same bandwidth is used in both stages, but even then, the estimator is not numerically translation invariant.} The estimator $\widehat{m}$ is free of this deficiency.

\subsubsection{Formal Results}
First, I show that in the special case when the same bandwidth is used in both stages, the estimator $\widehat{m}^{\sss NM}(\eta,x_0;h,h)$ is asymptotically equivalent to the infeasible estimator $\widetilde{m}(\eta,x_0;h)$, and I~give the exact rate of the remainder. Second, I derive the asymptotic distribution in the general case allowing for different bandwidths.

\begin{proposition}\label{prop:NM} Suppose that Assumptions~\ref{ass:ass1}, \ref{ass:quantile}, and \ref{ass:kernel} hold. Then
	\[
	R^{\sss NM}(\eta,x_0;h) \equiv \widehat{m}^{\sss NM}(\eta,x_0;h,h) - \widetilde{m}(\eta,x_0;h) = O_p((h+(nh)^{-1/2})(h^2+(nh)^{-1/2})). 
	\]
	Suppose additionally that either $f(x)$ is continuously differentiable and $x_0$ is an interior point or  $\partial_x^1Q(\eta,x_0)=0$. Then $R^{\sss NM}(\eta,x_0;h) = O_p(h^4 + (nh)^{-1})$.
\end{proposition}

Let $\widetilde{m}^{\sss NM}(x_0,\eta;h)$ be the oracle estimator corresponding to $\wh{m}^{\sss NM}(x_0,\eta;a,h)$, i.e. a local linear estimator with $\frac{1}{\eta} Y_i \one{Y_i \leq Q(\eta,X_i)}$ as the outcome variable.
\begin{proposition}\label{prop:A:NM}
	Suppose that Assumptions~\ref{ass:ass1}--\ref{ass:kernel} hold, $\partial^2_xQ(\eta,x)$ is continuous in $x$ on $\mc X$, and $h/a \to \rho \in (0, \infty)$. Then
	\begin{enumerate}[wide, labelindent=0pt]
		\item [(i)] $ \displaystyle \sqrt{nh}\big( \widetilde{m}^{\sss NM}(x_0,\eta;h)   - m(\eta,x_0) - \widetilde{\mathcal{B}}^{\sss NM}(\eta,x_0,h)  \big) \xrightarrow{d} \mathcal{N}(0, \widetilde{V}^{NM}(\eta,x_0)),$
		
		where 
		\begin{align*}
			\widetilde{\mathcal{B}}^{\sss NM }(\eta,x_0,h) &=\frac{1}{2} \mu \partial_x^2m(\eta,x_0)h^2+o_p(h^2),\\
			\widetilde{V}^{\sss NM}(\eta,x_0) &= \frac{\kappa}{\eta f_X(x_0)} \left( \Var[Y|Y\leq Q(\eta,x_0), X=x_0] + (1-\eta)m(\eta,x_0)^2 \right).
		\end{align*} 
		\item [(ii)] $ \displaystyle \sqrt{nh}\big(  \widehat{m}^{\sss NM}(x_0,\eta;a,h)    - m(\eta,x_0) - \mathcal{B}^{\sss NM}(\eta, x_0,a,h) \big) \xrightarrow{d} \mathcal{N}(0, V^{\sss NM}(\eta, x_0,\rho)),$
		
		where
		\begin{align*}
			\mathcal{B}^{\sss NM}(\eta,x_0,a,h) & =\frac{1}{2} \mu \left(\partial_x^2m(\eta,x_0)h^2 + C^{\sss NM}(\eta,x_0) \partial_x^2Q(\eta,x_0)(a^2-h^2) \right) + o_p(h^2),\\
			V^{\sss NM}(\eta, x_0,\rho) & = \frac{\kappa}{\eta f_X(x_0)}\Var[Y|Y\leq Q(\eta,x_0), X=x_0]+ \frac{1-\eta}{ \eta f_X(x_0)(\mu_0 \mu_2 - \mu_1^2)^2 }\\ 
			&\,\times  \int_{\mathcal{D}} \left( \rho k(v\rho)(\mu_2 - \mu_1 v \rho) Q(\eta,x_0)  -k(v)(\mu_2-\mu_1 v)m(\eta,x_0)  \right)^2 dv 
		\end{align*}
		with $C^{\sss NM}(\eta, x_0)= \frac{d}{dy} \E[\frac{1}{\eta} Y \one{ Y \leq y } |X=x_0 ] |_{y=Q(\eta,x_0)}=\frac{1}{\eta}f_{Y|X}(Q(\eta,x_0)|x_0)Q(\eta,x_0)$, $\mathcal{D}=[-1,1]$ if $x_0$ lies in the interior of $\mathcal{X}$, $\mathcal{D}=[0,1]$ if $x_0$ lies on the boundary of $\mathcal{X}$, and $\mu_j = \int_{\mathcal{D}} k(v)v^j dv$.
	\end{enumerate}
\end{proposition}
Both bandwidths appear in the bias formula and the ratio $\rho$ appears in the asymptotic variance. When $\rho$ is small, i.e. $a$ is large relative to $h$, then the variance of the feasible estimator is close to the variance of the oracle estimator because $V^{\sss NM}(\eta, x_0,0)=\widetilde{V}^{\sss NM}(\eta,x_0)$.
The bias $\mathcal{B}^{\sss NM}(\eta,x_0,a,h)$ differs from the oracle bias due to the fact that, first, $Q(\eta,\ccdot)$ is replaced by its local linear approximation, and, second, this approximation is estimated.

\subsection{Local Linear Estimator on a Truncated Sample}\label{A:TS}
\subsubsection{Definition and Discussion}
The definition of $m(\eta,x)$ in \eqref{eq:meta} motivates running a local linear regression on a truncated sample (TS) restricted to observations that fall below the estimated conditional $\eta$-quantile function.\footnote{This approach has been proposed in a working paper by \citet{gerard2016identification}, but they do not derive its asymptotic distribution.} Let
\begin{align*}
\widehat{m}^{\sss TS}(\eta,x_0;a,h) = e_1^\top \argmin_{ (\beta_0,\beta_1) } \Sum k_h(X_i-x_0)\left(  Y_i  -\beta_0 - \beta_1(X_i-x_0)\right)^2 \one{Y_i \leq \widehat{Q}^{ll}(\eta,X_i;x_0,a)}.
\end{align*}
This estimator is translation invariant. Unlike in the case of $\widehat{m}$, the asymptotic distribution of $\widehat{m}^{\sss TS}$ explicitly depends on the first-stage bandwidth, and in general it involves more complicated bias and variance formulas than those in  Theorem~\ref{th:asy_distribution}. 
Only in the special case when the bandwidths in both stages are equal, is $\widehat{m}^{\sss TS}$ asymptotically equivalent to the infeasible estimator $\widetilde{m}$, and hence it has the asymptotic distribution given in  Theorem~\ref{th:asy_distribution}. However, for boundary points, the remainder in the Bahadur representation of $\widehat{m}^{\sss TS}(\eta,x_0;h,h)$ is in general of larger order than $O_p(h^4 + (nh)^{-1})$ obtained in Lemma~\ref{lemma:asy_equivalence} for bandwidths converging at the same rates.

The estimator based on the truncated sample with equal bandwidths corresponds most closely to the unconditional truncated mean, where the same (full) sample is used to first estimate the quantile and then to calculate the truncated mean. However, I~advocate using the estimator $\wh{m}$, as it makes the parallel between estimation of conditional expectation functions and truncated conditional expectation functions explicit.\footnote{Standard inference methods cannot be simply applied to the truncated sample.}  The very small remainder in Lemma~\ref{lemma:asy_equivalence} provides a strong theoretical justification for conducting inference as if the oracle estimator was available.

\subsubsection{Formal Results}
First, I show that in the special case when the same bandwidth is used in both stages, the estimator $\widehat{m}^{\sss TS}(\eta,x_0;h,h)$ is asymptotically equivalent to the infeasible estimator $\widetilde{m}(\eta,x_0;h)$, and I give the exact rate of the remainder. Second, I derive the asymptotic distribution in the general case allowing for different bandwidths.

\begin{proposition}\label{prop:TS} Suppose that Assumptions~\ref{ass:ass1}--\ref{ass:kernel} hold. Then
	\[
	R^{\sss TS}(\eta,x_0;h) \equiv \widehat{m}^{\sss TS}(\eta,x_0;h,h) - \widetilde{m}(\eta,x_0;h) = O_p( (h+(nh)^{-1/2})(h^2+(nh)^{-1/2}) ).
	\]
	Suppose additionally that either $f(x)$ is continuously differentiable and $x_0$ is an interior point or  $\partial_x^1Q(\eta,x_0) = \partial_x^1m(\eta,x_0)$. Then $R^{\sss TS}(\eta,x_0;h) = O_p(h^4 + (nh)^{-1})$.
\end{proposition}

Let $\widetilde{m}^{\sss TS}(x_0,\eta;h)$ be the oracle estimator corresponding to the estimator $\wh{m}^{\sss TS}(x_0,\eta;a,h)$, i.e. a local linear estimator using observations with $Y_i \leq Q(\eta,X_i)$.

\begin{proposition}\label{prop:A:TS}
	Suppose that Assumptions~\ref{ass:ass1}--\ref{ass:kernel} hold,  $\partial^2_xQ(\eta,x)$ is continuous in $x$ on $\mc X$, and $h/a \to \rho \in (0, \infty)$. Then
	
	\begin{enumerate}[wide, labelindent=0pt]
		\item [(i)] $ \displaystyle	\sqrt{nh}(\widetilde{m}^{\sss TS}(\eta,x_0; h)  - m(\eta,x_0) - \widetilde{\mathcal{B}}^{\sss TS}(\eta,x_0,h) \big) \xrightarrow{d} \mathcal{N}(0, \widetilde{V}^{\sss TS}(\eta,x_0)),$
		where 
		\begin{align*}
			&\widetilde{\mathcal{B}}^{\sss TS}(\eta,x_0,h) = \frac{1}{2} \mu \partial_x^2m(\eta,x_0)h^2 + o_p(h^2), \\
			&\widetilde{V}^{\sss TS}(\eta,x_0) = \frac{\kappa}{\eta f_X(x_0)}\Var[Y|Y\leq Q(\eta,x_0), X=x_0].
		\end{align*}
		
		\item [(ii)] $ \displaystyle  \sqrt{nh}(\widehat{m}^{\sss TS}(\eta,x_0; a,h)  - m(\eta,x_0) - \mathcal{B}^{\sss TS}(\eta, x_0,a,h) \big) \xrightarrow{d} \mathcal{N}(0, V^{\sss TS}(\eta,x_0,\rho)),$
		where
		\begin{align*}
			&\mathcal{B}^{\sss TS}(\eta,x_0,a,h)  = \frac{1}{2} \mu \left( \partial_x^2m(\eta,x_0)h^2-C^{\sss TS}(\eta,x_0) \partial_x^2Q(\eta,x_0)(h^2-a^2)  \right) + o_p(h^2),\\
			&V^{\sss TS}(\eta,x_0,\rho)  = \frac{\kappa}{\eta f_X(x_0)} \Big(\Var[Y|Y\leq Q(\eta,x_0), X=x_0] + \rho  (1-\eta) \left(Q(\eta,x_0) - m(\eta,x_0)\right)^2 \Big)
		\end{align*}
		with $C^{\sss TS}(\eta,x_0)=\frac{d}{dy} \E[Y |X=x_0, Y\leq y ] \big|_{y=Q(\eta,x_0)} =\frac{1}{\eta}f_{Y|X}(Q(\eta,x_0)|x_0)(Q(\eta,x_0) - m(\eta,x_0))$.
	\end{enumerate}
\end{proposition}

Both bandwidths appear in the bias formula, and the ratio $\rho$ appears in the asymptotic variance.  When $\rho$ is small, i.e. $a$ is large relative to $h$, then the variance of the feasible estimator is close to the variance of the oracle estimator because $V^{\sss TS}(\eta, x_0,0)=\widetilde{V}^{\sss TS}(\eta,x_0)$. The bias $\mathcal{B}^{\sss TS}(\eta,x_0,a,h)$ differs from the oracle bias due to the fact that, first, $Q(\eta,\ccdot)$ is replaced by its local linear approximation, and, second, this approximation is estimated.

\begin{remark}
	The two-stage procedure yielding $\widehat{m}^{\sss TS}$ with equal bandwidths provides an intuitive decomposition of the asymptotic variance $V(\eta,x_0)$ defined in Theorem~\ref{th:asy_distribution}. The asymptotic variance of the infeasible local linear estimator  equals  the first component of $V(\eta,x_0)$. The second, strictly positive, component of $V(\eta,x_0)$ is due to the first-step estimation.\footnote{An analogous decomposition holds for the unconditional truncated mean. A similar point is also discussed by \citet[Remark 2.9]{dimitriadis2019joint} in a parametric model.} 
\end{remark}

\section{Proofs of the Results in the Online Appendix}

\subsection{Proof of Proposition \ref{prop:WNW}}
\begin{proof}
	Note that
	\begin{align*}
	l(x)&  \equiv \E[m(\eta,X) - \psi(\eta,Q(\eta,x_0))|X=x]  = \frac{1}{\eta} \int_{Q(\eta,x_0)}^{Q(\eta,x)} (y-Q(\eta,x_0)) f_{Y|X}(y|x)dy.
	\end{align*}
	By the Leibniz integral rule, it holds that
	\begin{align*}
	l'(x) & = \frac{1}{\eta} \partial^1_x Q(\eta,x)(Q(\eta,x)-Q(\eta,x_0))f_{Y|X}(Q(\eta,x)|x) + \frac{1}{\eta} \int_{Q(\eta,x_0)}^{Q(\eta,x)} (y-Q(\eta,x_0)) \partial_{x} f_{Y|X}(y|x)dy.
	\end{align*}
	The proof is concluded by noting that
	\[
	l''(x_0) = \frac{1}{\eta} \big(\partial^1_x Q(\eta,x_0)\big)^2 f_{Y|X}(Q(\eta,x_0)|x_0).\qedhere
	\]
\end{proof}

\subsection{Proofs of Propositions \ref{prop:NM} and  \ref{prop:TS}}
An essential result used to prove these two propositions, not required for the proof of Lemma~\ref{lemma:asy_equivalence}, are the approximate first-order conditions of the local linear quantile estimator.

\begin{lemma}\label{lemma:AFOC}
	Suppose that Assumptions~\ref{ass:ass1}~and~\ref{ass:kernel} hold. Then for $j \in \{0,1\}$, it holds that
	\[
	\frac{1}{n}\Sum \khi X^j_{h,i} (\eta - \one{Y_i \leq \widehat{Q}(\eta,X_i;h)}  ) =  O_p( (nh)^{-1}).
	\]
\end{lemma}
\begin{proof}
	Similar claims have been proven by \citet[Theorem 3.3]{koenker1978regression}  and \citet[Theorem 1]{Ruppert1980}. Let
	\[
	G_n(b) = \frac{1}{n} \Sum \khi \rho_\eta(Y_i'(b)),
	\]
	where $\rho_\eta(v)=v[\eta-\one{v\leq0}]$. 
	It holds that 
	$\partial^+_v\rho_\eta(v)  = \eta - \one{v < 0}$ and $\partial^-_v\rho_\eta(v)  = \eta - \one{v \leq 0}$.
	Therefore, also the left and right derivatives of the criterion function exist. For $j \in \{0,1\}$, 
	\begin{align*}
		\partial^+_{b_j}G_n(b) &  =  \frac{h^j}{n}  \Sum \khi  \Xihj   \left( ( \one{ Y_i'(b) < 0 } - \eta )\one{\Xihj < 0} +  ( \one{ Y_i'(b) \leq 0 } - \eta )\one{ 0 < \Xihj }\right), \\
		\partial^-_{b_j}G_n(b) & =  \frac{h^j}{n}  \Sum \khi  \Xihj \left( (\one{ Y_i'(b) \leq 0 } - \eta )\one{\Xihj < 0} +  ( \one{ Y_i'(b) < 0 } - \eta )\one{0 < \Xihj } \right).
	\end{align*}	
	At the minimum, it holds that $\partial^-_{b_j} G_n(\widehat{q}(\eta) ) \leq 0 \leq\partial^+_{b_j} G_n(\widehat{q}(\eta)  ) $. Using these inequalities, I~obtain the following bounds on the expression of interest:
	\begin{align*}
		0 & \leq \frac{h^j}{n} \Sum \khi   \Xihj \left( \one{Y_i \leq \widehat{Q}(\eta,X_i;h)} - \eta -  \one{Y_i= \widehat{Q}(\eta,X_i;h)}\one{\Xihj < 0} \right) \\
		& \leq  \partial^+_{b_j} G_n(\widehat{q}(\eta)) -\partial^-_{b_j} G_n(\widehat{q}(\eta) )\\
		& =  \frac{h^j}{n} \Sum \khi  \Xihj \left(  \one{Y_i = \widehat{Q}(\eta,X_i;h) } \one{ 0 \leq \Xihj } - \one{Y_i = \widehat{Q}(\eta,X_i;h)}  \one{\Xihj < 0} \right).
	\end{align*}
	The lemma follows from the facts that $k$ is bounded with bounded support and 
	\[
	\Sum \one{Y_i = \widehat{Q}(\eta,X_i;h) } \leq 2 \text{ w.p. 1}
	\]
	because the probability of having three collinear points in a sample is equal zero.
\end{proof}

\begin{lemma}
Suppose that Assumptions \ref{ass:ass1}, \ref{ass:quantile}, and \ref{ass:kernel} hold.
If additionally $x_0$ is an interior point, $f_X(x)$ is continuously differentiable, and $j$ is odd, then
$S_{n,j} = O_p(h+(nh)^{-1/2})$.
\end{lemma}

\begin{proof}
	Standard kernel calculations.
\end{proof}

\begin{proof}[Proof of Proposition \ref{prop:NM}]
Since $a=h$, $w_n=h^2+(nh)^{-1/2}\equiv r_n$.	It holds that
	\[
	\widehat{m}^{\sss NM}(\eta,x_0;h,h) -\widehat{m}(\eta,x_0;h,h) = \frac{S_{n,2}(T_{n,0} - \Psi_{n,0}(h)) - S_{n,1}(T_{n,1}- \Psi_{n,1}(h)) }{S_{n,2} S_{n,0} - S_{n,1}^2 }
	\]
	where  $T_{n,j}=\frac{1}{n} \Sum \khi \Xihj \frac{1}{\eta} Y_i \one{Y_i \leq \widehat{Q}(\eta,X_i;h)}$, and $\Psi_{n,j}(h)$ is defined in the proof of Lemma~\ref{lemma:asy_equivalence}.
	From Lemma \ref{lemma:AFOC} it immediately follows that
	\begin{align*}
		T_{n,0} - \Psi_{n,0}(h) & = O_p((nh)^{-1}), \\
		T_{n,1} - \Psi_{n,1}(h) & = \frac{1}{\eta} \widehat{q}_1(\eta;h) \frac{h}{n} \Sum \khi X_{h,i}^2 ( \one{Y_i \leq \widehat{Q}(\eta,X_i;h)} - \eta ) + O_p((nh)^{-1}) \\
		& =  \frac{1}{\eta} q_1(\eta) \frac{h}{n} \Sum \khi X_{h,i}^2 ( \one{Y_i \leq \widehat{Q}(\eta,X_i;h)} - \eta ) + O_p(r_n^2).
	\end{align*}
	Hence, 
	\[
	\widehat{m}^{\sss NM}(\eta,x_0;h,h) -\widehat{m}(\eta,x_0;h,h) = hS_{n,1}q_1(\eta)O_p(r_n) + O_p(r_n^2),
	\]
	which, combined with Lemma \ref{lemma:basics}, concludes the proof.
\end{proof}

\begin{proof}[Proof of Proposition \ref{prop:TS}]
	It holds that
	\[
	\widehat{m}^{\sss TS}(\eta,x_0;h,h) = \frac{\widehat{S}_{n,2} T_{n,0} - \widehat{S}_{n,1}T_{n,1}}{\widehat{S}_{n,2} \widehat{S}_{n,0} - \widehat{S}_{n,1}^2 },
	\]
	where $\widehat{S}_{n,j}=\frac{1}{\eta n} \Sum \khi \Xihj \one{Y_i \leq \widehat{Q}(\eta,X_i;h)}$, and $T_{n,j}$ is defined in the proof of Proposition~\ref{prop:NM}. It holds that
	\[
	\widehat{S}_{n,2} \widehat{S}_{n,0} - \widehat{S}_{n,1}^2 = S_{n,2} S_{n,0} - S_{n,1}^2 + O_p( r_n ).
	\]
	Let $m^*(\eta,x)= m(\eta,x_0) + \partial_x^1m(\eta,x_0)(x-x_0)$. By plugging in the expression
	$	Y_i = m^*(\eta,X_i)  + (Y_i-m^*(\eta,X_i))$ in the definition of $\widehat{m}^{\sss TS}(\eta,x_0;h,h)$, it follows that
	\[
	\widehat{m}^{\sss TS}(\eta,x_0;h,h) = m(\eta,x_0) + \frac{\widehat{S}_{n,2} U_{n,0} - \widehat{S}_{n,1}U_{n,1} }{\widehat{S}_{n,2} \widehat{S}_{n,0} - \widehat{S}_{n,1}^2 },
	\]
	where $U_{n,j}=\frac{1}{n} \Sum \khi \Xihj \frac{1}{\eta} (Y_i-m^*(\eta,X_i)) \one{Y_i \leq \widehat{Q}(\eta,X_i;h)}$.
	
	Lemma \ref{lemma:precision1} yields that for $j \in \{0,1\}$
	\begin{align*}
		U_{n,j}	& =\frac{1}{n} \Sum \khi \Xihj \frac{1}{\eta} (Y_i-Q(\eta,X_i)) \one{Y_i \leq Q(\eta,X_i)}, \\
		& + \frac{1}{n} \Sum \khi \Xihj \frac{1}{\eta} (Q(\eta,X_i) - m^*(\eta,X_i) ) \one{Y_i \leq \widehat{Q}(\eta,X_i;h)} + O_p(r_n^2).
	\end{align*}
	
	Moreover, by Lemmas~\ref{lemma:AFOC} and~\ref{lemma:precision1}, it holds that
	\begin{align*}
		&\frac{1}{n} \Sum \khi  \frac{1}{\eta} (Q(\eta,X_i) - m^*(\eta,X_i) ) \{ \one{Y_i \leq \widehat{Q}(\eta,X_i;h)} - \eta \} = O_p(r_n^2),\\
		&\frac{1}{n} \Sum \khi  X_{h,i} \frac{1}{\eta} (Q(\eta,X_i) - m^*(\eta,X_i) ) \{ \one{Y_i \leq \widehat{Q}(\eta,X_i;h)} - \eta \} \\
		&\quad = \frac{1}{\eta}  h(\partial_{x}^1 Q(\eta,x_0) - \partial_{x}^1 m(\eta,x_0)) \frac{1}{n} \Sum \khi X^2_{h,i} \{ \one{Y_i \leq \widehat{Q}(\eta,X_i;h)} - \eta \} + O_p(r_n^2)\\
		&\quad = h(\partial_{x}^1 Q(\eta,x_0) - \partial_{x}^1 m(\eta,x_0)) O_p(r_n) + O_p(r_n^2).
	\end{align*}
	
	Hence,
	\begin{align*}
		U_{n,0}	& =\frac{1}{n} \Sum \khi \frac{1}{\eta} (Y_i-Q(\eta,X_i)) \one{Y_i \leq Q(\eta,X_i)} +Q(\eta,X_i) - m^*(\eta,X_i)+O_p(r_n^2),\\
		U_{n,1}	& =\frac{1}{n} \Sum k_{h,i} X_{h,i} \frac{1}{\eta} (Y_i-Q(\eta,X_i)) \one{Y_i \leq Q(\eta,X_i)} +Q(\eta,X_i) - m^*(\eta,X_i) \\  & + h(\partial_{x}^1 Q(\eta,x_0) - \partial_{x}^1 m(\eta,x_0)) O_p(r_n) +O_p(r_n^2).
	\end{align*}
	In particular, $U_{n,j}=O_p(r_n)$, and hence
	\begin{align*}
		\widehat{m}^{\sss TS}(\eta,x_0;h,h) & = m(\eta,x_0) + \frac{S_{n,2} U_{n,0} - S_{n,1}U_{n,1} }{S_{n,2} S_{n,0} - S_{n,1}^2 } \\
		& = \widetilde{m}(\eta,x_0;h) + hS_{n,1}(\partial_{x}^1 Q(\eta,x_0) - \partial_{x}^1 m(\eta,x_0)) O_p(r_n) + O_p(r_n^2),
	\end{align*}
	which, combined with Lemma \ref{lemma:basics}, concludes the proof.
\end{proof}

\subsection{Proofs of Propositions \ref{prop:A:NM} and \ref{prop:A:TS}}
To prove these propositions, an explicit expansion of the estimators in the coefficients defining the trimming function is required.

\begin{lemma}
	Suppose that Assumptions~\ref{ass:ass1}, \ref{ass:quantile}, and \ref{ass:kernel} hold, and $\partial^2_xQ(\eta,x)$ is continuous in $x$ on $\mc X$. Then
	\begin{align*}
		\widehat{q}_0(\eta;a) - q_0(\eta) & = \frac{1}{2} \mu\,\partial_x^2Q(\eta,x_0) a^2 + \frac{\frac{1}{n}\Sum k_{a,i}( \mu_2 -  \mu_1  X_{a,i}) [\eta - \one{Y_i \leq Q(\eta,X_i)}]}{f_{Y|X}( q_0(\eta) |x_0)  f_X(x_0)(\mu_2\mu_0-\mu_1^2)} +  o_p(r_n).
	\end{align*}
\end{lemma}
\begin{proof}
	This representation follows from the proof of Theorem~2 of \citet{fan1994robust}.
\end{proof}

\begin{lemma}\label{lemma:precision3}  Suppose that Assumptions~\ref{ass:ass1}, \ref{ass:quantile}, and \ref{ass:kernel} hold. Then for $j \in \mathbb{N}$ it holds that
	\begin{align*} 
		 \frac{1}{n} \Sum & \khi \Xihj  \one{Y_i \leq \widehat{Q}(\eta,X_i;a) } = \frac{1}{n}	\Sum \khi \Xihj \one{Y_i \leq Q^*(\eta,X_i)} \\
		& + \frac{1}{n} \sum_{i=1}^n \khi \Xihj f_{Y|X}(Q(\eta,x_0)|x_0) \left( \widehat{q}_0(\eta;a)-q_0(\eta) + (\widehat{q}_1(\eta;a)-q_1(\eta))(X_i-x_0) \right) + o_p(r_n).
	\end{align*}
\end{lemma}

\begin{proof} Following the steps of the proof of Lemma \ref{lemma:equi2}, one can show that
	\begin{align*}
		 \frac{1}{n} \Sum \khi \Xihj & \left( \one{Y_i \leq \widehat{Q}(\eta,X_i)  }- \E[ \one{Y \leq L_i(b)}|X=X_i] \big|_{ b=\widehat{q}(\eta) } \right) \\
		& =  \frac{1}{n} \Sum \khi \Xihj  \left( \one{Y_i \leq Q^*(\eta,X_i) }- \E[\one{Y \leq Q^*(\eta,X)}|X=X_i] \right) + o_p(r_n).
	\end{align*}
	The result follows by a Taylor expansion using continuity of $f_{Y|X}(y|x)$.
\end{proof}

\begin{proof}[Proof of Proposition \ref{prop:A:NM}]
	\textit{Part (i).} The result follows by standard arguments, using the fact that
	\begin{align*}
		& \E\left[ \left(\frac{1}{\eta} Y \one{Y \leq Q(\eta,X)}-m(\eta,X)\right)^2\Big|X=x_0 \right] \\
		& \quad = \E\left[ \left(\frac{1}{\eta} (Y-m(\eta,X)) \one{Y \leq Q(\eta,X)}- \frac{1}{\eta} m(\eta,X)(\eta - \one{Y \leq Q(\eta,X)} ) \right)^2\Big|X=x_0 \right] \\
		& \quad =\frac{1}{\eta}\Var[Y|Y\leq Q(\eta,X), X=x_0] + \frac{(1-\eta)}{\eta}m(\eta,x_0)^2.
	\end{align*}
	
	\noindent\textit{Part (ii).} It holds that
	\[
	\widehat{m}^{\sss NM}(\eta,x_0;a,h) = \frac{S_{n,2}T_{n,0}(a) - S_{n,1}T_{n,1}(a) }{S_{n,2} S_{n,0} - S_{n,1}^2 }
	\]
	where  $T_{n,j}(a)=\frac{1}{n} \Sum \khi \Xihj \frac{1}{\eta} Y_i \one{Y_i \leq \widehat{Q}(\eta,X_i;a)}$.
	I consider the numerator. First, by Lemma~\ref{lemma:precision1}:
	\begin{align*}
		T_{n,j}(a) & = \frac{1}{n} \Sum \khi \Xihj \frac{1}{\eta} (Y_i-Q^*(\eta,X_i)) \one{Y_i \leq Q^*(\eta,X_i)}   \\
		& \quad + \frac{1}{n} \Sum \khi \Xihj  \frac{1}{\eta} Q^*(\eta,X_i)\one{Y_i \leq \widehat{Q}(\eta,X_i;a)} + O_p(r_n^2)\\
	\end{align*}
	Further, using Lemma~\ref{lemma:precision3}, 
	\begin{align*}
		 T_{n,j}(a) & = \frac{1}{n} \Sum \khi \Xihj  \frac{1}{\eta} Y_i \one{Y_i \leq Q^*(\eta,X_i)}\\
		& \quad + \frac{1}{n} \sum_{i=1}^n \khi \Xihj f_{Y|X}( q_0(\eta)|x_0) \frac{1}{\eta} q_0(\eta)  \left( \widehat{q}_0(\eta;a)-q_0(\eta) + (\widehat{q}_1(\eta;a)-q_1(\eta))(X_i-x_0)  \right)  \\
		&\quad  + \frac{1}{n} \sum_{i=1}^n \khi \Xihj f_{Y|X}( q_0(\eta)|x_0) \frac{1}{\eta} q_1(\eta)(X_i-x_0) \\
		&\quad \quad \times\left( \widehat{q}_0(\eta;a)-q_0(\eta) + (\widehat{q}_1(\eta;a)-q_1(\eta))(X_i-x_0)  \right) + o_p(r_n).
	\end{align*}
	The last term is of order $O_p(r_nh)$.
	Let $u_i^*(\eta)= \frac{1}{\eta} (Y_i-m^*(\eta,X_i)) \one{Y_i \leq Q^*(\eta,X_i)} $, $e_i^*(\eta)= \frac{1}{\eta}( \eta - \one{Y_i \leq Q^*(\eta,X_i)})$, and
	\begin{align*}
		E_{n,j}(a,h) & =\frac{1}{n} \Sum \khi \Xihj  (u_i^*(\eta) -m^*(\eta,x_0) e_i^*(\eta))  + \frac{1}{n} \Sum k_{a,i} X_{a,i}^j  q_0(\eta) e_i^*(\eta).
	\end{align*}
	It follows that
	\begin{align*}
		\widehat{m}^{\sss NM}(\eta,x_0;a,h) = m(\eta,x_0) + \frac{ \mu_2 E_{n,0}(a,h) - \mu_1 E_{n,1}(a,h) }{ (\mu_2 \mu_0 - \mu_1^2)f_X(x_0) } + o_p(r_n).
	\end{align*}
	Asymptotic normality follows from standard results. The bias expression follows from:
	\begin{align*}
		&	\frac{d^2}{dx^2}\E[ u^{*}(\eta)|X=x]|_{x=x_0}= \partial_{x}^2 m(\eta,x_0) - \frac{1}{\eta} f_{Y|X}(q_0(\eta)|x_0)(q_0(\eta)-m(\eta,x_0) )\partial_x^2Q(\eta,x_0),\\ 
		& \frac{d^2}{dx^2}\E[e^*(\eta)|X=x]|_{x=x_0}= \frac{1}{\eta} f_{Y|X}(q_0(\eta)|x_0) \partial_x^2Q(\eta,x_0).
	\end{align*}
	The variance is calculated as follows. Recall that $h/a \to \rho$. It holds that
	\begin{align*}
		\Var[u^*(\eta)|X=x_0]  = \frac{1}{\eta} \Var[Y|Y\leq Q(\eta,X), X=x_0] \text{ and } \Var[e^*(\eta)|X=x_0]=\frac{1-\eta}{\eta}, 
	\end{align*}
	Furthermore, $u^*(\eta)$ and $e^*(\eta)$ are uncorrelated conditional on $X=x_0$ and
	\begin{align*}
		\Var&\bigg[ \frac{1}{n} \Sum \left( k_a(X-x_0) (\mu_2 - \mu_1 X_a) Q(\eta,x_0) - k_h(X-x_0) (\mu_2 - \mu_1 X_h) m(\eta,x_0) \right) e^*(\eta) |X_1,...,X_n \bigg]\\
		& = \frac{1}{nh} \int_{\mc D} \left(  \rho k(v\rho)(\mu_2 - \mu_1 v \rho)  Q(\eta,x_0) -k(v)(\mu_2-\mu_1 v) m(\eta,x_0)  \right)^2 dv \frac{1-\eta}{\eta} f_X(x_0)(1+o_p(1)),
	\end{align*}
	which concludes the proof.
\end{proof}

\begin{proof}[Proof of Proposition \ref{prop:A:TS}]
	\textit{Part (i).} It holds that
	\[
	\widetilde{m}^{\sss TS}(\eta,x_0;h) = m(\eta,x_0) + \frac{\widetilde{S}_{n,2} \widetilde{U}_{n,0} - \widetilde{S}_{n,1} \widetilde{U}_{n,1} }{\widetilde{S}_{n,2} \widetilde{S}_{n,0} - \widetilde{S}_{n,1}^2 },
	\]
	where 
	\begin{align*}
		& \widetilde{U}_{n,j} = \frac{1}{n} \Sum \khi \Xihj \frac{1}{\eta} (Y_i-m^*(\eta,X_i)) \one{Y_i \leq Q(\eta,X_i)},\\  
		& \widetilde{S}_{n,j} = \frac{1}{\eta n} \Sum \khi \Xihj \one{Y_i \leq Q(\eta,X_i)}= \mu_j f_X(x_0) + o_p(1).
	\end{align*}
	The result follows from standard arguments.
	
	\noindent\textit{Part (ii).} It holds that
	\[
	\widehat{m}^{\sss TS}(\eta,x_0;a,h) = m(\eta,x_0) + \frac{\widehat{S}_{n,2} U_{n,0}(a,h) - \widehat{S}_{n,1}U_{n,1}(a,h) }{\widehat{S}_{n,2} \widehat{S}_{n,0} - \widehat{S}_{n,1}^2 },
	\]
	where 
	\begin{align*}
		& U_{n,j}(a,h)=\frac{1}{n} \Sum \khi \Xihj \frac{1}{\eta} (Y_i-m^*(\eta,X_i)) \one{Y_i \leq \widehat{Q}(\eta,X_i;a)},\\ 
		& \widehat{S}_{n,j}(a)=\frac{1}{\eta n} \Sum \khi \Xihj \one{Y_i \leq \widehat{Q}(\eta,X_i;a)} = \mu_j f_X(x_0) + o_p(1).
	\end{align*}
	Lemma \ref{lemma:precision1} yields that
	\begin{align*}
		U_{n,j}(a,h)	& =\frac{1}{n} \Sum \khi \Xihj \frac{1}{\eta} (Y_i-Q^*(\eta,X_i)) \one{Y_i \leq Q^*(\eta,X_i)} \\
		& \quad+ \frac{1}{n} \Sum \khi \Xihj \frac{1}{\eta} (Q^*(\eta,X_i) - m^*(\eta,X_i) ) \one{Y_i \leq \widehat{Q}(\eta,X_i;h)} + O_p(r_n^2).
	\end{align*}
	Further, using Lemma~\ref{lemma:precision3}, 
	\begin{align*}
		U_{n,j}(a,h)	& =\frac{1}{n} \Sum \khi \Xihj \frac{1}{\eta} (Y_i-Q^*(\eta,X_i)) \one{Y_i \leq Q^*(\eta,X_i)} \\
		& \quad + \frac{1}{n} \Sum \khi \Xihj \frac{1}{\eta} (Q^*(\eta,X_i) - m^*(\eta,X_i) ) \one{Y_i \leq Q^*(\eta,X_i)} \\
		& \quad + \frac{1}{n} \sum_{i=1}^n \khi \Xihj \frac{1}{\eta} (Q^*(\eta,X_i) - m^*(\eta,X_i) ) f_{Y|X}(Q(\eta,x_0)|x_0) \\
		& \quad \times \left( \widehat{q}_0(\eta;a)-q_0(\eta) + (\widehat{q}_1(\eta;a)-q_1(\eta))(X_i-x_0)  \right) + o_p(r_n).	
	\end{align*}
	It follows that 
	\begin{align*}
		\widehat{m}^{\sss TS}(\eta,x_0;h,h) & = m(\eta,x_0) + \frac{\widehat{S}_{n,2} U^*_{n,0}(a,h) - \widehat{S}_{n,1}U^*_{n,1}(a,h) }{\widehat{S}_{n,2} \widehat{S}_{n,0} - \widehat{S}_{n,1}^2 } \\
		& + \frac{1}{\eta} (Q(\eta,x_0) - m(\eta,x_0) ) f_{Y|X}(Q(\eta,x_0)|x_0) (\widehat{q}_0(\eta;a)-q_0(\eta)),
	\end{align*}
	where $U^*_{n,j}(h)=\frac{1}{n} \Sum \khi \Xihj u^{*}_i(\eta)$
	with $u^{*}_i(\eta)= \frac{1}{\eta} (Y_i-m^*(\eta,X_i)) \one{Y_i \leq Q^*(\eta,X_i)}$.
	The variance and bias expressions follow from the calculations in the proof of Proposition~\ref{prop:A:NM}.
\end{proof}

\section{Additional Simulations}\label{A:ROT}
\citet{armstrong2020simple} propose a rule of thumb to calibrate the bound on the second derivative of the conditional expectation function. They run a quartic, global regression, and estimate the maximal second derivative based on it.
I adapt this approach to calibrate the bound on $\partial^2_x m(\eta,x)$. In the first stage, I obtain  $\wh{Q}^{glob}(\eta,X_i)$ in a quartic quantile regression. In the second stage, I run a global quartic regression with $\psi_i(\eta, \wh{Q}^{glob}(\eta,X_i))$ as the outcome variable.

I investigate the performance of this procedure in the setting from Section \ref{sec:MonteCarlo}. The results are presented in Table \ref{table:sim2_A}. In this example, the rule of thumb leads to CIs with good coverage properties. This is consistent with the findings of \citet{armstrong2020simple}.

\begin{table}[hbt]
	\centering
	\caption{Coverage, average bandwidth, and average length of the 95\% CI.}
	\scalebox{0.95}{\begin{threeparttable}[hbt]
			\begin{tabular}{llccccccccc}
				\hline
				& & \multicolumn{3}{c}{Coverage} & \multicolumn{3}{c}{Bandwidth} & \multicolumn{3}{c}{CI length} \\
				& Design for $m_j$:	& 1 & 2 & 3 & 1 & 2 & 3 & 1 & 2 & 3 \\ 
				\hline
				
				\multicolumn{10}{l}{\textit{Homoskedastic errors}} \\
				
				\multirow{2}{*}{$\eta=0.2$} & Infeasible $\widetilde{m}$ & 93.6 & 92.1 & 95.4 & 0.231 & 0.310 & 0.257 & 0.128 & 0.113 & 0.120 \\ 
				& Feasible $\wh{m}$ & 93.4 & 92.2 & 95.7 & 0.227 & 0.307 & 0.260 & 0.128 & 0.113 & 0.119\vspace{4pt} \\ 
				\multirow{2}{*}{$\eta=0.5$} & Infeasible $\widetilde{m}$  & 95.0 & 93.1 & 96.0 & 0.207 & 0.279 & 0.231 & 0.104 & 0.091 & 0.098 \\ 
				& Feasible $\wh{m}$  & 94.9 & 93.3 & 96.1 & 0.204 & 0.277 & 0.233 & 0.104 & 0.092 & 0.098\vspace{4pt} \\ 
				\multirow{2}{*}{$\eta=0.8$} & Infeasible $\widetilde{m}$ & 95.7 & 94.0 & 96.2 & 0.197 & 0.266 & 0.222 & 0.095 & 0.083 & 0.089 \\ 
				& Feasible $\wh{m}$  & 95.7 & 94.0 & 96.4 & 0.196 & 0.265 & 0.222 & 0.095 & 0.084 & 0.089 \\ \hline

				\multicolumn{10}{l}{\textit{Heteroskedastic errors}} \\
				
				\multirow{2}{*}{$\eta=0.2$} & 		Infeasible $\widetilde{m}$ & 93.4 & 92.6 & 95.6 & 0.239 & 0.310 & 0.250 & 0.129 & 0.115 & 0.123 \\ 
				&		Feasible $\wh{m}$ & 93.5 & 92.9 & 95.8 & 0.235 & 0.307 & 0.254 & 0.129 & 0.116 & 0.122\vspace{4pt}\\ 
				\multirow{2}{*}{$\eta=0.5$} &	Infeasible $\widetilde{m}$  & 95.0 & 93.6 & 96.5 & 0.213 & 0.277 & 0.225 & 0.104 & 0.093 & 0.100 \\ 
				&		Feasible $\wh{m}$  & 95.1 & 93.7 & 96.5 & 0.210 & 0.276 & 0.227 & 0.105 & 0.094 & 0.100\vspace{4pt} \\ 
				\multirow{2}{*}{$\eta=0.8$} & Infeasible	$\widetilde{m}$  & 95.7 & 94.3 & 96.6 & 0.202 & 0.264 & 0.215 & 0.095 & 0.085 & 0.091 \\ 
				&		Feasible $\wh{m}$  & 95.7 & 94.3 & 96.7 & 0.201 & 0.263 & 0.216 & 0.096 & 0.085 & 0.092 \\ \hline
			\end{tabular}
			\begin{tablenotes}\small
				\item\hspace{-3pt}\textit{Notes:} Estimators evaluated with their respective RMSE-optimal bandwidths.  The sample size is $n = 1,000$, and the number of simulations is $S = 10,000$. The smoothness constant is selected using the rule of thumb discussed in Online Appendix~\ref{A:ROT}.
			\end{tablenotes}
	\end{threeparttable}}
	\label{table:sim2_A}
\end{table}

\end{document}